\newif\ifllncs\llncsfalse
\newif\ifanon\anonfalse

\ifllncs
\documentclass[runningheads,envcountsame]{llncs}

\else
\documentclass[11pt,letterpaper]{article}
\fi

\newcommand{\refsec}[1]{\hyperref[#1]{Section~\ref*{#1}}}
\newcommand{\reftab}[1]{\hyperref[#1]{Table~\ref*{#1}}}
\newcommand{\reffig}[1]{\hyperref[#1]{Figure~\ref*{#1}}}
\newcommand{\reflem}[1]{\hyperref[#1]{Lemma~\ref*{#1}}}
\newcommand{\refthm}[1]{\hyperref[#1]{Theorem~\ref*{#1}}}
\newcommand{\refprop}[1]{\hyperref[#1]{Proposition~\ref*{#1}}}
\newcommand{\refcor}[1]{\hyperref[#1]{Corollary~\ref*{#1}}}
\newcommand{\refalg}[1]{\hyperref[#1]{Algorithm~\ref*{#1}}}
\newcommand{\refapp}[1]{\hyperref[#1]{Appendix~\ref*{#1}}}
\newcommand{\refeqs}[1]{\hyperref[#1]{Equations~(\ref*{#1})}}
\newcommand{\refeqplain}[1]{\hyperref[#1]{Equation~\ref*{#1}}}
\newcommand{\refeqsplain}[1]{\hyperref[#1]{Equations~\ref*{#1}}}
\newcommand{\refeqlowerplain}[1]{\hyperref[#1]{equation~\ref*{#1}}}
\newcommand{\refconj}[1]{\hyperref[#1]{Conjecture~\ref*{#1}}}

\newif\ifnotes
\notesfalse

\title{Streaming Regev's Quantum Factoring Algorithm}

\ifanon
    \author{}
    \institute{}
    \date{}
\else
    \author{
    Katherine Van Kirk\thanks{Email: \texttt{kvankirk@g.harvard.edu}. This work was done while the author was working at Google Quantum AI in Santa Barbara. In her academic capacity, she is supported by the Fannie and John Hertz Foundation, an NDSEG fellowship, and the Harvard Quantum Initiative. }\\Google Quantum AI, Harvard
    \and
    Craig Gidney\thanks{Email: \texttt{craiggidney@google.com}.}\\Google Quantum AI
    }
    \date{}
\fi

\ifllncs\else
\date{\today}
\fi

\usepackage[T1]{fontenc}
\usepackage{amssymb}
\usepackage{amsmath}
\usepackage{amsthm}
\usepackage{braket}
\PassOptionsToPackage{hyphens}{url}
\usepackage[ruled, algosection]{algorithm2e}
\usepackage{ulem}
\usepackage{multirow}
\usepackage{float}

\ifllncs
\usepackage[pdfstartview=FitH,colorlinks,urlcolor=black,linkcolor=black,
    citecolor=black,pdfpagelabels,bookmarksdepth=2,bookmarksopen=true,bookmarksnumbered,breaklinks]{hyperref}
\usepackage[font={small}]{caption}
\usepackage[nottoc,notlot,notlof]{tocbibind}

\else
\usepackage[utf8]{inputenc}
\usepackage{fullpage}
\usepackage[margin=1in]{geometry}
\usepackage{caption}
\usepackage{xcolor}
\definecolor{DarkBlue}{RGB}{0,0,150}
\definecolor{NotSoDarkBlue}{RGB}{15,15,210}
\definecolor{DarkRed}{RGB}{150,0,0}
\definecolor{DarkGreen}{RGB}{0,100,0}
\usepackage[pdfstartview=FitH,pdfpagemode=UseNone,colorlinks,linkcolor=DarkRed,filecolor=blue,citecolor=DarkRed,urlcolor=DarkRed,pagebackref,breaklinks]{hyperref}
\usepackage{microtype}      
\usepackage{libertine}
\usepackage{libertinust1math}

\usepackage{breakcites}
\usepackage[noadjust,sort,nocompress]{cite}
\usepackage{graphicx}

\usepackage{tikz,pgfplots}
\pgfplotsset{compat=1.18}
\usetikzlibrary{calc}

\usepackage{mathtools, xparse}
\usepackage{graphicx}
\usepackage{rotating}
\usepackage{url}
\usepackage{bm}
\usepackage{soul}

\usepackage{adjustbox}

\fi

\usepackage{orcidlink}
\usepackage{booktabs}
\usepackage[table,dvipsnames]{xcolor}

\usepackage{bbm}
\usepackage{mathtools}

\usepackage{booktabs,tabularx,makecell,adjustbox,multirow,array}
\newcolumntype{Y}{>{\centering\arraybackslash}X} 

\usepackage{framed}
\usepackage{float} 
\usepackage{multicol}
\usepackage[capitalise]{cleveref}
\usepackage[inline]{enumitem}

\ifllncs
\spnewtheorem{informal-theorem}[theorem]{Informal Theorem}{\bfseries}{\itshape}
\spnewtheorem{experiment}{Experiment}[section]{\bfseries}{}
\spnewtheorem{fact}[theorem]{Fact}{\bfseries}{\itshape}
\spnewtheorem{defn}[theorem]{Definition}{\bfseries}{}
\spnewtheorem{const}[theorem]{Construction}{\bfseries}{}
\spnewtheorem{assumption}{Assumption}{\bfseries}{}
\else 

\usepackage{amsthm}
\theoremstyle{plain}
\newtheorem{theorem}{Theorem}[section]{\bfseries}{\itshape}
\newtheorem{informal-theorem}[theorem]{Informal Theorem}{\bfseries}{\itshape}
{\bfseries}{}
\newtheorem{lemma}[theorem]{Lemma}{\bfseries}{\itshape}
\newtheorem{proposition}[theorem]{Proposition}{\bfseries}{\itshape}
{\bfseries}{\itshape}
\newtheorem{conjecture}{Conjecture}

\theoremstyle{remark}
\newtheorem{remark}[theorem]{Remark}{\itshape}{}
\theoremstyle{definition}
{\bfseries}{}
\newtheorem{definition}[theorem]{Definition}{\bfseries}{}
{\bfseries}{}
\newtheorem{corollary}[theorem]{Corollary}{\bfseries}{\itshape}
{\bfseries}{\itshape}
{\bfseries}{\itshape}

\numberwithin{theorem}{section}
\numberwithin{conjecture}{section}
\numberwithin{problem}{section}
\numberwithin{assumption}{section}
\fi

\ifllncs
\else

\fi

\newcommand{\RR}{\mathbb{R}}

\newcommand{\ZZ}{\mathbb{Z}}

\newcommand{\poly}{\mathsf{poly}}

\DeclareMathOperator{\argmax}{argmax}

\DeclarePairedDelimiter{\floor}{\lfloor}{\rfloor}
\DeclarePairedDelimiter{\ceil}{\lceil}{\rceil}

\allowdisplaybreaks

\DeclarePairedDelimiterX{\bigket}[1]{\bigg\lvert}{\bigg\rangle}{\,#1}
\DeclarePairedDelimiter{\norm}\lVert\rVert

\renewcommand{\refeq}[1]{\hyperref[#1]{Equation~(\ref*{#1})}}

\begin{document}
\maketitle

\thispagestyle{plain}
\frenchspacing

\begin{abstract}

    A major practical obstacle to Regev's factoring algorithm is its large space usage.
    In particular, the algorithm was believed to be incompatible with "qubit recycling", 
    which streams the control qubits in phase estimation rather than storing them all simultaneously.
    We show how to fix this streaming problem.
    We build on the construction of Ragavan and Vaikuntanathan \cite{rv24}, which used Fibonacci accumulators to reduce
    the qubit count from $O(n^{1.5})$ to $11.32n$.
    We make two key improvements.
    First, we use superposition masking to prevent unwanted information from entering the accumulators,
    removing the need for an uncomputation.
    Second, we show how to perform phase estimation with Fibonacci powers of a unitary operation, rather than power-of-2 powers. This avoids various encoding conversions required of Ref~\cite{rv24}.
    After these improvements, the phase estimation qubits are used only once in the circuit, allowing them to be streamed.
    
    Our improvements reduce the space usage of Regev's factoring algorithm to $n + M(n) + O(\sqrt{n})$, where $M(n)$ is the space of an in-place multiplication.
    For example, using the in-place multiplier from \cite{LuoEtAl26}, Regev's algorithm can factor an $n$-bit integer using $4n + O(\sqrt n)$ qubits.
    Although our work still falls short of the $n/2 + \epsilon n$ qubits sufficient for the leading implementations of Shor's algorithm \cite{DBLP:journals/iacr/ChevignardFS24,gidney2025factor2048bitrsa}, it is the lowest-space instantiation of Regev's algorithm to date, and it suggests Regev’s algorithm is still far from fully optimized. 
\end{abstract}

\ifllncs
\else
\thispagestyle{empty}
\newpage
\thispagestyle{empty}
\tableofcontents
\newpage
\pagenumbering{arabic}
\fi

\section{Introduction}

Since Shor introduced a polynomial-time quantum algorithm for factoring
$n$-bit integers~\cite{shor97}, substantial effort has gone into creating new implementations \cite{beckman,vedral,seifert2001using,Cop02,CleveW00,beauregard,takahashi,zalka2006shors,DBLP:conf/pqcrypto/EkeraH17,gidney2017factoring,hrs17,gidney2019windowed,DBLP:journals/quantum/GidneyE21,DBLP:journals/iacr/ChevignardFS24,kahanamokumeyer2024fast,gidney2025factor2048bitrsa}.
A particularly notable advance came from 
Regev in 2023, when he
promoted Shor's one-dimensional period-finding problem to its
high-dimensional analog~\cite{Regev23}. 
His factoring algorithm differed from Shor's in various interesting ways. It uses $O(\sqrt{n})$ generators instead of one generator. Its gate cost is dominated by unconditionally squaring a superposed accumulator, rather than by conditionally multiplying a classical value into a superposed accumulator. And finally, it distributes the computation over $O(\sqrt{n})$ shots,
reducing the number of quantum
multiplications per run from $O(n)$ to $O(\sqrt n)$.
As such, a common misconception is that Regev's factoring algorithm uses asymptotically fewer gates than Shor's. His per-run improvement often gets conflated with total gate count, where the benefit of fewer operations per run is canceled out by requiring more shots.  

While the differences in 
Regev's formulation make it a source of potential future improvements to quantum factoring, a major practical issue is its large qubit count.
In Regev's original paper, he used $\Omega(n^{1.5})$ qubits. 
This gap in the concrete resource efficiency between Shor and Regev spurred a sequence of works in the following years~\cite{cryptoeprint:2024/636,rv24,ekeragartnercomparison,KRVK26}.
In particular, Ragavan and Vaikuntanathan were able to reduce the space cost to $O(n)$, but their $11.32n$ qubit count~\cite{rv24} is still far from the $n/2 + \epsilon n$ achieved by certain variants of Shor's algorithm~\cite{DBLP:journals/iacr/ChevignardFS24,gidney2025factor2048bitrsa}.
In this paper, our goal is to further improve the qubit count of Regev's algorithm.
This motivates
the central question of our work: 
\begin{quote}
    \centering
    \textit{
Can we improve the efficiency of Regev's algorithm,  further narrowing the gap \\ between Shor and Regev for  quantumly factoring 2048-bit integers?}
\end{quote}

A major obstacle in reducing Regev's high qubit count has been its presumed incompatibility with ``qubit recycling''~\cite{Parker2000,Mosca1999}.
Recasting factoring as phase estimation, one can dramatically reduce the space required, \textit{if} the control qubits can be used, measured, and recycled. If they can be recycled, they can take turns existing and do not have to all be stored simultaneously. 
Unfortunately, previous implementations of Regev's algorithm always perform multiple sweeps over the qubits in the phase estimation register, and so these qubits could not be recycled.
Concretely, in the low-space implementation of Ragavan and Vaikuntanathan, these multiple sweeps come from extra uncomputations and extra conversions \cite{rv24}.

We present a streamed implementation of Regev's algorithm. We overcome these two obstacles to make the algorithm compatible with qubit recycling.
First, the extra-uncomputation obstacle arose because Ref~\cite{rv24} uses two accumulators but only wants to measure one.
The second accumulator ``knows too much'', forcing a reverse pass over the phase estimation qubits to uncompute it.
We address this using superposition masking~\cite{JaquesGidney20,zalka2006shors,Gidney19Encoded}, preventing the unwanted information from entering into the accumulators.
Second, the extra-conversions obstacle came from converting the phase estimation register back and forth between binary and Zeckendorf representations~\cite{Zeckendorf72}.
They need a binary representation when initially preparing the register, need a Zeckendorf representation for the control operations, and need a binary representation when performing the inverse QFT. 
We make these conversions obsolete by showing how to perform phase estimation directly on Fibonacci powers of an operation, rather than power-of-2 powers.

Our implementation of Regev's algorithm is the lowest-space implementation to date.
We need $n + M(n) + O(\sqrt{n})$ qubits, where $M(n)$ is the number of qubits required for an in-place modular multiplication.
We parameterize the space cost in this way because there has been substantial recent progress improving this subroutine~\cite{LuoEtAl26,BabbushEtAl26,Schrottenloher26,long2026}.
For example, using the in-place multiplier from \cite{LuoEtAl26}, the space cost would be $4n + O(\sqrt{n})$. See \reftab{tab:intro-factoring-comparison}.
Regev's algorithm can be implemented with half the space required by the previous state of the art. 

While the margin between Regev and Shor is now narrower, our construction still falls short of the $n/2$ + $\epsilon n$ qubits sufficient for the leading implementations of Shor’s algorithm
\cite{DBLP:journals/iacr/ChevignardFS24, gidney2025factor2048bitrsa}.
That said, we see no inherent obstacle to building on top of our ideas to even further
reduce the cost of Regev’s algorithm, and we believe future work to be critical. 
Shor's circuit has benefited from decades of optimizations, while Regev's circuit has only existed for a few years. 
Even if we ultimately determine Regev does not outperform Shor, we cannot make this determination until we have the best practical implementation of Regev.
Our work makes substantial progress towards this goal.

\begin{table}[t]
    \centering
    \renewcommand{\arraystretch}{1.15}
    \begin{tabular}{@{}llcc@{}}
        \toprule
        Algorithm & Implementation & Qubits & $n$-bit Multiplications (per run) \\
        \midrule
        Shor
        & Chevignard et al.~\cite{DBLP:journals/iacr/ChevignardFS24}
        & $n/2$ & N/A \\
        \midrule
        \multirow{4}{*}{Regev}
        & Regev~\cite{Regev23}
        & $3n^{3/2}$ & $12\sqrt{n}$ \\
        & Meyer--Ragavan--Van Kirk~\cite{KRVK26}
        & $3.72n\log n$ & $12 \sqrt{n}$ \\
        & Ragavan--Vaikuntanathan~\cite{rv24}
        & $11.32n$ & $86.4\sqrt{n}$ * \\
        & Ragavan~\cite{cryptoeprint:2024/636}
        & $10.43n$ & $45.7\sqrt{n}$ * \\
        & \textbf{This work}
        & $\mathbf{}$ & $\mathbf{}$ \\
         & \hspace{2mm} • Dialog Primitive \cite{BabbushEtAl26, Schrottenloher26}
        & $\mathbf{6.33n}$ & $\mathbf{30.25\sqrt{n}}$ \\
        & \hspace{2mm}  • Luo \textit{et al.} Primitive \cite{LuoEtAl26}
        & $\mathbf{4n}$ & $\mathbf{30.25\sqrt{n}}$ \\
        \bottomrule
    \end{tabular}
    \caption{High-level space and multiplication costs for leading
    implementations of Shor's and Regev's factoring algorithms. We report
    leading-order qubit counts and the number of quantum-quantum modular
    multiplications per circuit run; lower-order terms are omitted. Since Meyer, Ragavan, and Van Kirk \cite{KRVK26} did not report multiplication ancilla space, we assume they use the in-place multiplication primitive of \cite{LuoEtAl26}. Also note that both Ref~\cite{rv24} and Ref~\cite{cryptoeprint:2024/636} report their construction in terms of out-of-place multiplications (marked with *). 
    For our streamed implementation, we plug in two different in-place modular multiplication primitives, and reported multiplication count comes from our numerical simulations. We also derive an upper bound on required multiplications, for which the $30.25\sqrt{n}$ would simply be replaced by $O(\sqrt{n}\log n)$.
    }\label{tab:intro-factoring-comparison}
\end{table}

Our paper is structured as follows.
In \refsec{sec:techoverview}, we give a technical overview. We show how to stream Regev's factoring, directly applying superposition masking and phase estimation with Fibonacci powers. However, these key techniques are likely to also be of independent interest. So,
in \refsec{sec:space-efficient-exponentiation}, we describe how superposition masking works and pedagogically show how it prevents unwanted information from entering into accumulators.
Then,
in \refsec{sec:iterative-fibonacci-phase-estimation} we describe how to perform phase estimation with Fibonacci powers of an operation.
Finally, in \refsec{sec:regev-factoring-as-qpe} we assemble these pieces in a formal analysis, providing rigorous guarantees for our lowest-space implementation of Regev's algorithm.

\section{Technical Overview}\label{sec:techoverview}

\subsection{Streaming Regev's Factoring Algorithm}
Regev's algorithm \cite{Regev23} to factor $n$-bit numbers $N$ is the multidimensional analog of Shor's algorithm.
He extends Shor's periodic function to multiple dimensions, replacing the traditional
scalar exponent by a vector $z \in \ZZ^d$ and using a different base $a_i$ for each dimension $i$. His function takes the form
\begin{equation}
    \label{eqn:regev-function-periodic}
    f(z)=\prod_{i=1}^d a_i^{z_i}\bmod N,
    \qquad z\in\ZZ^d.
\end{equation}
This function also exhibits periodicity, $f(z+p) = f(z)$, and just as in Shor's algorithm, knowledge of this period can be used to factor $N$. 
All period vectors $p$ of this function together define a lattice $\mathcal{L}$: $p \in\mathcal{L}$ if and only if $f(z + p)=f(z)$ for $z\in\ZZ^d$. The goal of Regev's quantum circuit is to generate outputs that, with efficient classical overhead, allow us to recover a period vector $p$ in $\mathcal{L}$.

\paragraph{Recasting Regev Factoring as Iterative Quantum Phase Estimation.}
We stream Regev's algorithm by recasting the quantum part as iterative quantum phase estimation (IQPE). Recall that, for some unitary $U$ and its phase-labeled eigenstate $\ket{u_\alpha}$,  IQPE uses the eigenrelation $U\ket{u_\alpha}=e^{2\pi i\alpha}\ket{u_\alpha}$ to learn $\alpha$ with a single control qubit. 
Each iteration prepares the control qubit in $\ket{+}$, applies a
controlled power of $U$, and measures the control qubit in a chosen
basis. The classical outcomes determine the measurement bases used
in later iterations, and each iteration further refines the estimate of $\alpha$. 
In our recasting of Regev's algorithm as IQPE, we put it in this form. 
One run of Regev's original quantum circuit can be viewed as sampling a point
$\theta\in\mathcal{L}^*/\ZZ^d$ \cite{Regev23}, and our Regev IQPE circuit will sample $\theta$ from the
same distribution.

Due to its multidimensionality, Regev's algorithm as IQPE will require controlled powers of $d$ different modular multiplication unitaries, corresponding to the generators $\{a_i\}_{i=1}^d$.
As a result, the ``phase'' we learn will be a $d$-dimensional vector
$\theta=(\theta_1,\ldots,\theta_d)$.
Each iteration further refines all $\theta_i$ estimates in parallel, using a $d$-qubit control register rather than a single control. The $i$th qubit
controls the $i$th modular multiplication unitary $U_i$, which is multiplication by $a_i \pmod N$.
An iteration takes the following steps:
\begin{enumerate}
    \item Prepare the
control register in $\ket{+}^{\otimes d}$.
\item Each qubit controls its corresponding power of $U_i$.
\item Measure the controls in a basis determined by the previous measurement outcomes. 
\end{enumerate}
Given that the modular multiplication unitaries $\{U_i\}_{i=1}^d$ commute, we use their simultaneous eigenstates $\ket{u_\theta}$ for the target state. 
Each $U_i$ probes a different component $\theta_i$ of $\theta$,
\begin{equation}
    U_i \ket{u_\theta} = e^{2\pi i\theta_i}\ket{u_\theta},
\end{equation}
and we show that $\ket{1\bmod N}$ is an equal superposition of
these common eigenstates: $\ket{1\bmod N}
=
\frac{1}{\mathcal{N}}
\sum_{\theta}\ket{u_\theta}$, where
 $\mathcal{N}$ is  normalization. 
Therefore, we start our IQPE with the initial target state $\ket{1\bmod N}$, and as the circuit proceeds, it will select one of the phase vectors $\theta$ from this superposition.

\subsubsection{Superposition masking for streaming modular exponentiation}

\begin{figure}[t]
    \centering

    \begin{minipage}{\textwidth}
    \makebox[\linewidth][l]{\footnotesize\bfseries (a)}\par
    \vspace{0.2em}
    \centering
    \includegraphics[width=\linewidth]{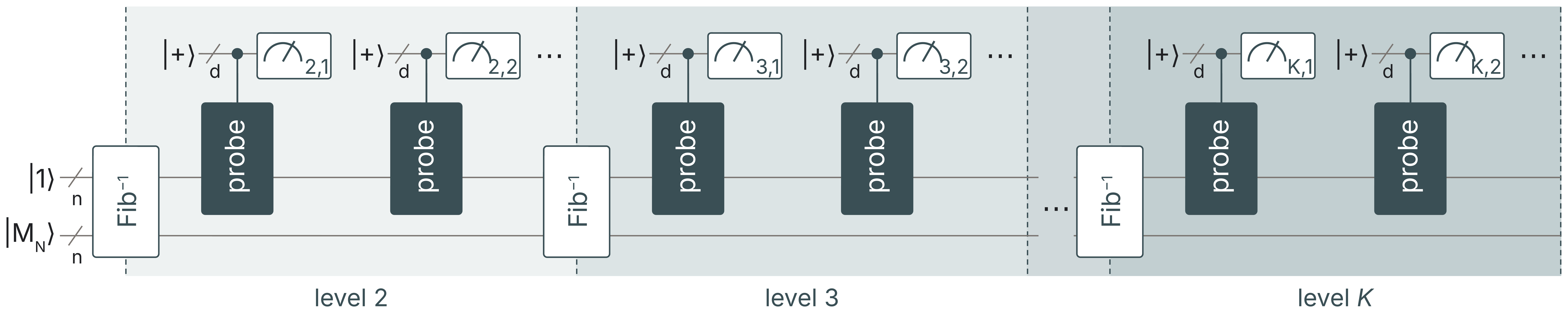}
\end{minipage}

\vspace{0.8em}

\begin{minipage}{\textwidth}
    \makebox[\linewidth][l]{\footnotesize\bfseries (b)}\par
    \vspace{0.2em}
    \centering
    \includegraphics[width=0.75\linewidth]{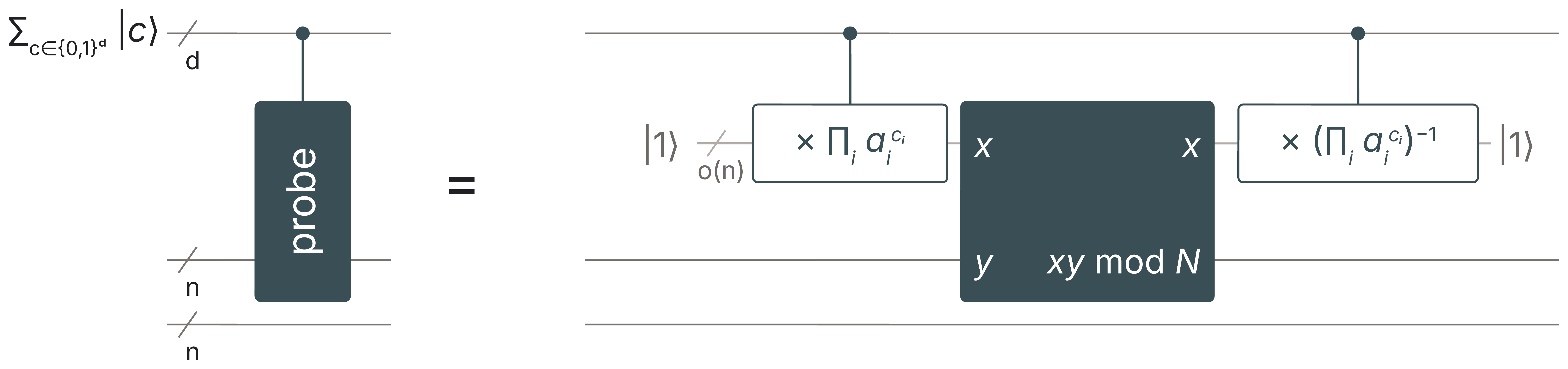}
\end{minipage}

    \caption{\textbf{Streamed implementation of the Fibonacci phase probes.}
    (a) The two persistent $n$-qubit accumulators are retained throughout
    the computation. Each inverse Fibonacci update advances the traveling
    eigenstate to the next level. At level $k$ we are in traveling eigenstate $\ket{\widetilde{u_\theta}^{(k)}}$. The recycled $d$-qubit
    control register is initialized, used for a simultaneous phase probe of $F_k\theta$,
    measured in that probe's adaptive basis (labeled $k,s$), and reset. 
    This can be repeated
    for multiple shots, indexed by $s$, before advancing.
    (b) Construction of one simultaneous probe. On a computational basis state
    branch $c=(c_1,\ldots,c_d)$ of the control register, the selected product
    $\prod_i a_i^{c_i}$ is coherently constructed in an $o(n)$-qubit
    register, multiplied into one accumulator register, and then uncomputed. The other $n$-qubit
    accumulator remains unchanged.}
    \label{fig:streaming-regev-circuit}
\end{figure}

\paragraph{Building Fibonacci phase probes via ``superposition masking.''}
Our first major technical advance is how we implement the controlled powers of $U_i$ space-efficiently, without losing Regev's favorable per-shot depth asymptotics. 
We build on
the Fibonacci accumulator method of Ragavan and
Vaikuntanathan~\cite{rv24}, which generates Fibonacci powers $U_i^{F_k}$ reversibly and online. As such we implement phase estimation with Fibonacci powers rather than the traditional power-of-2 powers.  
Their construction begins with
$\ket{1}\ket{a_i}$ and repeatedly applies the Fibonacci accumulator update $\mathsf{Fib}: \ket{x}\ket{y}
    \longmapsto
    \ket{y}\ket{xy\bmod N}.$
After $k-1$ updates, the two accumulators contain
\begin{equation}
    \mathsf{Fib}^{k-1}\ket{1}\ket{a_i} = \ket{a_i^{F_{k-1}}\bmod N}
    \ket{a_i^{F_k}\bmod N}.
    \label{eqn:RV24fibupdate}
\end{equation}
The second register has our target multiplier. One may think we are now done, but we cannot simply discard the first register. In a controlled computation,
it contains unwanted information about the control. At this point one must copy out $a_i^{F_k}$ and reverse the
Fibonacci computation to erase $a_i^{F_{k-1}}$. Repeating this forward and
backward traversal for every phase probe would eliminate the advantage of
streaming. 

We avoid this uncomputation using superposition masking.
Rather than trying to generate an efficient circuit for controlled exponentiation $U_i^{F_k}$, we define a ladder of eigenstates with eigenvalues $F_k \theta_i$. Indexed by ladder level $k$, these eigenstates $\ket{\widetilde{u_\theta}^{(k)}}$ are 
simultaneous eigenstates of all $\mathbb{I} \otimes U_i$, with eigenrelation $\mathbb{I} \otimes U_i \ket{\widetilde{u_\theta}^{(k)}} = e^{2 \pi i F_k \theta_i}\ket{\widetilde{u_\theta}^{(k)}}$.
This has the same effect as Fibonacci exponentiation, but we have offloaded the work of exponentiation to stepping along this ladder of ``traveling eigenstates.'' They take the form
\begin{equation}
\ket{\widetilde{u_\theta}^{(k)}}
\coloneqq
\mathsf{Fib}^{-(k-1)}
\bigl(\ket{\mathsf M_N}\ket{u_\theta}\bigr),
\label{eq}
\end{equation} 
where
$\ket{\mathsf M_N}$ is the uniform superposition over the multiplicative group modulo $N$. Crucially, we can efficiently step between states on this ladder. 
We move between the level $k$ and level $k+1$ eigenstates by applying the inverted Ragavan and Vaikuntanathan Fibonacci accumulator update,
$\mathsf{Fib}^{-1}$: $\mathsf{Fib}^{-1}\ket{\widetilde{u_\theta}^{(k)}} = \ket{\widetilde{u_\theta}^{(k+1)}}$. See \hyperref[fig:streaming-regev-circuit]
{Figure~\ref*{fig:streaming-regev-circuit}}(a).

The key insight in constructing these traveling eigenstates was that multiplying $\ket{\mathsf M_N}$, by any element of the multiplicative group modulo $N$, just
permutes the terms of
$\ket{\mathsf M_N}$. This is called ``superposition masking,'' where $\ket{\mathsf M_N}$ is our mask. 
Because multiplication (mod $N$) leaves the mask unchanged, 
we could use it to hide the unwanted information in Ragavan and Vaikuntanathan's first accumulator register (\refeq{eqn:RV24fibupdate}). This lets us construct the eigenstates
$\ket{\widetilde{u_\theta}^{(k)}}$ with the desired phase kickback.

\paragraph{Space cost of Fibonacci IQPE.}
Our streamed quantum circuit's qubit count is set by the space required for one level of our Fibonacci IQPE procedure. Each level consists of (1) one backward Fibonacci update between consecutive
levels and (2) some number of phase probes (controlled multiplications into the right register of the traveling eigenstate). Both feature in-place quantum-quantum multiplications, but the latter sets the space cost; there is an extra $n$ lying around for the left register of the traveling eigenstate. Again,  see \hyperref[fig:streaming-regev-circuit]
{Figure~\ref*{fig:streaming-regev-circuit}}(a). 
For this latter step, we follow Regev's careful accounting. Since the generators $a_i$ are small \cite{Regev23}, 
instead of control-multiplying each $a_i$ into the traveling eigenstate directly, we first combine the $a_i$s and then multiply in their $o(n)$-sized product.
See \hyperref[fig:streaming-regev-circuit]
{Figure~\ref*{fig:streaming-regev-circuit}}(b). 
This turns $d$ quantum-quantum in-place multiplications, where the quantum input is large, into only one such multiplication.  
This gate improvement comes at a sublinear space cost because of the additional register required 
to aggregate the $o(n)$-sized product.
Since each $a_i$ is $O(\log n)$, the product is $O(d \log n) = o(n)$ because $d = \floor{\sqrt{n}}$. 

Therefore, the space-dominant step is the multiplication of this product into the right register of the traveling eigenstate. If the space cost of an in-place multiplication of $n$-bit registers is $M(n)$, then the space cost is
\begin{equation}
n + M(n) + O(\sqrt{n}),
\end{equation}
where the extra $n$ comes from the idle left register in the traveling eigenstate. The extra $O(\sqrt{n})$ comes from the controls.

\subsubsection{Fibonacci IQPE for learning dual lattice points}

\paragraph{Regev Factoring as backwards Fibonacci IQPE.} 
Instead of traditional IQPE with power-of-2 powers, which begins with the largest power of 2 and steps down, we start with $F_2$ and then step up through the Fibonacci numbers.
We run “backwards Fibonacci IQPE”.
At the lowest level of the ladder, $F_1 = 1$, we generate a superposition over the traveling eigenstates by preparing $\ket{\mathsf M_N}\ket{1}$. This level is used only for initialization; we then apply one $\mathsf{Fib}^{-1}$ update and begin probing at level
$k=2$, for which $F_2=1$.
At each Fibonacci level $k\geq 2$, 
we probe the phase $F_k\theta$ a few times and then move to the $k+1$ traveling eigenstate.
After the final level $K$, the two accumulator registers contain the final traveling eigenstate $\ket{\widetilde{u_\theta}^{(K)}}$, corresponding to the selected $\theta$. Since we were always in an eigenstate of our controlled modular multiplications, this state is not entangled with our control qubits  and can be
discarded. 

The remaining challenge is the classical inverse problem: given the measurement transcript produced by our IQPE circuit, we must generate an estimate $w$ of the selected phase vector $\theta$.
Our streamed circuit's measurements provide a sequence of noisy observations of the Fibonacci
phases $F_k\theta_i$, so this is not immediate. 
Our second major technical advance is, therefore, to address this inverse problem. 
We provide a classical protocol for determining how many $F_k$ phase probes are needed at each level, choosing the measurement basis for the probe controls based on previous outcomes, and finally efficiently generating an estimator $w \approx \theta$.

\paragraph{Factoring from our IQPE output: A slow-refinement adaptive Bayesian protocol.}
Although each iteration measures all \(d\) coordinates simultaneously, the
\(i\)th output bit depends only on the scalar phase \(\theta_i\). We can
therefore recover the coordinates separately, treating each one as its own scalar
inverse problem.
For each coordinate, we construct a Bayesian estimator that progressively zooms in on the
unknown phase. It maintains a weighted partition of the phase circle into
\(\textsf{poly}(n)\) cells, where each cell represents a phase interval and is weighted according to the posterior probability that
\(\theta_i\) lies in that interval. Before each measurement, we adaptively
choose the control qubit's measurement basis to maximize the
information gain, based on these weights. We then use the observed outcome to update the cell
weights.

As the measurements accumulate, the posterior concentrates around a smaller
region of the phase circle. We move to larger Fibonacci powers to amplify
the remaining uncertainty, allowing subsequent measurements to efficiently
distinguish candidate phases in this high-probability region. At the same time, we refine the
cells only when there is significant posterior weight. This prevents
the classical representation from growing beyond \(\textsf{poly}(n)\).

We numerically simulate this adaptive protocol and find it needs
approximately six scalar probes per Fibonacci level. We considered bit lengths up to the cryptographically relevant sizes $n = 1024, 2048, 4096$. At each bit length, we performed 1,000 independent scalar trials.  Across the tested
instances, $99.7\%$ of scalar trials were accurate, to the level required by Regev's classical postprocessing to factor the composite $n$-bit number with constant probability.
In addition to these experiments, we also provide a worst-case asymptotic guarantee on the number of scalar probes needed per Fibonacci level. 
We show we need at most $O(\log n)$ probes at each of $K=\Theta(\sqrt n)$ Fibonacci levels
to factor an $n$-bit number with 
constant probability. 
Because each simultaneous probe
measures all $d$ coordinates in parallel, this gives
$O(\sqrt n\log n)$ total IQPE iterations.

\paragraph{Remark: Prior work in Bayesian QPE.}
Our approach is related to prior Bayesian phase-estimation techniques
\cite{SvoreHastingsFreedman13,WiebeGranade16,VanDenBerg21,DeNeeveEtAl24}.
These methods sequentially update a classical posterior distribution for an
unknown phase and can obtain high precision from relatively few
measurements. However, our inverse problem exhibits a more restricted setting. For example, Ref.\ \cite{SvoreHastingsFreedman13} allows the exponent of the
next controlled unitary to be chosen adaptively. Whereas, our exponents are fixed to
the Fibonacci numbers, and since our traveling eigenstate climbs the Fibonacci ladder, we cannot query frequencies $F_2,F_3,\ldots,F_K$ in any arbitrary order.
Finally, some prior methods also admit expensive classical postprocessing, such as direct Bayesian implementation on a uniform
grid. This requires a number of grid points inversely proportional to the
target precision, which for us would be
$2^{\Theta(\sqrt n)}$, making the classical space and time intractable.

\paragraph{Remark: The Gaussian input state of Regev's algorithm is not necessary.}
Pervasive among all variants of Regev's algorithm to date has been the need for an entangled input state \cite{Regev23, rv24, cryptoeprint:2024/636, ekeragartnercomparison, KRVK26}. For instance, in Regev's original work, the exponent registers are initialized in a discretized Gaussian
superposition over vectors in \(\ZZ^d\). This choice permits an elegant
analysis of the high-dimensional Fourier transform and yields the low-noise
dual-lattice samples required by his classical postprocessing.

Our work shows such an entangled exponent-register state is not
necessary. 
Instead, our construction replaces the original exponent registers with a sequence of \(d\)-qubit control registers initialized
to \(\ket{+}^{\otimes d}\). We show that our streamed phase-estimation procedure satisfies the same
postprocessing requirements as Regev's low-noise dual-lattice samples. This is shown formally via our analytic upper bound on the number of probes required per Fibonacci level
(\refprop{prop:streamed-regev-correctness} and
\refcor{cor:fibonacci-phase-estimation-for-regev}).

\subsection{Preliminaries}\label{sec:preliminaries}

We retain the notation of Regev~\cite{Regev23} and
Ragavan--Vaikuntanathan~\cite{rv24} and restate it here for convenience.
Let $N<2^n$ be the $n$-bit integer to be factored. Let
$d=\floor{\sqrt n}$, let $b_1,\ldots,b_d$ be small $O(\log n)$-bit integers
(for example, the first $d$ primes), and set $a_i=b_i^2\bmod N$. For an
integer $t$, let $[t]=\{1,\ldots,t\}$. All logarithms are base two, and
$\phi$ denotes the golden ratio. Regev's algorithm uses the following
parameters:
\begin{itemize}
    \item Let $C>0$ be the absolute constant in
    \refconj{conj:short-nontrivial-relation} below.

    \item Let $A>C$ be a constant specified later. For Regev's factoring
    algorithm, one takes $A=C+2+o(1)$.

    \item Let $R=2^{(A+o(1))\sqrt n}$.

    \item Let $D$ be a power of two in $[2\sqrt d\,R,4\sqrt d\,R]$.
    Also, $D=2^{(A+o(1))\sqrt n}$. These are the same parameters $R$ and
    $D$ defined by~\cite{Regev23}.
\end{itemize}

\noindent
Our results rely heavily on the structure of the Fibonacci sequence. In this paper, we use the conventional initial conditions for the Fibonacci sequence.  Namely, we use the following Fibonacci convention:
\begin{equation}
    F_0=0,\qquad F_1=1,\qquad
    F_{j+1}=F_j+F_{j-1},
\end{equation}
so that $F_2=1$, $F_3=2$, and so on.

\paragraph{Regev's factoring lattice.}
Regev defines the following lattices in $d$ dimensions:
\begin{align}
    \mathcal{L}
    &\coloneqq
    \left\{(z_1,\ldots,z_d)\in\ZZ^d:
    \prod_{i=1}^d a_i^{z_i}\equiv1\pmod N\right\},
    \label{eq:factoring-lattice}\\
    \mathcal{L}_0
    &\coloneqq
    \left\{(z_1,\ldots,z_d)\in\ZZ^d:
    \prod_{i=1}^d b_i^{z_i}\equiv\pm1\pmod N\right\}
    \subseteq\mathcal{L}.
    \label{eq:trivial-factoring-lattice}
\end{align}
The dual lattice of $\mathcal{L}$ is
\begin{equation}
    \mathcal{L}^*\coloneqq
    \left\{y\in\RR^d:
    \langle x,y\rangle\in\ZZ\quad\text{for all }x\in\mathcal{L}\right\}
    \supseteq\ZZ^d.
    \label{eq:dual-factoring-lattice}
\end{equation}

Later, when we pose Regev's algorithm as quantum phase estimation, the phase
vectors we read out will be elements of $\mathcal{L}^*/\ZZ^d$. Notice that
vectors $y$ and $y+k$, for $k\in\ZZ^d$, represent the same element: adding
an integer to any coordinate $i$ does not change a phase $e^{2\pi i y_i}$.
Thus $\mathcal{L}^*/\ZZ^d$ is the set of these equivalence classes. We
represent each class by the vector $y$ satisfying $0\leq y_i<1$ for every
$i\in[d]$. 

\paragraph{Factoring from a lattice relation.}
Regev's algorithm~\cite{Regev23} starts by ruling out simple possibilities
where factoring is easy, i.e., if $N$ is even, a prime power, or shares a
common factor with any of the $b_i$. For the remainder of this paper, we
assume none of these is the case.

The goal of Regev's algorithm~\cite{Regev23} is to find one vector
$z=(z_1,\ldots,z_d)$ in $\mathcal{L}\setminus\mathcal{L}_0$. Given such a
vector, one can construct
\begin{equation}
    b=\prod_{i=1}^d b_i^{z_i}\bmod N.
\end{equation}
Since $z\in\mathcal{L}$, $b$ must be a square root of $1$ modulo $N$;
moreover, it must be a nontrivial square root of $1$ since
$z\notin\mathcal{L}_0$. Hence $N$ divides $(b-1)(b+1)$ but not either term
individually, so $\gcd(b-1,N)$ will be a nontrivial divisor of $N$.

Regev's algorithm can be broken down into two distinct pieces. First, there
is a quantum circuit that essentially produces one noisy sample from
$\mathcal{L}^*/\ZZ^d$. This circuit is run $d+4$ times. Second, there is a
classical postprocessing procedure that takes these samples and essentially
recovers a basis for all vectors in $\mathcal{L}$ of norm at most
$T=2^{C\sqrt n}$. Under the heuristic assumption made by~\cite{Regev23}, at
least one of these vectors belongs to $\mathcal{L}\setminus\mathcal{L}_0$:

\begin{conjecture}
\label{conj:short-nontrivial-relation}
For some absolute constant $C>0$, the lattice
$\mathcal{L}\setminus\mathcal{L}_0$ contains a vector of Euclidean norm at
most $T=2^{C\sqrt n}$.
\end{conjecture}

Pilatte proved a variant of this conjecture for larger, randomly chosen prime
bases; the bases $b_1, ...  , b_d$ are allowed to be as large as $\exp(O(d \log d))$. This proven variant yields unconditional correctness for a slightly modified version of Regev's algorithm, at the cost of
polylogarithmic factors in the gate and space complexity~\cite{pilatte2024unconditional}.

\section{Streaming Modular Exponentiation}\label{sec:space-efficient-exponentiation}

\subsection{Superposition masking}
\label{sec:superposition-masking}

Suppose that a reversible computation acts as $\ket{z}\ket{0}\ket{0}
    \longmapsto
    \ket{z}\ket{f(z)}\ket{g(z)}.$
Here, $z$ is the input, $f(z)$ is the function we want to compute, and the third register contains arithmetic ``garbage'' $g(z)$ left over from computing $f(z)$.
Unfortunately, when $z$ is in superposition, the workspace containing $g(z)$ is generally
entangled with the output and cannot be measured or discarded. A standard
reversible implementation removes this garbage by copying out $f(z)$ and
running the computation backward. 

Superposition masking allows us to overcome this obstacle by arranging for the final workspace register
to have the same state for every $z$. In other words, it no longer exhibits the $z$ dependence that entangles it with the other registers, and measuring it tells us nothing about the rest of the computation. 
Superposition masking was first introduced by Jaques and Gidney~\cite{JaquesGidney20}; see also~\cite{Zalka06,Gidney19Encoded}.
Consider the following mask: let
\begin{equation}
    \ket{\mathsf{M}_N}
    \coloneqq
    \frac{1}{\sqrt{|\ZZ_N^*|}}
    \sum_{m\in\ZZ_N^*}\ket{m}
    \label{eq:uniform-multiplicative-mask}
\end{equation}
be the uniform superposition over the multiplicative group modulo $N$. If $g(z)\in\ZZ_N^*$ for every $z$, we can utilize this state as a mask; we arrange for the final workspace register to end up in this state for every $z$. 
The following lemma states the masking property that we use throughout
the paper.

\begin{lemma}[Multiplicative masking]
\label{lem:multiplicative-masking}
For every $h\in\ZZ_N^*$,
\begin{equation}
    \frac{1}{\sqrt{|\ZZ_N^*|}}
    \sum_{m\in\ZZ_N^*}\ket{mh\bmod N}
    =\ket{\mathsf{M}_N}.
\end{equation}
\end{lemma}

\begin{proof}
Multiplication by $h$ permutes the elements of $\ZZ_N^*$, so the claim follows
by relabeling the summation variable as $m'=mh\bmod N$.
\end{proof}

More specifically, if $g(z)\in\ZZ_N^*$ for every $z$, then we want to set up our computation such that our garbage ends up inside the mask $\ket{\mathsf{M}_N}$. As such, following \reflem{lem:multiplicative-masking}, 
\begin{align}
    &\sum_z \alpha_z\ket{z}\ket{f(z)}
      \left(
      \frac{1}{\sqrt{|\ZZ_N^*|}}
      \sum_{m\in\ZZ_N^*}\ket{m g(z)\bmod N}
      \right) \nonumber\\
    &\hspace{5em}=
      \left(\sum_z\alpha_z\ket{z}\ket{f(z)}\right)
      \ket{\mathsf{M}_N}.
    \label{eq:masked-garbage-factorizes}
\end{align}
The workspace is therefore unentangled from the computation. Measuring it
returns a uniformly random element of $\ZZ_N^*$ and reveals nothing about
$g(z)$.

\paragraph{Efficiently creating the masked state.}

This subsection describes how to efficiently prepare the mask used by our
algorithm. 
Instead of preparing the masked state $\ket{\mathsf{M}_N}$ exactly, we prepare the uniform superposition
over all residues modulo $N$,
\begin{equation}
    \ket{R_N}
    \coloneqq
    \frac{1}{\sqrt N}
    \sum_{m=0}^{N-1}\ket{m}.
    \label{eq:uniform-residue-mask}
\end{equation}
In addition to every element of $\ZZ_N^*$, this state contains zero and all
residues sharing a nontrivial common factor with $N$. For
$n=\ceil{\log_2 N}$, it can be prepared using $2n+O(1)$ qubits and
$O\!\left(n\right)$ gates and depth ~\cite[Figure~12]{babbush2018encoding}. 
This  circuit from Babbush \textit{et
al. }
begins with a uniform superposition over all $2^n$ $n$-bit strings and uses amplitude
amplification to retain only strings representing integers less
than $N$. 
We will show that, when $N$ is a balanced semiprime, replacing
$\ket{\mathsf{M}_N}$ by $\ket{R_N}$ changes the output distribution of our
algorithm by only an exponentially small amount.

In the lemma below and throughout this paper, \(G\), \(S\), and \(D\) denote gate count, total qubit count, and circuit depth, respectively.

\begin{lemma}[Uniform-residue mask preparation]
\label{lem:uniform-residue-mask-preparation}
Assume $N$ odd.
There exists a quantum circuit that prepares the
state $\ket{R_N}$ using
\begin{equation}
    G_{\mathrm{mask}}=O(n) \text{ gates},
    \qquad
    S_{\mathrm{mask}}=2n+O(1)\text{ space},
    \qquad
    D_{\mathrm{mask}}=O(n)\text{ depth}.
\end{equation}
\end{lemma}

\begin{proof}
Apply the uncontrolled uniform-state preparation circuit of
Babbush et al.~\cite[Figure~12]{babbush2018encoding} with $L=N$ and $k=0$.
Their circuit spans $2n+O(1)$ qubits.
Its gates consist of Hadamards, comparisons, and Clifford operations
of linear size. Executing these components serially gives the stated gate
and depth bounds.
\end{proof}

\paragraph{Negligible error from the approximate mask.} To verify that we can use the approximate mask state, rather than the true mask, we want to bound the error that the approximate mask introduces in our factoring algorithm. 
To do so, we will show that the initial approximation error cannot grow during the
remainder of the algorithm. 
Let $\varphi(N)$ count the number of residues invertible modulo \(N\):
\begin{equation}
\varphi(N)=\left|\{m\in\{0,\ldots,N-1\}:\gcd(m,N)=1\}\right|.
\end{equation}
For example, for \(N=pq\) with distinct primes, we have $\varphi(N)=(p-1)(q-1)$.

\begin{proposition}[Error from the approximate mask, per Regev run]
\label{prop:approximate-mask-error}
Let $P_{\mathrm{exact}}$ denote the distribution of the complete measurement
transcript when the algorithm begins with mask $\ket{\mathsf{M}_N}$, and let
$P_{\mathrm{approx}}$ denote the corresponding distribution when the algorithm begins with the
preparation circuit of
\reflem{lem:uniform-residue-mask-preparation}. Then
\begin{equation}
    d_{\mathrm{TV}}
    \!\left(
        P_{\mathrm{approx}},
        P_{\mathrm{exact}}
    \right)
    \leq
    \sqrt{1-\frac{\varphi(N)}{N}}.
    \label{eq:approximate-mask-transcript-error}
\end{equation}
\end{proposition}

\begin{proof}
The squared overlap between the two mask states is
\begin{equation}
    \left|\braket{\mathsf{M}_N|R_N}\right|^2
    =
    \frac{\varphi(N)}{N}.
    \label{eq:mask-state-overlap}
\end{equation}
Their trace distance is therefore
\begin{equation}
    \sqrt{1-\frac{\varphi(N)}{N}}.
    \label{eq:mask-state-trace-distance}
\end{equation}
By the data-processing inequality for trace distance, trace distance cannot increase under quantum
operations or measurements, and it upper-bounds the total variation
distance between the resulting measurement distributions. 
We put this into action on the state $\rho$ generated by the preparation circuit of
\reflem{lem:uniform-residue-mask-preparation}. Let
$\mathcal{E}$ denote our algorithm, including all adaptive
measurements. The data-processing
inequality and triangle inequality give
\begin{align}
d_{\mathrm{TV}}
\!\left(P_{\mathrm{approx}},P_{\mathrm{exact}}\right)
&\leq
D_{\mathrm{tr}}
\!\left(
\mathcal{E}(\rho),
\mathcal{E}\!\left(
\ket{\mathsf{M}_N}\!\bra{\mathsf{M}_N}
\right)
\right)
\notag\\
&\leq
D_{\mathrm{tr}}
\!\left(
\rho,
\ket{\mathsf{M}_N}\!\bra{\mathsf{M}_N}
\right)
\notag\\
&\leq
\sqrt{1-\frac{\varphi(N)}{N}}.
\end{align}
\end{proof}

In other words,  no subsequent operation or measurement can make the initial
approximation more distinguishable than it was when the mask was prepared.

If $N=pq$, then $\varphi(N)=(p-1)(q-1)$. Substituting this into
\refprop{prop:approximate-mask-error}, we get the following total variation distance: 
\begin{align}
    d_{\mathrm{TV}}
    \!\left(
        P_{\mathrm{approx}},
        P_{\mathrm{exact}}
    \right)
    &\leq
    \sqrt{1-\frac{(p-1)(q-1)}{pq}}
    \notag\\
    &=
    \sqrt{
        \frac{1}{p}
        +\frac{1}{q}
        -\frac{1}{pq}
    }.
    \label{eq:semiprime-mask-transcript-error}
\end{align}
So for balanced $n$-bit semiprimes, the error in
\refeq{eq:semiprime-mask-transcript-error} is
$O(2^{-n/4})$. In particular, if $N$ is an RSA-2048 modulus
whose factors are both 1024-bit primes, the mask infidelity is less than
$2^{-1022}$ and the resulting trace-distance error is less than $2^{-511}$. 
This is  $1.49...\times10^{-154}$.

For clarity, the remainder of our circuit analysis uses the ideal mask
$\ket{\mathsf M_N}$. We account for the negligible discrepancy introduced
by the exactly prepared state $\ket{R_N}$ in the final correctness result.

\subsection{Application: Modular Fibonacci exponentiation}\label{sec:masked-fib-exponentiation}

\subsubsection{Ragavan and Vaikuntanathan's Fibonacci exponentiation.}
We can now apply superposition masking to online Fibonacci exponentiation. Recall that Ragavan and Vaikuntanathan \cite{rv24} replaced the typical exponentiation procedure by Fibonacci exponentiation: implementing reversible exponentiation with
Fibonacci numbers in the exponent, rather than the usual powers of 2.  Looking at the typical Fibonacci recurrence $F_{j+1}=F_j+F_{j-1}$, one can notice that generating sequential Fibonacci powers
translates directly into multiplication:
\begin{equation}
a^{F_j}a^{F_{j-1}}
    =a^{F_j+F_{j-1}}
    =a^{F_{j+1}}\pmod N.
\end{equation}
Thus, successive Fibonacci powers can be generated by alternately multiplying two accumulator registers into one another. 
For a fixed $a\in\ZZ_N^*$, the accumulator
update $\ket{x}\ket{y}
    \longmapsto
    \ket{y}\ket{xy\bmod N}$  sends
\begin{equation}
    \ket{a^{F_j}\bmod N}\ket{a^{F_{j+1}}\bmod N}
    \longmapsto
    \ket{a^{F_{j+1}}\bmod N}\ket{a^{F_{j+2}}\bmod N}.
\end{equation}
Thus the construction of Ragavan and Vaikuntanathan~\cite{rv24} begins with
\(
    \ket{1}\ket{a}
    =\ket{a^{F_0}}\ket{a^{F_1}}
\)
and, after $k-1$ updates, obtains
\(
    \ket{a^{F_{k-1}}\bmod N}\ket{a^{F_k}\bmod N}.
\)
If only $a^{F_k}$ is desired, the first register cannot simply be erased.
Their clean implementation copies $a^{F_k}$ into a fresh register and then
reverses the $k-1$ accumulator updates.

\paragraph{Example: Superposition masking for 1D Fibonacci exponentiation.}
Our goal is to avoid the clean-up of the
unwanted intermediate value $\ket{a^{F_{k-1}}}$.  Consider masking the latter so that the
 computation instead ends in
\begin{equation}
    \frac{1}{\sqrt{|\ZZ_N^*|}}
    \sum_{m\in\ZZ_N^*}
        \ket{a^{F_{k-1}}m\bmod N}\ket{a^{F_k}\bmod N}
    =\ket{\mathsf{M}_N}\ket{a^{F_k}\bmod N}.
    \label{eq:masked-fibonacci-goal}
\end{equation}
The first register can then be measured and discarded without revealing the
value in the second register. The question is then: \textit{how do we modify our Fibonacci computation to end in this state?} If we isolate out the masked part from \refeqplain{eq:masked-fibonacci-goal} (i.e. removing the $a$ contributions), we have $\sum_{m\in\ZZ_N^*}\ket{m}\ket{1}$ up to normalization,
and this state can itself be written in terms of Fibonacci exponents. In particular, 
\begin{equation}
\frac{1}{\sqrt{|\ZZ_N^*|}}
\sum_{m\in\ZZ_N^*}\ket{m}\ket{1}
= 
\frac{1}{\sqrt{|\ZZ_N^*|}}
\sum_{m\in\ZZ_N^*}\ket{m^{F_1}}\ket{1^{F_0}}.
\end{equation}
We want to end in this state after the Fibonacci computation terminates. Since the original computation of \(a^{F_k}\) consists of \(k-1\) accumulator updates, we will end in this required mask state if we start in the state that is this state with the Fibonacci recurrence run backwards  \(k-1\) times.
Running the recurrence backward $k-1$ times from the exponent pair \((1,0)=(F_{-1},F_0)\) gives \((F_{-k},F_{1-k})\). Notice that the Fibonacci sequence can be extended to have negative indices where $F_{-j}=(-1)^{j+1}F_j$.
Therefore, if we begin in the state
\begin{equation}
\frac{1}{\sqrt{|\ZZ_N^*|}}
\sum_{m\in\ZZ_N^*}
\ket{m^{F_{-k}}\bmod N}
\ket{a\,m^{F_{1-k}}\bmod N},
\label{eq:masked-fibonacci-initial-state}
\end{equation}
then after \(k-1\) forward updates, the \(a\)-exponents evolve from \((F_0,F_1)\) to \((F_{k-1},F_k)\), while the \(m\)-exponents evolve from \((F_{-k},F_{1-k})\) to \((F_{-1},F_0)=(1,0)\). Hence \refeq{eq:masked-fibonacci-initial-state} evolves into \refeq{eq:masked-fibonacci-goal}.

For example, take $k=5$, so $F_4=3$ and $F_5=5$. Running four accumulator steps backward from the desired mask factors gives
\[
\bigl(m^5,m^{-3}\bigr)
\longleftarrow
\bigl(m^{-3},m^2\bigr)
\longleftarrow
\bigl(m^2,m^{-1}\bigr)
\longleftarrow
\bigl(m^{-1},m\bigr)
\longleftarrow
\bigl(m,1\bigr).
\]
We therefore initialize the two accumulators as
\begin{equation}
\frac{1}{\sqrt{|\ZZ_N^*|}}
\sum_{m\in\ZZ_N^*}
\ket{m^5}\ket{a m^{-3}}.
\end{equation}
Applying the same multiply-and-swap update $\ket{x}\ket{y}\longmapsto\ket{y}\ket{xy\bmod N}$
at each step gives
\begin{align*}
\sum_m\ket{m^5}\ket{a m^{-3}}
&\longmapsto
\sum_m\ket{a m^{-3}}\ket{a m^2}\\
&\longmapsto
\sum_m\ket{a m^2}\ket{a^2m^{-1}}\\
&\longmapsto
\sum_m\ket{a^2m^{-1}}\ket{a^3m}\\
&\longmapsto
\sum_m\ket{a^3m}\ket{a^5}\\
    &=\sum_m\ket{m} \ket{a^5},
\end{align*}
where we omit the common normalization factor and every value is taken modulo $N$. The final equality follows by relabeling $m'=a^3m$.
The desired Fibonacci power has separated from the accumulator
garbage without running the recurrence backward, achieving the desired effect.

Crucially, while this construction is pedagogical, \textit{it turns out to do more than we need in practice}. To see this, recall that we will want to probe every Fibonacci power through $F_K$: $a^{F_2},a^{F_3},\ldots,a^{F_K}.$
We can initialize the accumulators in $\ket{\mathsf{M}_N}\ket{1}$ and move backward through the Fibonacci levels, using the level-$k$ controls to multiply $a$ into the appropriate accumulator. 
After the final controlled multiplication at level $K$, the construction above would suggest traversing the Fibonacci recurrence forward, to restore the accumulators to their initial state before discarding them. 

However, this final traversal is unnecessary. 
At this point, every control has already acquired its desired phase and been measured. 
The forward traversal acts only on the accumulator registers, so it cannot change the probability distribution of the recorded control outcomes. By the deferred-measurement principle, we could imagine measuring the controls only after this traversal, but because the forward Fibonacci updates act only on accumulators, these forward updates will not affect the controls.  We may therefore discard the accumulators immediately after the final level, avoiding the forward traversal entirely.
Finally, note that we will develop this idea formally in \refsec{sec:regev-factoring-as-qpe}, using the traveling eigenstate picture.

\vspace{3mm}
\subsubsection{Masking Fibonacci Exponentiation}

\paragraph{Extension to multiple dimensions and efficiency guarantees.} 
We now extend the one-dimensional construction to Regev's full $d$-dimensional setting. Instead of just one $a$, we have an $a_i$ for each $i \in [d]$. 
Recall that, in each QPE iteration, we use $d$ controls, one for each of the $a_i$ generators. 
Let
    $c=(c_{1},\ldots,c_{d})\in\{0,1\}^d$
denote the computational-basis value of the control qubits. Then, at Fibonacci level $k$, the
multiplication operation required for this probe is
$\prod_{i=1}^d U_i^{F_kc_{i}},$
where $U_i$ denotes multiplication by $a_i$ modulo $N$. 

Rather than multiply each selected $a_i$ separately into an $n$-bit
accumulator, we first aggregate all contributions having the same Fibonacci
order. 
Define
\begin{equation}
p(c)
\coloneqq
\prod_{i=1}^d a_i^{c_{i}}.
    \label{eq:same-order-contribution}
\end{equation}
Thus, the multiplier required at level $k$ is
$\prod_{i=1}^d a_i^{F_kc_{i}}=p(c)^{F_k}\pmod N.$
We combine the $d$ controlled contributions and then multiply the product into the $n$-bit accumulator.
The benefit of doing this is that $p(c)$ is substantially smaller than $N$. Each
$a_i$ has $O(\log n)$ bits, so every possible value of $p(c)$ can be
represented using $O(d\log n)
=
O(\sqrt n\log n)$ qubits.
In particular, this register is sublinear, so $p(c)<N$ for sufficiently
large $n$ and no modular reduction is required.
This scheme of first combining the $a_i$s and then multiplying into the accumulator is depicted in \hyperref[fig:streaming-regev-circuit]
{Figure~\ref*{fig:streaming-regev-circuit}}(b).

\paragraph{Efficiency guarantees for aggregating same-order contributions across dimensions.}
We compute $p(c)$ coherently using the balanced selected-product tree
organization introduced by Regev in his original circuit~\cite{Regev23} and stated in clean
reversible form by Ragavan and
Vaikuntanathan~\cite[Lemma~5.5]{rv24}. At the leaves of the tree, the $i$th
control selects either $1$ or $a_i$; the selected values are then multiplied
together until the root register $P$ contains $\ket{p(c)}$. Once the tree has produced $\ket{p(c)}$, we perform the single
small-by-large in-place modular multiplication
\begin{equation}
\ket{p(c)}_P\ket{y}
\longmapsto
\ket{p(c)}_P
\ket{p(c)y\bmod N}
\label{eq:small-by-big-multiplication}
\end{equation}
on the relevant $n$-bit accumulator.
Finally, we reverse the selected-product tree, returning $P$ and every
intermediate tree register to $\ket{1}$. Hence no arithmetic workspace
remains entangled with the controls when they are rotated and measured.

We now formalize the clean interface provided by this construction. Let
$G_{\mathrm{tree}}$, $S_{\mathrm{tree}}$, and $D_{\mathrm{tree}}$ denote
the gate count, total qubit count, and depth of the selected-product tree.
Similarly, let $G_{\mathrm{sb}}$, $S_{\mathrm{sb}}$, and
$D_{\mathrm{sb}}$ denote the corresponding costs of the small-by-large
modular multiplication in place (\refeq{eq:small-by-big-multiplication}). Thus $S_{\mathrm{tree}}$
and $S_{\mathrm{sb}}$ include their respective input and output registers. The lemma below states the clean $p(c)$-into-accumulator update in terms of these primitives.

\begin{lemma}[Clean same-order accumulator update]
\label{lem:clean-same-order-accumulator-update}
There exists a quantum circuit implementing
\begin{equation}
    \ket{c}\ket{y}\ket{0}^{\otimes s_\mathrm{ancilla}}
    \longmapsto
    \ket{c}
    \ket{p(c)y\bmod N}
    \ket{0}^{\otimes s_\mathrm{ancilla}},
    \label{eq:clean-same-order-accumulator-update}
\end{equation}
where the ancillary workspace is size $s_\mathrm{ancilla}=\max\!\left\{ S_{\mathrm{tree}}- d, S_{\mathrm{sb}}-n\right\}$.
The circuit has gate count $2G_{\mathrm{tree}}+G_{\mathrm{sb}}$, depth $2D_{\mathrm{tree}}+D_{\mathrm{sb}}$, and total qubit count 
\begin{equation}
    d + n + s_\mathrm{ancilla} = 
    \max\!\left\{
        n+S_{\mathrm{tree}},
        d+S_{\mathrm{sb}}
    \right\}.
    \label{eq:clean-same-order-qubit-count}
\end{equation}
\end{lemma}

\begin{proof}
First compute $p(c)$ into the root register $P$ using
the selected-product tree. Apply the small-by-large multiplication from
\refeq{eq:small-by-big-multiplication}, and then reverse the tree.
This restores $P$ and every intermediate register to $\ket{0}$, proving
\refeq{eq:clean-same-order-accumulator-update}.

The tree is executed once forward and once backward, giving the 2 for product-tree computation gates and depth. The two subroutines are executed
sequentially, so their workspace can be reused. In addition to the
root register, we therefore require only the larger of the ancillas for the product-tree computation versus the ancillas for the small-big-multiplication. Including the $d$ control
qubits and the $n$-qubit accumulator gives
\refeq{eq:clean-same-order-qubit-count}.
\end{proof}

We achieve the second lemma below when using the construction of Regev and Ragavan-Vaikuntanathan (\cite{rv24}, Lemma~5.5), for the selected-product tree computation of Regev~\cite{Regev23}.  We rewrite it in our notation: 

\begin{lemma}[Selected-product computation]
\label{lem:selected-product-computation}
There exists a clean quantum circuit implementing 
$\ket{c}\ket{0}
    \longmapsto
    \ket{c}\ket{p(c)}\ket{0},$
whose gate count, total qubit count, and depth satisfy
\begin{equation}
    G_{\mathrm{tree}}=O(d\log^3 d),
    \qquad
    S_{\mathrm{tree}}=\widetilde O(d),
    \qquad
    D_{\mathrm{tree}}=O(d\log^2 d).
    \label{eq:selected-product-resource-bounds}
\end{equation}
\end{lemma}
\begin{proof}
The gate and space bounds follow from
Ref.~\cite[Lemma~5.5]{rv24}. They do not report depth, so we account for it
here.

In Regev's balanced product tree, each leaf is one of the $a_i$ contributions we want to multiply in, and we accumulate products as we move ``up'' the product tree. 
At level $j$ of Regev's balanced product tree, the operands have
$O(2^j\log d)$ bits. All multiplications within that level act on disjoint
registers and can therefore be performed in parallel. So a compute-copy-uncompute construction for this level has depth $D_{\mathrm{mult}}(2^j\log d)$,
where $D_{\mathrm{mult}}(b)$ denotes the depth required to multiply two $b$-bit integers. The compute-copy-uncompute only changes the depth by a constant factor. Therefore, we must simply add up the depth contribution from each level and choose the same multiplication as Ragavan and Vaikuntanathan. They invoke fast integer multiplication~\cite{HvdH21}, whose reversible multiplication circuit has depth
$D_{\mathrm{mult}}(b)=O(b\log b)$. Since level $j$ has integers with bits $b = 2^j \log d$, we find
\begin{align}
    D_{\mathrm{tree}}
    &=
    O\!\left(
        \sum_{j=0}^{\ceil{\log_2 d}-1}
        2^j\log d\,
        \log\!\left(2^j\log d\right)
    \right) \notag\\
    &=
    O\!\left(
        \log d
        \sum_{j=0}^{\ceil{\log_2 d}-1}
        2^j\left(j+\log\log d\right)
    \right) \notag\\
    &=
    O\!\left(
        \log d\left(d\log d+d\log\log d\right)
    \right) \notag\\
    &=
    O(d\log^2 d).
    \label{eq:selected-product-tree-depth}
\end{align}
\end{proof}

\section{Learning Dual Lattice Points from Streamed Circuit}
\label{sec:iterative-fibonacci-phase-estimation}

We present our protocol for Adaptive Fibonacci IQPE, which addresses the inverse problem underlying our streamed circuit. However, before presenting our protocol and the associated numerics, we also give an analytic upper bound on (1) the number of required Fibonacci levels and (2) the number of required phase probes per level.  The section proceeds as follows: 
\begin{itemize}
    \item  \refsec{sec:regev-required-precision}: We review the precision and success probability required by Regev's
classical lattice postprocessing~\cite[Corollary~4.5]{Regev23} when recovering dual lattice points. 

\item \refsec{sec:fibonacci-estimator-analysis}: We prove it is possible to achieve these requirements using $O(\log n)$ phase probes per level and $\Theta(\sqrt n)$ Fibonacci levels.

\item \refsec{sec:FIQPE-protocol}: Finally, we outline
our ``Slow-Refinement Bayesian Fibonacci IQPE algorithm'' and provide results of numerical experiments.
We show that, for factoring a 2048-bit number, we need on average six probes per Fibonacci level. 
\end{itemize}

\paragraph{Setting up the inverse problem of Adaptive Fibonacci IQPE.}
Our instantiation, of Regev's factoring algorithm as IQPE, proceeds through the Fibonacci levels
$k=2,\ldots,K$ in increasing order. 
At Fibonacci level $k$, 
each probe has one control qubit per dimension $i\in[d]$, and our circuit's phase kickback makes the $i$th control pick up a relative phase $e^{2\pi iF_k\theta_i}$.
We then apply an adaptive rotation $R_{k,s}$, where $s$ labels the probe shot and
$R_{k,s}  =   \bigotimes_{i=1}^d R_z(\phi_{k,i,s}).$ Here $\phi_{k,i,s}$ denotes the rotation
angle applied to the control associated with coordinate \(i\). 
The angle may depend on the outcomes of earlier measurements and may differ
between coordinates and shots. Up to a global phase, for the $s$th shot at the $k$th level, the \(i\)th control is
transformed as
\begin{equation}
    \ket{+}
    \longmapsto
    \frac{\ket{0}+e^{2\pi iF_k\theta_i}\ket{1}}{\sqrt2}
    \longmapsto
    \frac{
        \ket{0}
        +e^{i(2\pi F_k\theta_i+\phi_{k,i,s})}\ket{1}
    }{\sqrt2}.
    \label{eq:rotated-fibonacci-control}
\end{equation}
After applying a Hadamard gate and measuring in the computational basis,
the resulting bit \(b_{k,i,s}\in\{0,1\}\) satisfies
\begin{equation}
    \Pr\!\left[
        b_{k,i,s}=b
        \,\middle|\,
        \theta_i,F_k,\phi_{k,i,s}
    \right]
    =
    \frac{
        1+(-1)^b
        \cos(2\pi F_k\theta_i+\phi_{k,i,s})
    }{2}.
    \label{eq:fibonacci-probe-likelihood}
\end{equation}
We call the measurement of one coordinate in
\refeq{eq:fibonacci-probe-likelihood} a
\textit{scalar Fibonacci phase probe}.
A simultaneous phase probe performs
this measurement for all \(d\) coordinates in parallel and returns the
\(d\)-bit outcome $b_{k,s} = (b_{k,1,s},\ldots,b_{k,d,s}).$
We append these outcomes to a measurement transcript \(\tau\).
Given this transcript, our task is to solve the inverse problem: recover the sampled vector
$\theta\in\mathcal{L}^*/\ZZ^d$ accurately enough that, after $d+4$
independent runs, Regev's classical lattice postprocessing recovers a short
vector in $\mathcal{L}\setminus\mathcal{L}_0$ with constant probability. 
As discussed in Preliminaries (\refsec{sec:preliminaries}), this lets us factor our composite number with constant probability.

\subsection{Precision required for Regev's postprocessing}
\label{sec:regev-required-precision}

Regev's classical postprocessing takes $d+4$ noisy dual-lattice samples.
More precisely, let
$\theta^{(1)},\ldots,\theta^{(d+4)}
    \in\mathcal{L}^*/\ZZ^d$
be independent and uniformly distributed, and let
$w^{(j)}\in[0,1)^d$ approximate $\theta^{(j)}$. Regev's analysis requires
every sample to satisfy
\begin{equation}
    \operatorname{dist}_{\RR^d/\ZZ^d}
    \bigl(w^{(j)},\theta^{(j)}\bigr)
    <
    \frac{\sqrt d}{\sqrt{2}R},
    \qquad j\in[d+4],
    \label{eq:regev-required-vector-precision}
\end{equation}
where $\operatorname{dist}_{\RR^d/\ZZ^d}$ denotes distance on the
$d$-dimensional torus.
We use the torus distance because phases are defined modulo one, so coordinates that differ by an integer
represent the same phase. The torus distance between
$x,y\in\RR^d/\ZZ^d$ takes the form
\begin{equation}
    \operatorname{dist}_{\RR^d/\ZZ^d}(x,y)
    \coloneqq
    \min_{z\in\ZZ^d}\norm{x-y-z}_2.
    \label{eq:torus-distance}
\end{equation}
In other words, we may shift each coordinate of $x-y$ by an integer before
measuring its Euclidean length. 
Regev proves that, provided all $d+4$ estimates satisfy
\refeq{eq:regev-required-vector-precision}, his classical
postprocessing recovers a vector in
$\mathcal{L}\setminus\mathcal{L}_0$ with constant probability
\cite[Corollary~4.5 and Section~5]{Regev23}.

\begin{theorem}[Regev's noisy-sample recovery~{\cite[Corollary~4.5 and
Section~5]{Regev23}}]
\label{thm:regev-noisy-sample-recovery}
Assume \refconj{conj:short-nontrivial-relation}. Let
$\theta^{(1)},\ldots,\theta^{(d+4)}$ be independent uniform elements of
$\mathcal{L}^*/\ZZ^d$. If their estimates $w^{(1)},\ldots,w^{(d+4)}$ satisfy
\refeq{eq:regev-required-vector-precision}, then Regev's classical
postprocessing recovers a vector in $\mathcal{L}\setminus\mathcal{L}_0$ with
probability at least $1/4$ over the sampled dual-lattice points.
Consequently, it recovers a nontrivial factor of $N$.
\end{theorem}

We import this consequence of Regev's postprocessing analysis as
the target for our phase estimator.
While
the analysis in  Ref~\cite{Regev23}'s Sections 4 and 5 assumes that all estimates satisfy \refeq{eq:regev-required-vector-precision}, this is, of course, not always guaranteed. 
Regev shows his quantum sampling procedure satisfies this desired estimation error, across his $d+4$ independent
executions, except
with probability $1/\poly(d)$~\cite[Sections~3 and~5]{Regev23}.  Notice that this estimation error guarantee is separate from the constant failure probability in
Regev's lattice postprocessing.
We match this estimation error guarantee for the complete collection of estimates.
We require all our
$d+4$ estimates to satisfy \refeq{eq:regev-required-vector-precision}
jointly with probability at least $1-1/\poly(d)$:
\begin{equation}
    \Pr\!\left[
        \operatorname{dist}_{\RR^d/\ZZ^d}
        \bigl(w^{(j)},\theta^{(j)}\bigr)
        < \frac{\sqrt d}{\sqrt{2}R}
        \ \text{for every }j\in[d+4]
    \right]
    \geq 1-\frac{1}{\poly(d)}.
    \label{eq:regev-batch-precision-target}
\end{equation}

Although each probe measures all $d$ coordinates simultaneously, its $i$th
outcome depends only on the scalar phase $\theta_i$. To simplify our analysis, we therefore only need to analyze
the scalar estimate $w_i^{(j)}$ of $\theta_i^{(j)}$, where $i\in[d]$ indexes
the coordinate and $j\in[d+4]$ indexes the independent Regev sample. 
Therefore,
we will state our iterative phase estimation's formal target in terms of learning a single coordinate. 
Our
goal then, when analyzing Fibonacci IQPE, is to choose the accuracy and failure probability of each scalar
estimate so that all $d+4$ vector estimates simultaneously satisfy
\refeq{eq:regev-batch-precision-target}.

\begin{proposition}[Scalar phase-estimation target for Regev's algorithm]
\label{thm:regev-phase-estimation-target}
Let $p$ be any polynomial. For each $j\in[d+4]$, let
$\theta^{(j)}$ be an independent uniform element of
$\mathcal{L}^*/\ZZ^d$. Suppose that, for every $i\in[d]$ and $j\in[d+4]$,
our scalar phase estimator outputs $w_i^{(j)}$ satisfying
\begin{equation}
    \Pr\!\left[
        \operatorname{dist}_{\RR/\ZZ}
        \bigl(w_i^{(j)},\theta_i^{(j)}\bigr)
        \geq
        \frac{1}{\sqrt{2}R}
    \right]
    \leq
    \frac{1}{d(d+4)p(d)}.
    \label{eq:regev-required-scalar-failure-probability}
\end{equation}
Then, with probability at least $1-1/p(d)$, all $d+4$ estimated vectors
$w^{(j)}=(w_1^{(j)},\ldots,w_d^{(j)})$ simultaneously satisfy the precision requirement of Regev's
classical postprocessing:
\begin{equation}
    \operatorname{dist}_{\RR^d/\ZZ^d}
    \bigl(w^{(j)},\theta^{(j)}\bigr)
    <
    \frac{\sqrt d}{\sqrt{2}R}
    .
    \label{eq:regev-vector-precision-from-scalar-estimates}
\end{equation}
\end{proposition}

\begin{proof}
There are $d(d+4)$ scalar estimates. By a union bound, the probability that
any one of them has error at least $\frac{1}{\sqrt{2}R}$ is at most
\begin{equation}
    d(d+4)\frac{1}{d(d+4)p(d)}
    =
    \frac{1}{p(d)}.
\end{equation}
Suppose that no such failure occurs. Then for every $j\in[d+4]$, the minimization
inside the torus distance of \refeq{eq:torus-distance} separates across coordinates. Taking the square root of the expression below gives
\refeq{eq:regev-vector-precision-from-scalar-estimates}
and completes the claim.
\begin{align}
    \operatorname{dist}_{\RR^d/\ZZ^d}
    \bigl(w^{(j)},\theta^{(j)}\bigr)^2
    &=
    \sum_{i=1}^d
    \operatorname{dist}_{\RR/\ZZ}
    \bigl(w_i^{(j)},\theta_i^{(j)}\bigr)^2
    \notag\\
    &<
    d\left(\frac{1}{\sqrt{2}R}\right)^2.
\end{align}
\end{proof}

\refprop{thm:regev-phase-estimation-target} reduces the remaining
estimation task to $d(d+4)$ scalar inverse problems: one for each coordinate
$i\in[d]$ of each Regev run $j\in[d+4]$. 
As such, our analysis can abstract away the lattice problem and just consider
the scalar inverse problem for general precision $\varepsilon$ and failure
probability $\delta$.
Recall that Regev defines $d=\floor{\sqrt n}$ and $R=2^{(A+o(1))\sqrt n}$.
The required scalar precision and failure probability are therefore $\varepsilon=\varepsilon_R$ and $\delta=\delta_R$, where 
\begin{equation}
    \varepsilon_R
    \coloneqq
    \frac{1}{\sqrt{2}R}
    =
    2^{-(A+o(1))\sqrt n},
    \label{eq:regev-required-scalar-precision}
\end{equation}
\begin{equation}
    \delta_R
    \coloneqq
    \frac{1}{d(d+4)p(d)}
    =
    \frac{1}{\poly(d)}.
    \label{eq:regev-required-scalar-failure}
\end{equation}
By \refprop{thm:regev-phase-estimation-target}, any scalar estimator
achieving these two targets supplies all $d+4$ samples required by Regev's
postprocessing. In the next section, we prove
that these targets can be achieved using $O(\log n)$ probes per Fibonacci
level.

\subsection{Upper bound: \(O(\sqrt n\log n)\) phase probes per Regev run}
\label{sec:fibonacci-estimator-analysis}

Here we upper bound the number of phase probes required to solve the scalar inverse problem.
We first establish a general guarantee in terms of the final Fibonacci level
\(K\) and the number \(M_{\mathrm{lev}}\) of probes taken at each level. The value \(F_K\)
determines the final resolution, while \(M_{\mathrm{lev}}\) controls the probability of making a bad estimate.
We then use the resulting guarantee to choose
\(K=\Theta(\sqrt n)\) and \(M_{\mathrm{lev}}=O(\log n)\), such that the resulting estimate
meets the precision $\varepsilon_R$ and failure-probability $\delta_R$ requirements of Regev's postprocessing. As a result, one complete Regev run
uses \(KM_{\mathrm{lev}} = O(\sqrt n\log n)\) simultaneous phase probes. Note that no
additional factor of \(d\) appears in this probe count because all \(d\) scalar estimators are executed in parallel.

\paragraph{Upper bound on the required phase probes.}
We establish the upper bound by construction, using a deliberately simple estimator whose
purpose is only to prove the desired asymptotic guarantee. 
Fix one of the \(d+4\) independent Regev runs, indexed by \(j\in[d+4]\),
and one coordinate \(i\in[d]\). To simplify notation, write
\begin{equation}
    \theta\coloneqq\theta_i^{(j)}
    \qquad\text{and}\qquad
    w\coloneqq w_i^{(j)}.
\end{equation}
The scalar inverse problem is to estimate the unknown phase
\(\theta\in\RR/\ZZ\).

Our constructed estimator uses $2M_{\mathrm{lev}}$ probes at the first Fibonacci level and $M_{\mathrm{lev}}$ probes at each subsequent level.
At the first level \(F_2=1\), the estimator uses two measurement
quadratures to locate \(\widehat\theta_{2} \approx \theta\) on the unit circle. 
At every
subsequent level \(k\), it uses the preceding estimate
\(\widehat\theta_{k-1}\) to isolate and learn the current, unknown residual phase $\theta-\widehat\theta_{k-1}$. 
The outcomes of the \(M_{\mathrm{lev}}\) probes at this level $k$ then let us apply a correction to \(\widehat\theta_{k-1}\), obtaining a more precise estimate that we call \(\widehat\theta_k\). 
The estimator processes its outcomes online and is classically efficient.
\refapp{app:asymptotic_factoring_proof} gives the complete
construction and proves its accuracy, probe complexity, and classical
efficiency. The resulting guarantee is stated below.

\begin{theorem}[Accuracy and probe complexity of slow-refinement Fibonacci phase estimation]
\label{thm:slow-refinement-fibonacci-phase-estimation}
For every \(\theta\in\RR/\ZZ\) and integers \(K\geq2\) and \(M_{\mathrm{lev}}\geq1\),
there exists an adaptive phase-estimation procedure using exactly \(KM_{\mathrm{lev}}\)
scalar Fibonacci phase probes that outputs an estimate \(w\in\RR/\ZZ\)
satisfying
\begin{equation}
    \Pr\!\left[
        \operatorname{dist}_{\RR/\ZZ}(w,\theta)
        >
        \frac{1}{8F_K}
    \right]
    \leq
    2K e^{-M_{\mathrm{lev}}/24}.
    \label{eq:slow-refinement-scalar-failure-probability}
\end{equation}
\end{theorem}
\begin{proof}
    See \refapp{app:asymptotic_factoring_proof}.
\end{proof}

\vspace{3mm}

We will give a few sentences about the proof idea. 
\refthm{thm:slow-refinement-fibonacci-phase-estimation} proceeds by induction over the Fibonacci levels. 
The
analysis maintains the below invariant for the estimator $\widehat\theta_k$ at level $k$
\begin{equation}
    \operatorname{dist}_{\RR/\ZZ}
    (\widehat\theta_k,\theta)
    \leq
    \frac{1}{8F_k}.
\end{equation}
The probability that a single level violates the invariant decreases
exponentially with \(M_{\mathrm{lev}}\). We use Hoeffding's inequality to
bound the probability that a single update violates the invariant by
\(2e^{-M_{\mathrm{lev}}/24}\), and then the theorem follows from a union bound over the levels.

\paragraph{Achieving Regev's required precision.}

We choose the final level $K$ and the per-level probe budget $M_{\mathrm{lev}}$ so
that \refthm{thm:slow-refinement-fibonacci-phase-estimation} meets
the scalar accuracy $\varepsilon_R$ and failure-probability $\delta_R$ targets of Regev's postprocessing. Let
\begin{equation}
    K
    \coloneqq
    \min\!\left\{
        k\geq2:
        F_k>\frac{R}{4\sqrt2}
    \right\}
    \label{eq:regev-final-fibonacci-level}
\end{equation}
and
\begin{equation}
    M_{\mathrm{lev}}
    \coloneqq
    \left\lceil
        24\log\!\bigl(2K d(d+4)p(d)\bigr)
    \right\rceil,
    \label{eq:regev-per-level-fibonacci-probes}
\end{equation}
where $p$ is any polynomial.
The choice of $K$ makes the final accuracy
$1/(8F_K)$ smaller than Regev's required scalar accuracy
$\varepsilon_R = 1/(\sqrt2R)$, while the choice of $M_{\mathrm{lev}}$ makes the failure probability of
each scalar estimate at most $\delta_R = 1/\textsf{poly}(d)$. 
The following corollary
combines these guarantees across all $d$ coordinates;  for a run of Regev's algorithm, the total number of
probes is $O(\sqrt n\log n)$

\begin{corollary}[Fibonacci phase estimation for Regev's samples]
\label{cor:fibonacci-phase-estimation-for-regev}
Let \(p\) be any polynomial, and choose \(K\) and \(M_{\mathrm{lev}}\) according to
\refeqs{eq:regev-final-fibonacci-level}
and~\eqref{eq:regev-per-level-fibonacci-probes}, respectively.
Then the \(d+4\) estimated vectors satisfy Regev's precision requirement
simultaneously with probability at least
$1-\frac{1}{p(d)}$.
Each Regev run uses
    $KM_{\mathrm{lev}}
    =
    O(\sqrt n\log n)$
simultaneous phase probes.
\end{corollary}

\begin{proof}
Following
\refthm{thm:slow-refinement-fibonacci-phase-estimation},
the choice of $K$, where $\frac{1}{8F_K}
    <
    \frac{1}{\sqrt2R}$, gives an accuracy guarantee stronger
than the scalar accuracy required by
\refprop{thm:regev-phase-estimation-target}.
Moreover, the choice of $M_{\mathrm{lev}}$ gives
\begin{equation}
2Ke^{-M_{\mathrm{lev}}/24}
\leq
\frac{1}{d(d+4)p(d)}.
\end{equation}
Therefore, as required, each scalar estimate fails to meet Regev's required accuracy
with probability at most $\frac{1}{d(d+4)p(d)}$.
\refprop{thm:regev-phase-estimation-target} then gives the joint
precision guarantee for all $d(d+4)$ scalar estimates. 
All \(d+4\) vector estimates, satisfy
Regev's precision requirement simultaneously with probability at least
\(1-1/p(d)\).

It remains to count the probes. Since
$F_K=\Theta(\varphi^K)$ and
$R=2^{(A+o(1))\sqrt n}$, $K
    =
    \Theta(\log R)
    =
    \Theta(\sqrt n).$
Since $d=\floor{\sqrt n}$ and $p$ is any polynomial,
\refeq{eq:regev-per-level-fibonacci-probes} gives
    $M_{\mathrm{lev}}
    =
    O(\log n).$
Every phase probe estimates all $d$ coordinates simultaneously, so one
Regev run uses $KM_{\mathrm{lev}}
    =
    O(\sqrt n\log n)$
probes.
\end{proof}

Thus, our streaming setup with the parameters defined here, coupled with Regev's postprocessing approach, recovers a lattice relation with the same guarantees as Regev \cite{Regev23}.

\paragraph{Aside: The Gaussian input state of Regev's algorithm is not necessary.}
Our work shows a discretized Gaussian exponent-register state is not
necessary, to obtain dual-lattice estimates with Regev's required precision
and success probability.
Combining
\refprop{prop:streamed-regev-correctness} with
\refcor{cor:fibonacci-phase-estimation-for-regev},
our streamed circuit obtains the same algorithmically relevant guarantee
without preparing this Gaussian state.
\refprop{prop:streamed-regev-correctness} shows that, for the
uniformly sampled phase vector
\(\theta\in\mathcal L^*/\ZZ^d\), the streamed circuit produces the required
Fibonacci phase probes while preserving the corresponding traveling
eigenstate. Each probe uses a fresh \(d\)-qubit control register initialized
in \(\ket{+}^{\otimes d}\), followed by phase kickback, adaptive rotations,
measurement, and reset.
\refcor{cor:fibonacci-phase-estimation-for-regev} then shows that the
resulting measurement transcript determines \(\theta\) to the precision
required by Regev's postprocessing.

\subsection{Adaptive Bayesian Fibonacci IQPE}
\label{sec:FIQPE-protocol}

The preceding subsection proves our asymptotic probe bound using a
deliberately simple analytic estimator. That estimator is useful for
establishing that Fibonacci phase estimation is efficient, but it is not
intended to give a concretely low probe count. We therefore introduce a
separate adaptive Bayesian estimator and study its performance numerically.
Throughout this subsection, we consider one scalar inverse problem. We fix a
coordinate \(i\in[d]\) and an independent Regev run \(j\in[d+4]\), and write
\begin{equation}
    \theta\coloneqq\theta_i^{(j)}
    \qquad\text{and}\qquad
    w\coloneqq w_i^{(j)}.
\end{equation}
We likewise suppress these fixed indices from the scalar probe notation,
writing \(\phi_{k,s}\) and \(b_{k,s}\) for the rotation angle and measurement
outcome of the \(s\)th probe at level \(k\).

Finding the right rate, for climbing the Fibonacci levels, is an art. A level-\(k\) probe observes the wrapped phase
\(F_k\theta\bmod 1\). Increasing \(k\) makes the probe more sensitive to small
changes in \(\theta\), but it also introduces more possible inverse branches
around the phase circle. If we advance too quickly, several widely separated
values of \(\theta\) may remain consistent with the observed wrapped phase. If we advance too slowly, we require many low-\(k\) probes to resolve fine distinctions in \(\theta\), which would be easily distinguishable with only a few high-\(k\) probes.
Our task therefore boils down to deciding when the current posterior has
localized \(\theta\) sufficiently for us to begin using the next Fibonacci
exponent.

Our estimator resolves this tension through slow refinement. It
continues probing at the current Fibonacci level until the likely phase
region is sufficiently localized, advances by one level, and then increases
the resolution of its posterior representation only where that additional
resolution is useful. The estimator has three main ingredients:
\begin{enumerate}
    \item \textit{A compact adaptive posterior.}
    We represent the current posterior by a weighted collection of cells on
    the phase circle. Initially, these cells have equal width and equal
    weight. Each probe updates their weights by Bayes' rule. Upon entering a
    new Fibonacci level, the estimator bisects cells that are both wide at
    the new Fibonacci scale and likely under the current posterior. Thus, we
    refine the regions in which \(\theta\) is likely to lie, without
    maintaining a uniformly fine grid over the entire phase circle. We cap
    the number of live cells at \(B_{\max}=\poly(n)\), keeping the classical
    representation polynomial in size.

    \item \textit{Adaptive probes and level transitions.}
    The first two level-\(2\) probes use rotation angles \(0\) and
    \(-\pi/2\), respectively, thereby sampling the cosine and sine
    quadratures of the unknown phase. Every subsequent probe chooses the rotation angle
    that the current posterior predicts will maximize the expected information gain.

    At each nonfinal level \(k\), the estimator continues taking probes until
    a transition guard permits it to advance. For a transition parameter
    \(\delta_{\mathrm{move}}\), let \(A^\star\) be a shortest circular arc
    containing at least \(1-\delta_{\mathrm{move}}\) of the posterior mass,
    when the candidate phases are dilated and wrapped around the unit circle via \(F_k\theta\bmod 1\). The
    estimator advances to level \(k+1\) only when
    \begin{equation}
        \frac{F_{k+1}}{F_k}\lvert A^\star\rvert
        \leq \frac{1}{2}.
    \end{equation}
    In other words, this guard checks that, when rescaling to the next Fibonacci level, the
    high-posterior-mass region occupies at most half of the next, $F_{k+1}$-dilated unit circle.
    Posterior mass outside \(A^\star\) is retained, so passing the guard does
    not guarantee that every modular alias has been eliminated. If the guard
    does not pass, the estimator takes another adaptively chosen probe,
    updates the posterior, and checks again. As a result, the number of
    probes taken at each level can vary and is determined by the observed outcomes.

    \item \textit{Final certification.}
    At the final level \(K\), the cell grid remains fixed and the estimator
    continues probing until certifying the desired accuracy $\varepsilon_R$. For a
    certification parameter \(\delta_{\mathrm{final}}\), it certifies when an arc, of 
    length at most \(2\varepsilon_R\), carries posterior mass at least
    \(1-\delta_{\mathrm{final}}\). It then returns the midpoint of that
    arc as \(w\). If a measurement cap or the live-cell cap is
    exhausted before the relevant condition is met, the estimator returns
    \textsc{failure}.
\end{enumerate}

The complete estimator, including its refinement, transition, and
certification rules, appears in
\refapp{app:numerics-iqpe-fibonacci}. Next, our experiments below study this
scalar estimator and report scalar probe counts.

\subsubsection{Numerical experiments}
\label{sec:fibonacci-estimator-numerics}

We now simulate our adaptive Bayesian Fibonacci IQPE estimator.  Regev's
postprocessing requires sufficiently accurate estimates of every coordinate
of the sampled dual-lattice vectors. Therefore, we must determine the
number of scalar phase probes needed to achieve the accuracy required for this postprocessing. We must also assess the protocol's success probability: how often the protocol returns an estimate that is
accurate \textit{and} that the protocol certifies as being correct.  
Our simulations cover up to cryptographically relevant problem instances of $1024$-, $2048$-, and $4096$-bit numbers.
Our main results are as follows:
\begin{itemize}
    \item Probes needed per Fibonacci level:  We take $1000$ trials per value of $n$, ranging from $n=4$ to $n=4096$. Across these experiments, our protocol uses  approximately six scalar probes per Fibonacci level.
    \item Success Probability: Across all our experiments, $99.676\%$ of our trials output an estimator that is both accurate, to the required target precision, and certified. 
\end{itemize}

\paragraph{Experimental setup.}
We test $21$ values of $n$ ranging from $4$ through $4096$, with $1{,}000$
independent scalar trials at each value. The true
phase in each trial is sampled uniformly from  $\RR/\ZZ$.   To put the $1{,}000$ trials in context, this is the same order
as the scalar estimation tasks appearing in a complete $2048$-bit Regev batch.  
For the simulations, we instantiate Regev's asymptotic parameter $R$ as $R(n)\coloneqq 2^{3\sqrt n}$, and we set the largest Fibonacci level to $K=K_{\mathrm{required}}+3,$
where $K_{\mathrm{required}}$ is the required final Fibonacci level for achieving
Regev's required precision $\varepsilon_R=1/(\sqrt{2}R)$ from \refcor{cor:fibonacci-phase-estimation-for-regev}. 
\begin{equation}
    K_{\mathrm{required}}
    \coloneqq
    \min\!\left\{
        k\geq2:
        F_k>\frac{R}{4\sqrt2}
    \right\},
    \label{eq:numerical-required-fibonacci-level}
\end{equation}
We take three extra Fibonacci steps to hone in on the desired precision $\varepsilon_R$. As a result, the settings corresponding
to $n=1024,2048,4096$ have $K=140, 197, 278$ 
respectively.

For every trial, we record the number of scalar probes used at each level
$2,\ldots,K$.  We call the total number of scalar probes, divided by the
$K-1$ sampled Fibonacci levels, the observed number of probes per level.   We
report both the fraction of trials that certify according to
\refalg{alg:adaptive-bayesian-fibonacci-estimator} and the fraction that are both \emph{accurate and certified}. In other words, the latter fraction is all trials that
certify with an estimate $w$ satisfying the scalar precision
required by \refprop{thm:regev-phase-estimation-target}: $\operatorname{dist}_{\RR/\ZZ}(w,\theta)
    \leq
    \varepsilon_R$.

The randomly-sampled true phases, the complete set of tested values of $n$, their corresponding $K$, and all other numerical parameters are given in
\refapp{app:numerics-iqpe-fibonacci}. Note that the numerical parameters were
not obtained through an exhaustive optimization, and further tuning may
improve the empirical constants reported below.

\begin{figure*}[t]
    \centering
    \includegraphics[width=\textwidth]{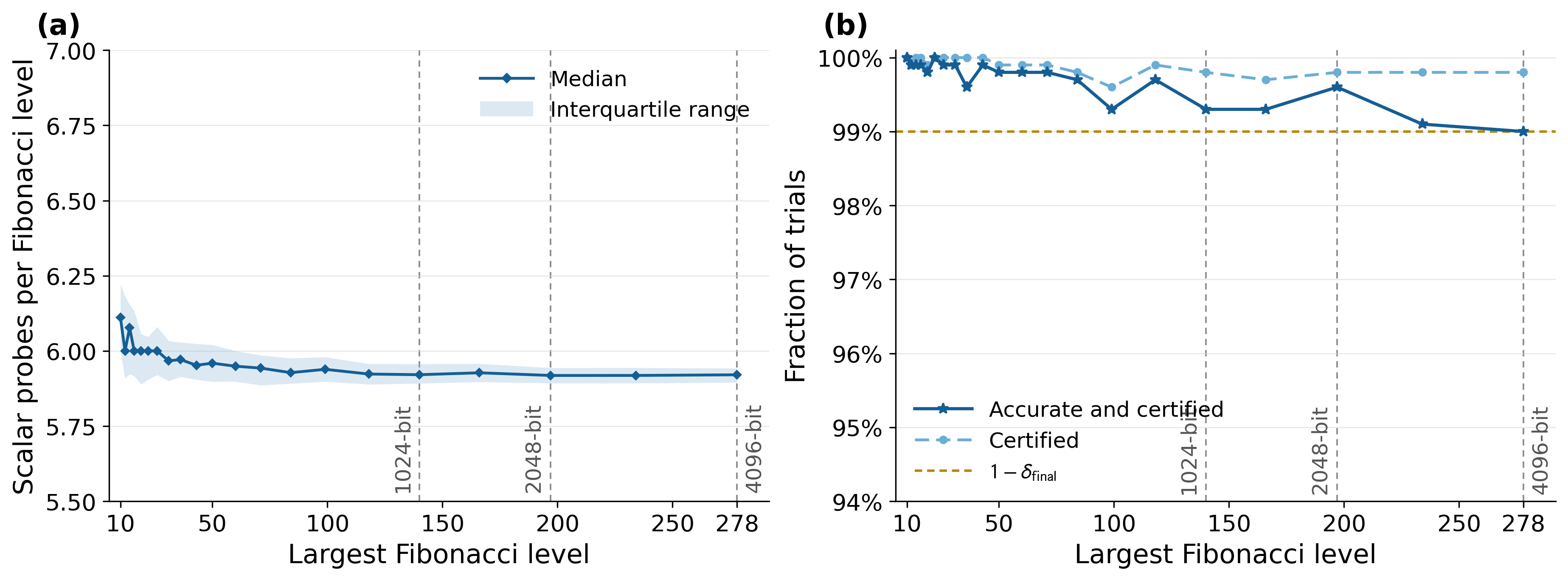}
    \caption{
    Simulating our Adaptive Bayesian Fibonacci IQPE. Vertical dashed lines mark the
    settings corresponding to factoring $1024$-, $2048$-, and $4096$-bit
    composites.
    (a) Median number of scalar probes per Fibonacci level is consistently six. The per-shot fluctuations become increasingly stable with increasing $n$; the interquartile range becomes progressively narrower.
    (b) Fraction of trials that certify and fraction that are both
    accurate and certified.  The horizontal line at
    $1-\delta_{\mathrm{final}}=0.99$ indicates the posterior-mass threshold
    used by the certification rule in our simulations. 
    }
    \label{fig:fibonacci-iqpe-numerics}
\end{figure*}

\paragraph{Results.}
We provide our results on probe complexity per Fibonacci level and the associated success probabilities in \hyperref[fig:fibonacci-iqpe-numerics]
{Figure~\ref*{fig:fibonacci-iqpe-numerics}}.
\hyperref[fig:fibonacci-iqpe-numerics]
{Figure~\ref*{fig:fibonacci-iqpe-numerics}}(a) shows that the scalar probe cost
per Fibonacci level remains close to six throughout the tested range.  
The early
finite-size variation has nearly disappeared by the cryptographically-relevant settings: 
$n=1024,2048,4096$. 
\reftab{tab:fibonacci-iqpe-rsa-results} gives the corresponding detailed
values.  
\hyperref[fig:fibonacci-iqpe-numerics]
{Figure~\ref*{fig:fibonacci-iqpe-numerics}}(b) shows the empirical certification
and accurate-and-certified rates.  Across all $21{,}000$ trials,
$99.890\%$ certify, while $99.676\%$
are both accurate and certified. 
The difference between these two curves is
the false-certification rate.  

To maintain that our protocol remains numerically tractable, we next consider the cost of our Fibonacci IQPE's classical postprocessing.
The number of live cells is the principal state-size measure for the
classical posterior and directly controls its storage and the cost of the
expected-information-gain evaluations.  
The number of live cells is a principal contributor to the storage cost and cost of evaluating the
expected-information-gain of the classical posterior.  At
$n=4096$, the median final
live-cell count is $1521$, and the
largest live-cell count anywhere in our sweep over $n$ is $2018$. Recall that in \refalg{alg:adaptive-bayesian-fibonacci-estimator}, we guarantee efficient classical processing by bounding $B_\text{max}$ to be at most $\textsf{poly}(n)$. Here, we set $B_\text{max}=30,000$, and this largest live-cell count of $2018$ does not get close to this $B_\text{max}.$  The cell budget is therefore comfortably
nonbinding throughout the tested range, and the classical IQPE
postprocessing remains easily simulatable at all parameter settings
considered here.

\paragraph{Dominant error mechanism is due to aliasing.}
In \refapp{app:numerics-iqpe-fibonacci} we do a small investigation into the dominant failure modes of our algorithm's false certifications.  The majority of our false certifications are consistent with aliasing.  A level-$k$
probe depends on the candidate phase $\theta$ only through
$F_k\theta\bmod 1$. A single level therefore admits $F_k$ inverse branches,
separated by $1/F_k$ around the phase circle. The earlier Fibonacci levels
and the transition guard are intended to identify the branch containing the
true phase before the estimator advances. The errors we observe are consistent with our protocol assuming the wrong branch when advancing.

\begin{table*}[t]
    \centering
    \resizebox{\textwidth}{!}{%
    \begin{tabular}{c c c c c c c}
        \toprule
        $n$
        & $\varepsilon_R$
        & $K$
        & Median probes/level [IQR]
        & Certified
        & Accurate and certified
        & Median live cells [IQR]
        \\
        \midrule
        $1024$
        & $2^{-96.5}$
        & $140$
        & $5.921\ [5.892,\,5.957]$
        & $99.8\%$
        & $99.3\%$
        & $779\ [768,\,789]$
        \\
        $2048$
        & $2^{-136.3}$
        & $197$
        & $5.918\ [5.893,\,5.944]$
        & $99.8\%$
        & $99.6\%$
        & $1084.5\ [1072,\,1097]$
        \\
        $4096$
        & $2^{-192.5}$
        & $278$
        & $5.921\ [5.895,\,5.942]$
        & $99.8\%$
        & $99.0\%$
        & $1521\ [1505,\,1536]$
        \\
        \bottomrule
    \end{tabular}%
    }
    \caption{
    Numerical Results for Factoring $1024$-, $2048$-, and $4096$-bit numbers.  The
    median probe counts are the total number of scalar probes divided by the
    $K-1$ sampled Fibonacci levels.  The final column reports the median
    number of live posterior cells and its interquartile range. \refapp{app:numerics-iqpe-fibonacci} gives the complete
    parameter specification and a finite-range linear fit to the mean total probe count.
    }\label{tab:fibonacci-iqpe-rsa-results}
\end{table*}

\vspace{5mm}
\noindent
For more information,
\refapp{app:numerics-iqpe-fibonacci}
records the complete simulator parameters, the detailed RSA-scale results,
the finite-range fit, the stopping events, and the error diagnostic.

\section{Regev Factoring as Iterative Quantum Phase Estimation}\label{sec:regev-factoring-as-qpe}

We pull together the key technical advances from the previous sections to repose Regev's algorithm as IQPE. 
\refsec{sec:space-efficient-exponentiation} and \refsec{sec:iterative-fibonacci-phase-estimation} developed, respectively, the masked Fibonacci arithmetic and a Fibonacci phase-estimation procedure that meets Regev’s postprocessing requirements. We now combine these ingredients into a complete streamed implementation and analyze its efficiency.
The cost of our algorithm will depend on two ingredients: the number of probes required to
recover the sampled phase and the resources
required to implement each phase probe. 
For the former, we simply import the results of \refsec{sec:iterative-fibonacci-phase-estimation}.
We will address the latter here, by constructing the required phase probes
from modular multiplication, Fibonacci exponentiation, and superposition
masking.

\subsection{Multiplication unitaries and their eigenstates}

We first define the Regev phase estimation unitaries and corresponding eigenstates. For each $i\in[d]$, we have the modular multiplication unitary $U_i\ket{x}=\ket{a_i x\bmod N},$ for $x\in\ZZ_N^*.$
Recall from the previous sections that we will use $U_i$ in phase estimation to learn the $i$th component, $\theta_i$ of the unknown phase vector $\theta$.  
We can define the simultaneous eigenbasis of these mutually commuting unitaries $U_1,\ldots,U_d$. For each $\theta\in\mathcal{L}^*/\ZZ^d$, let
\begin{equation}
    \ket{u_\theta}
    \coloneqq
    \frac{1}{\sqrt{\left|\ZZ^d/\mathcal{L}\right|}}
    \sum_{[x]\in\ZZ^d/\mathcal{L}}
    e^{-2\pi i\langle\theta,x\rangle}\ket{f(x)},
    \label{eq:regev-joint-eigenstate}
\end{equation}
where $[x]$ denotes the coset that contains every $x+\ell$ with $\ell\in\mathcal L$.
Just as in Shor's algorithm, the distinct values of $f$ are indexed by the cosets in $\ZZ^d/\mathcal{L}$. 
Thus $\left|\ZZ^d/\mathcal{L}\right|$ is the number of resulting cosets -- or, equivalently, the number of distinct values of his periodic function $f$.
Notice that although the sum in \refeqplain{eq:regev-joint-eigenstate} is indexed by cosets $[x]$, each summand is written using a particular representative $x$. The particular choice of $x$ does not matter. Replacing $x$ by $x+\ell$, where $\ell\in\mathcal{L}$, changes neither the basis state nor its phase:
\begin{equation}
f(x+\ell)=f(x)
\qquad\text{and}\qquad
e^{-2\pi i\langle\theta,x+\ell\rangle}
=e^{-2\pi i\langle\theta,x\rangle}.
\end{equation}

\begin{lemma}[\textbf{Regev QPE eigenrelation}]\label{lem:regev-joint-eigenstates}
For every $\theta\in\mathcal{L}^*/\ZZ^d$, the state $\ket{u_\theta}$ is a common eigenstate of $U_1,\ldots,U_d$, with $U_i\ket{u_\theta}
    =e^{2\pi i\theta_i}\ket{u_\theta}$ for every $i\in[d]$.
More generally, for every $z=(z_1,\ldots,z_d)\in\ZZ^d$,
\begin{equation}
    \left(\prod_{i=1}^d U_i^{z_i}\right)\ket{u_\theta}
    =e^{2\pi i\langle\theta,z\rangle}\ket{u_\theta}.
    \label{eq:regev-joint-eigenvalue}
\end{equation}
\end{lemma}

\begin{proof}
First, notice that for any $x,z\in\ZZ^d$,
\begin{equation}
    \left(\prod_{i=1}^d U_i^{z_i}\right)\ket{f(x)}
    =\ket{f(x+z)}.
\end{equation}
Therefore, reindexing the sum by $x'=x+z$ gives
\begin{align}
    \left(\prod_{i=1}^d U_i^{z_i}\right)\ket{u_\theta}
    &=
    \frac{1}{\sqrt{\left|\ZZ^d/\mathcal{L}\right|}}
    \sum_{[x]\in\ZZ^d/\mathcal{L}}
    e^{-2\pi i\langle\theta,x\rangle}\ket{f(x+z)}\\
    &=
    \frac{1}{\sqrt{\left|\ZZ^d/\mathcal{L}\right|}}
    \sum_{[x']\in\ZZ^d/\mathcal{L}}
    e^{-2\pi i\langle\theta,x'-z\rangle}\ket{f(x')}\\
    &=
    e^{2\pi i\langle\theta,z\rangle}\ket{u_\theta}.
\end{align}
Moreover, taking $z$ to be the $i$th standard basis vector gives the first claim specifying the eigenvalue $e^{2\pi i \theta_i}$ of $U_i$.
\end{proof}

These eigenstates $\{\ket{u_\theta}:\theta\in\mathcal{L}^*/\ZZ^d\}$ form an
orthonormal basis for the span of
$\{\ket{f(x)}:[x]\in\ZZ^d/\mathcal{L}\}$. For instance, we can decompose $\ket{1\bmod N}=\ket{f(0)}$ in this eigenbasis.

\begin{lemma}[\textbf{Uniform decomposition into Regev eigenstates}] We can decompose $\ket{1\bmod N}$ in the eigenbasis $\{\ket{u_\theta}\}_\theta$ defined by \refeqplain{eq:regev-joint-eigenstate}
    \begin{equation}
\ket{1\bmod N}
=
\frac{1}{\sqrt{\left|\mathcal{L}^*/\ZZ^d\right|}}
\sum_{\theta\in\mathcal{L}^*/\ZZ^d}\ket{u_\theta}.
\label{eq:one-as-regev-eigenstates}
\end{equation}
\label{lem:uniform-decomposition-of-one}
\end{lemma}
\begin{proof}
    Consider summing the phase factors over all $\theta$:
\begin{equation}
    \sum_{\theta\in\mathcal{L}^*/\ZZ^d}
      e^{-2\pi i\langle\theta,x\rangle}
    =
    \begin{cases}
        \left|\mathcal{L}^*/\ZZ^d\right|, & x\in\mathcal{L},\\
        0, & x\notin\mathcal{L}.
    \end{cases}
    \label{eq:dual-phase-cancellation}
\end{equation}
The first line follows because, if $x \in \mathcal L$, then by definition $\langle\theta,x\rangle = 0 \pmod 1$. 
For the second line, if \(x\notin\mathcal L\), we can choose some
\(\theta_0\in\mathcal L^*/\ZZ^d\) such that
\(e^{2\pi i\langle\theta_0,x\rangle}\neq1\). Since adding
\(\theta_0\) merely permutes the elements of
\(\mathcal L^*/\ZZ^d\), we can reindex the sum:
\begin{equation}
x \notin\mathcal L: 
\hspace{5mm}
\sum_{\theta\in\mathcal{L}^*/\ZZ^d}
      e^{-2\pi i\langle\theta,x\rangle}
    =
    \sum_{\theta + \theta_0\in\mathcal{L}^*/\ZZ^d}
      e^{-2\pi i\langle\theta,x\rangle}
    =
    e^{2\pi i \langle\theta_0,x\rangle}
    \sum_{\theta\in\mathcal{L}^*/\ZZ^d}
      e^{-2\pi i\langle\theta,x\rangle}
\end{equation}
In doing so, we get an extra phase term $e^{2\pi i \langle\theta_0,x\rangle}$.
Now we have an expression that equals itself,
multiplied by this nontrivial phase. The expression must therefore vanish.
We will use \refeqplain{eq:dual-phase-cancellation} in the expansion below. 

Substituting \refeq{eq:regev-joint-eigenstate} for $\ket{u_\theta}$ gives the first line, and in the second line we use that $\left|\mathcal{L}^*/\ZZ^d\right| =\left|\ZZ^d/\mathcal{L}\right|$. Finally, all terms cancel except the one indexed by the zero coset, for which
$f(0)=1 \mod N$.
\begin{align}
\frac{1}{\sqrt{\left|\mathcal{L}^*/\ZZ^d\right|}}
\sum_{\theta\in\mathcal{L}^*/\ZZ^d}\ket{u_\theta}
&\quad=
\frac{1}{
\sqrt{\left|\mathcal{L}^*/\ZZ^d\right|}
\sqrt{\left|\ZZ^d/\mathcal{L}\right|}}
\sum_{[x]\in\ZZ^d/\mathcal{L}}
\left(
\sum_{\theta\in\mathcal{L}^*/\ZZ^d}
e^{-2\pi i\langle\theta,x\rangle}
\right)\ket{f(x)}\\
&\quad=
\frac{1}{
\left|\mathcal{L}^*/\ZZ^d\right|}
\sum_{[x]\in\ZZ^d/\mathcal{L}}
\left(
\sum_{\theta\in\mathcal{L}^*/\ZZ^d}
e^{-2\pi i\langle\theta,x\rangle}
\right)\ket{f(x)}\\
&\quad=\ket{f(0)}\\
&\quad=\ket{1\bmod N}.
\end{align}
\end{proof}

Thus, simultaneous phase estimation of
\(U_1,\ldots,U_d\) starting from \(\ket{1\bmod N}\) samples a
\(\theta\in\mathcal L^*/\ZZ^d\) uniformly. 
Moreover, phase-estimation with this starting target state samples the same uniformly random dual-lattice point
$\theta\in\mathcal L^*/\ZZ^d$
that underlies the output distribution of Regev's original quantum
circuit~\cite[Section~3, Equations~(2)--(3), and
Proposition~A.6]{Regev23}.
After the projection into $\ket{u_\theta}$, a probe $U_i^{z_i}$ continues to learn the phase component $z_i\theta_i\bmod 1$, while leaving $\ket{u_\theta}$ unchanged.
We therefore analyze our circuit's measurement outcomes by first choosing $\theta\in\mathcal{L}^*/\ZZ^d$
uniformly and then retaining $\ket{u_\theta}$ throughout the computation. This is analytically equivalent to projecting onto some $\ket{u_\theta}$ as we slowly learn $\theta$ via our phase probe measurements.

\subsection{Constructing the ladder of traveling eigenstates}

We now formalize the traveling-eigenstate picture using the multiplication eigenstates $\ket{u_\theta}$. In our IQPE
formulation, a level-\(k\) phase probe initializes \(d\) controls in
\(\ket{+}^{\otimes d}\) and must efficiently imprint the phases
\(F_k\theta_1,\ldots,F_k\theta_d\) across these controls. Each computational-basis
branch $c=(c_1,\ldots,c_d)\in\{0,1\}^d$
of the control register must acquire the relative phase
\(e^{2\pi iF_k\langle\theta,c\rangle}\).
As discussed in \refsec{sec:space-efficient-exponentiation}, we achieve this by utilizing Regev's observation \cite{Regev23} that the multipliers $a_i$ are small. We first use our controls to combine the $a_i$s into the product 
\begin{equation}
    p(c)
    \coloneqq
    \prod_{i=1}^d a_i^{c_i},
\end{equation}
and then multiply
\(p(c)^{F_k}\) into \(\ket{u_\theta}\) to produce
the desired phase. 
However, we do not actually explicitly
construct the Fibonacci power \(p(c)^{F_k}\).

Instead, we use superposition masking to build a  ``ladder of traveling eigenstates.'' When multiplied by $p(c)$, the level-$k$ traveling eigenstate gives the desired phase kickback \(e^{2\pi iF_k\langle\theta,c\rangle}\). 
We begin with two persistent accumulator registers in 
\begin{equation}
\ket{\mathsf M_N}\ket{1\bmod N}
=
\frac{1}{\sqrt{\left|\mathcal{L}^*/\ZZ^d\right|}}
\sum_{\theta\in\mathcal{L}^*/\ZZ^d}\ket{\mathsf M_N}\ket{u_\theta}.
\end{equation}
By \reflem{lem:uniform-decomposition-of-one}, the second accumulator is the uniform
superposition of the states \(\ket{u_\theta}\). It
suffices to analyze just one fixed eigenstate component
$\ket{\mathsf M_N}\ket{u_\theta}$ because our measurements will select out one $\theta$.
We move this state $\ket{\mathsf M_N}\ket{u_\theta}$  backward through the Fibonacci recurrence. Recall that
the forward and backward accumulator updates are
\begin{equation}
    \mathsf{Fib}\ket{x}\ket{y}
    =
    \ket{y}\ket{xy\bmod N},
    \qquad
    \mathsf{Fib}^{-1}\ket{x}\ket{y}
    =
    \ket{yx^{-1}\bmod N}\ket{x}.
    \label{eq:fibonacci-accumulator-update}
\end{equation}
After \(k-1\) backward updates, we multiply the right accumulator
by \(p(c)\). In our analysis below, we show that  multiplication by \(p(c)\) in the right accumulator is equivalent to 
multiplication by $p(c)^{F_{k-1}}$ and $p(c)^{F_k}$ in the original left and right accumulators, $\ket{\mathsf M_N}\ket{u_\theta}$, respectively.
This leaves \(\ket{\mathsf M_N}\) unchanged because
\(p(c)^{F_{k-1}}\in\ZZ_N^*\) and therefore merely permutes the terms in the
mask. The second multiplication acts on \(\ket{u_\theta}\) with eigenvalue $e^{2\pi iF_k\langle\theta,c\rangle}.$
Thus, after \(k-1\) backward updates, the accumulator pair is an eigenstate
of multiplication by \(p(c)\) in its right register, with precisely the
eigenvalue required at level \(k\). 
We call this state, after \(k-1\) backward Fibonacci updates, the
traveling eigenstate at level \(k\):
\begin{equation}
    \ket{\widetilde{u_\theta}^{(k)}}
    \coloneqq
    \mathsf{Fib}^{-(k-1)}
    \bigl(
        \ket{\mathsf M_N}\ket{u_\theta}
    \bigr).
    \label{eq:moving-fibonacci-eigenstate}
\end{equation}
We formally show in
\reflem{lem:streamed-fibonacci-phase-kickback} that multiplying the right accumulator by $p(c)$ gives the desired phase back onto the
controls while leaving the persistent accumulator state unchanged: $e^{2\pi iF_k\langle\theta,c\rangle}
    \ket{\widetilde{u_\theta}^{(k)}}.$

This gives a natural way to move through the levels of the phase estimator.
Only the two accumulators persist
between probes and levels. See \hyperref[fig:streaming-regev-circuit]
{Figure~\ref*{fig:streaming-regev-circuit}}(a).
We take one backward Fibonacci update and perform all \(M_2\) probes at
level \(2\). Because every probe leaves the accumulators in the same
traveling eigenstate, that state can be reused for every probe at the level.
We then apply one additional backward update and perform all \(M_3\) probes. We
continue in this way through the final level \(K\):
\[
    \text{initialize }\ket{\mathsf M_N}\ket{1}
    \xrightarrow{\textsf{Fib}^{-1}}
    \text{level-\(2\) probes}
    \xrightarrow{\textsf{Fib}^{-1}}
    \text{level-\(3\) probes}
    \xrightarrow{\textsf{Fib}^{-1}}\cdots\xrightarrow{\textsf{Fib}^{-1}}
    \text{level-\(K\) probes}.
\]
Consequently, reaching level \(k\) requires exactly \(k-1\) backward
updates, and the complete algorithm uses \(K-1\) such updates.
After the final level $K$, one might be tempted to apply \(K-1\) forward
updates to restore the accumulators to their starting state. 
This is unnecessary. Every control register has already acquired
its desired phase, been measured, and reset. The forward traversal
would therefore act only on the two accumulator registers that are about to
be discarded, so we omit it entirely.

\subsection{Formal streaming algorithm}

\begin{algorithm}[t]
    \KwData{The modulus $N$, bases $a_1,\ldots,a_d$, largest level $K$,
    probe counts $M_2,\ldots,M_K$, the adaptive rotation rule, and a mask
    register prepared in $\ket{\mathsf M_N}$.}
    \KwResult{The measurement transcript $\tau$ from one streamed run.}
    \begin{enumerate}
        \item Initialize the persistent accumulators $X,Y$ in
        $\ket{\mathsf M_N}_X\ket{1}_Y$, and set $\tau$ to be empty.

        \item For $k=2,\ldots,K$, in increasing order:
        \begin{enumerate}
            \item Apply $\mathsf{Fib}^{-1}$ to $X,Y$.

            \item For $s=1,\ldots,M_k$:
            \begin{enumerate}
                \item Prepare the reusable $d$-qubit control register $C$ in
                $\ket{+}^{\otimes d}$.

                \item Coherently compute
                \[
                    p(c)=\prod_{i=1}^d a_i^{c_i}
                \]
                into the root register $P$.

                \item Multiply $P$ into the right accumulator:
                \[
                    \ket{p(c)}_P\ket{y}_Y
                    \longmapsto
                    \ket{p(c)}_P\ket{p(c)y\bmod N}_Y.
                \]

                \item Reverse the computation of $p(c)$, returning $P$ and
                its workspace to $\ket{0}$.

                \item Apply the rotation \(R_{k,s}\) prescribed by the current transcript, measure $C$ in the $X$-basis, append the outcome to $\tau$, and
                reset $C$.
            \end{enumerate}
        \end{enumerate}

        \item Discard $X,Y$ and return $\tau$.
    \end{enumerate}
    \caption{One streamed simultaneous-phase-estimation run.}
    \label{alg:streamed-regev-run}
\end{algorithm}

We now formalize the construction above. The phase-estimation analysis of
\refsec{sec:iterative-fibonacci-phase-estimation} specifies the number
\(M_k\) of probes performed at each level and the adaptive measurement
rotations. For the \(s\)th probe at level \(k\), we apply the adaptive rotation
\begin{equation}
    R_{k,s}
    \coloneqq
    \bigotimes_{i=1}^d
    R_z\!\left(\phi_{k,i,s}\right),
    \label{eq:adaptive-control-rotation}
\end{equation}
where the angles \(\phi_{k,i,s}\) may depend on the outcomes of all
preceding measurements. 
We treat these quantities as inputs to the quantum circuit.

Here we describe the quantum circuit that realizes the required probes.
At level $k$, each fresh set of $d$ controls must acquire the phases $F_k\theta_1,\ldots,F_k\theta_d$..
Conditioned on a fixed
\(\theta\in\mathcal L^*/\ZZ^d\), the correctness target for each IQPE probe at
level \(k\) will be
\begin{equation}
    \ket{+}^{\otimes d}
    \ket{\widetilde{u_\theta}^{(k)}}
    \longmapsto
    \left(
        \bigotimes_{i=1}^d
        \frac{
            \ket{0}
            +e^{2\pi iF_k\theta_i}\ket{1}
        }{\sqrt{2}}
    \right)
    \ket{\widetilde{u_\theta}^{(k)}},
    \label{eq:desired-streamed-control-state}
\end{equation}
and all nonpersistent
arithmetic-workspace registers must return to \(\ket{0}\) before the
controls are measured. 
We state the complete streamed procedure in
\refalg{alg:streamed-regev-run}.

\vspace{5mm}

We now state the correctness and resource guarantees of the complete
streamed algorithm. Correctness will follow by proving that the persistent
accumulators satisfy \refeq{eq:desired-streamed-control-state} at
every level, while the resource bounds follow by counting the \(K-1\)
big-by-big Fibonacci updates and the $M\coloneqq
    \sum_{k=2}^K M_k$  clean same-order accumulator
updates.
For the resource accounting, 
let \(G_{\mathrm{bb}}\), \(S_{\mathrm{bb}}\), and
\(D_{\mathrm{bb}}\) denote the gate count, total qubit count, and depth of
one big-by-big backward Fibonacci update
\(\mathsf{Fib}^{-1}\) on the two \(n\)-qubit accumulators. Recall that
\(G_{\mathrm{tree}}\), \(S_{\mathrm{tree}}\), and
\(D_{\mathrm{tree}}\) denote the corresponding costs of the selected-product
tree, while \(G_{\mathrm{sb}}\), \(S_{\mathrm{sb}}\), and
\(D_{\mathrm{sb}}\) denote the costs of the small-by-big in-place modular
multiplication. The complete clean same-order accumulator update, including
the forward and reverse product-tree computations, is given by
\reflem{lem:clean-same-order-accumulator-update}.
Finally, \(G_{\mathrm{mask}}\), \(S_{\mathrm{mask}}\), and
\(D_{\mathrm{mask}}\) denote the cost of preparing the multiplicative mask. 

\vspace{3mm}
\begin{theorem}[Streamed phase-probe algorithm]
\label{thm:streamed-phase-probe-algorithm}
Suppose that $\ket{\mathsf M_N}$ can be prepared using
$G_{\mathrm{mask}}$ gates, $S_{\mathrm{mask}}$ qubits, and depth
$D_{\mathrm{mask}}$. 
\refalg{alg:streamed-regev-run} produces the same distribution on
measurement transcripts as the reference procedure that directly applies
the controlled unitaries $U_i^{F_k}$, followed by the same adaptive
rotations and measurements, using $M_k$ probes at every level
$k\in\{2,\ldots,K\}$.
The algorithm uses
\begin{align}
    G_{\mathrm{stream}}
    &=
    G_{\mathrm{mask}}
    +(K-1)G_{\mathrm{bb}}
    +M\bigl(2G_{\mathrm{tree}}+G_{\mathrm{sb}}\bigr)
    +O(dM),
    \label{eq:streamed-algorithm-gate-count}\\
    S_{\mathrm{stream}}
    &=
    \max\!\left\{
        S_{\mathrm{mask}},
        S_{\mathrm{bb}},
        2n+S_{\mathrm{tree}},
        n+d+S_{\mathrm{sb}}
    \right\},
    \label{eq:streamed-algorithm-qubit-count}\\
    D_{\mathrm{stream}}
    &=
    D_{\mathrm{mask}}
    +(K-1)D_{\mathrm{bb}}
    +M\bigl(2D_{\mathrm{tree}}+D_{\mathrm{sb}}\bigr)
    +O(M).
    \label{eq:streamed-algorithm-depth}
\end{align}
gates, qubits, and depth, respectively.
\end{theorem}

\vspace{5mm}
\noindent
\paragraph{Correctness.}
We first show that \refalg{alg:streamed-regev-run} implements the
simultaneous phase probes required for iterative quantum phase estimation. 
In particular, we show
that each probe produces the correct phase kickback on its control register,
while leaving the two persistent accumulators in the same state.  We use the following register names throughout our correctness and resource
accounting:
\begin{itemize}
    \item \textbf{Accumulators $X$ and $Y$:}
    The two persistent $n$-qubit accumulator registers. They are
    initialized in $\ket{\mathsf M_N}_X\ket{1}_Y$. In the eigenstate
    decomposition of $\ket{1}$, $Y$ initially contains $\ket{u_\theta}$; after the
    backward Fibonacci updates, the pair $XY$ jointly contains the traveling
    eigenstate $\ket{\widetilde{u_\theta}^{(k)}}$. 

    \item \textbf{Controls $C$:}
Each probe uses a $d$-qubit control register, with one qubit for every
dimension $i\in[d]$, initialized in $\ket{+}^{\otimes d}$. On a
computational-basis branch $c\in\{0,1\}^d$, its bits determine the
selected multiplier
\begin{equation}
    p(c)=\prod_{i=1}^d a_i^{c_{i}}
\end{equation}
and, therefore, the resulting phase $e^{2\pi iF_k\langle\theta,c\rangle}.$
The register is then rotated, measured,
reset, and reused for the next probe.

    \item \textbf{Product register $P$:}
    This $O(\sqrt n\log n)$-qubit register is initialized in $\ket{0}$ and
    will hold the product $p(c)$
    to the accumulator $Y$. It is uncomputed, reset to $\ket{0}$ before $C$ is rotated and measured.
\end{itemize}

\vspace{3mm}
For $p\in\ZZ_N^*$, let $\operatorname{Mul}^{(X)}_p$ and
$\operatorname{Mul}^{(Y)}_p$ denote multiplication by $p$ in registers $X$
and $Y$, respectively. 
We first notice how multiplication of the second accumulator behaves when sandwiched by Fibonacci updates:

\begin{lemma}[Fibonacci conjugation identity]
\label{lem:fibonacci-conjugation-identity}
For every $p\in\ZZ_N^*$ and integer $r\geq 0$,
\begin{equation}
    \mathsf{Fib}^{r}
    \operatorname{Mul}^{(Y)}_p
    \mathsf{Fib}^{-r}
    =
    \operatorname{Mul}^{(X)}_{p^{F_r}}
    \operatorname{Mul}^{(Y)}_{p^{F_{r+1}}}.
    \label{eq:fibonacci-conjugation}
\end{equation}
\end{lemma}

\begin{proof}
The claim follows by induction on $r$.
For $r=0$, 
$F_0=0$ and $F_1=1$, so \refeq{eq:fibonacci-conjugation} becomes $\operatorname{Mul}^{(Y)}_p = \operatorname{Mul}^{(X)}_{1}\operatorname{Mul}^{(Y)}_{p}.$
More generally, if the exponent pair after $r$ conjugations is $(F_r,F_{r+1})$, 
then the $r+1$ statement takes the form
\begin{equation}
    \mathsf{Fib}^{r+1}
    \operatorname{Mul}^{(Y)}_p
    \mathsf{Fib}^{-(r+1)}
    =
    \mathsf{Fib} \operatorname{Mul}^{(X)}_{p^{F_r}}
    \operatorname{Mul}^{(Y)}_{p^{F_{r+1}}} 
    \mathsf{Fib}^{-1}
    =
    \operatorname{Mul}^{(X)}_{p^{F_{r+1}}}
    \operatorname{Mul}^{(Y)}_{p^{F_{r+2}}},
\end{equation}
where the final line follows from the definition of the $\mathsf{Fib}$ accumulator update. If the exponent pair after $r$ conjugations is
$(\alpha,\beta)$, then one further conjugation replaces it by $(\beta, \alpha+\beta).$
\end{proof}

\begin{lemma}[Streamed Fibonacci phase kickback]
\label{lem:streamed-fibonacci-phase-kickback}
For every $k\geq2$ and $c\in\{0,1\}^d$,
\begin{equation}
    \operatorname{Mul}^{(Y)}_{p(c)}
    \ket{\widetilde{u_\theta}^{(k)}}
    =
    e^{2\pi iF_k\langle\theta,c\rangle}
    \ket{\widetilde{u_\theta}^{(k)}}.
    \label{eq:streamed-fibonacci-phase-kickback}
\end{equation}
\end{lemma}

\begin{proof}
Applying \reflem{lem:fibonacci-conjugation-identity} with $r=k-1$
gives
$\mathsf{Fib}^{k-1}
    \operatorname{Mul}^{(Y)}_{p(c)}
    \mathsf{Fib}^{-(k-1)}
    =
    \operatorname{Mul}^{(X)}_{p(c)^{F_{k-1}}}
    \operatorname{Mul}^{(Y)}_{p(c)^{F_k}}.$
We can, therefore, expand the multiplication on the traveling eigenstate as
follows:
\begin{align}
    \operatorname{Mul}^{(Y)}_{p(c)}
    \ket{\widetilde{u_\theta}^{(k)}}
    &=
    \operatorname{Mul}^{(Y)}_{p(c)}
    \mathsf{Fib}^{-(k-1)}
    \bigl(\ket{\mathsf M_N}_X\ket{u_\theta}_Y\bigr)
    \notag\\
    &=
    \mathsf{Fib}^{-(k-1)}
    \left(
        \mathsf{Fib}^{k-1}
        \operatorname{Mul}^{(Y)}_{p(c)}
        \mathsf{Fib}^{-(k-1)}
    \right)
    \bigl(\ket{\mathsf M_N}_X\ket{u_\theta}_Y\bigr)
    \notag\\
    &=
    \mathsf{Fib}^{-(k-1)}
    \left(
        \operatorname{Mul}^{(X)}_{p(c)^{F_{k-1}}}
        \ket{\mathsf M_N}_X
        \otimes
        \operatorname{Mul}^{(Y)}_{p(c)^{F_k}}
        \ket{u_\theta}_Y
    \right).
    \notag
\intertext{The first multiplication leaves the mask $\ket{\mathsf M_N}$ unchanged by definition 
because $p(c)^{F_{k-1}}\in\ZZ_N^*$. The second multiplication is on register $Y$, and can be decomposed as $\operatorname{Mul}^{(Y)}_{p(c)^{F_k}}
    =
    \prod_{i=1}^d U_i^{F_kc_i}$. 
By \refeq{eq:regev-joint-eigenvalue}, this unitary will 
act on $\ket{u_\theta}_Y$ with eigenvalue
$e^{2\pi iF_k\langle\theta,c\rangle}$. }
    \operatorname{Mul}^{(Y)}_{p(c)}
    \ket{\widetilde{u_\theta}^{(k)}}
    &=
    e^{2\pi iF_k\langle\theta,c\rangle}
    \mathsf{Fib}^{-(k-1)}
    \bigl(\ket{\mathsf M_N}_X\ket{u_\theta}_Y\bigr)
    \notag\\
    &=
    e^{2\pi iF_k\langle\theta,c\rangle}
    \ket{\widetilde{u_\theta}^{(k)}}.
    \notag
\end{align}
As a result, the claim follows after applying $\mathsf{Fib}^{-(k-1)}$ and using
\refeq{eq:moving-fibonacci-eigenstate}.
\end{proof}

\reflem{lem:streamed-fibonacci-phase-kickback} shows that the clean
multiplication by $p(c)$ produces exactly the desired level-$k$ phase,
while leaving the traveling eigenstate unchanged. Since the computation of
$p(c)$ is subsequently reversed, the product register $P$ and its
workspace return to $\ket{0}$. Hence each controlled level $k$ probe performs the correctness target operation of \refeqplain{eq:desired-streamed-control-state}:
\begin{equation}
    \ket{+}^{\otimes d}\ket{\widetilde{u_\theta}^{(k)}}
    \longmapsto
    \left(
        \bigotimes_{i=1}^d
        \frac{
            \ket{0}
            +e^{2\pi iF_k\theta_i}\ket{1}
        }{\sqrt{2}}
    \right)
    \ket{\widetilde{u_\theta}^{(k)}}.
    \label{eq:masked-simultaneous-phase-probe}
\end{equation}
The controls acquire the same phases as under
$U_1^{F_k},\ldots,U_d^{F_k}$. It remains to show that this invariant persists through every
probe and level, and that averaging over the initial eigenstate
superposition gives the same measurement distribution as ideal simultaneous
phase estimation.

\begin{proposition}[Correctness of the streamed circuit]
\label{prop:streamed-regev-correctness}
Conditioned on a fixed $\theta\in\mathcal{L}^*/\ZZ^d$, the following hold:
\begin{enumerate}
    \item Every probe at level $k$ in
    \refalg{alg:streamed-regev-run} produces the control state in
    \refeq{eq:masked-simultaneous-phase-probe}.

    \item After every probe at level $k$, the persistent accumulators remain
    in $\ket{\widetilde{u_\theta}^{(k)}}$.

    \item After every probe, the product register $P$ and all other
    nonpersistent arithmetic-workspace registers have returned to $\ket{0}$.
\end{enumerate}
Without conditioning on $\theta$, every possible measurement transcript
$\tau$ has the same
distribution as if we chose $\ket{u_\theta}$ for
$\theta\in\mathcal{L}^*/\ZZ^d$ uniformly, directly applied the controlled
unitaries $U_i^{F_k}$, and performed the same adaptive rotations and
measurements.
\end{proposition}

\begin{proof}
By \refeq{eq:one-as-regev-eigenstates}, the initial accumulator
state decomposes as
\begin{equation}
    \ket{\mathsf M_N}_X\ket{1}_Y
    =
    \frac{1}{\sqrt{|\mathcal{L}^*/\ZZ^d|}}
    \sum_{\theta\in\mathcal{L}^*/\ZZ^d}
    \ket{\mathsf M_N}_X\ket{u_\theta}_Y.
\end{equation}
Fix one value of $\theta$. Before the first backward update, the
accumulators are in
    $\ket{\mathsf M_N}_X\ket{u_\theta}_Y
    =
    \ket{\widetilde{u_\theta}^{(1)}}$
and all workspace ancillas are in $\ket{0}$.
The backward update at the beginning of level $2$ places them in
$\ket{\widetilde{u_\theta}^{(2)}}$, and \reflem{lem:streamed-fibonacci-phase-kickback}
shows that every level-$2$ probe produces the required phase while
leaving this accumulator state unchanged. The clean interface in
\refeq{eq:clean-same-order-accumulator-update} also returns $P$ and
its workspace to $\ket{0}$.

Now assume that, after all probes at level $k$, the accumulators remain in
$\ket{\widetilde{u_\theta}^{(k)}}$. The next backward update maps
    $\ket{\widetilde{u_\theta}^{(k)}}
    \longmapsto
    \ket{\widetilde{u_\theta}^{(k+1)}}$.
For each of the subsequent $M_{k+1}$ probes, the clean controlled
multiplication and
\reflem{lem:streamed-fibonacci-phase-kickback}, now applied at level
$k+1$, give
\begin{equation}
    \frac{1}{\sqrt{2^d}}
    \sum_{c\in\{0,1\}^d}
    \ket{c}_C
    \ket{\widetilde{u_\theta}^{(k+1)}}_{XY}
    \longmapsto
    \frac{1}{\sqrt{2^d}}
    \sum_{c\in\{0,1\}^d}
    e^{2\pi iF_{k+1}\langle\theta,c\rangle}
    \ket{c}_C
    \ket{\widetilde{u_\theta}^{(k+1)}}_{XY}.
\end{equation}
This is the factorized control state in
\refeq{eq:masked-simultaneous-phase-probe} at level $k+1$.
Because the controls are unentangled from the accumulators, their adaptive
rotation and measurement do not disturb the traveling eigenstate. The same
argument therefore applies to all $M_{k+1}$ probes, completing the induction.

Finally, traveling eigenstates associated with distinct values of $\theta$
remain orthogonal because the same unitary
$\mathsf{Fib}^{-(k-1)}$ is applied to each of them. Consequently, tracing
out the accumulators removes all cross terms between distinct $\theta$.
The probability of any transcript $\tau$ is therefore the uniform average
of its probabilities conditioned on $\theta$, completing the claim.
\end{proof}

\vspace{3mm}
\noindent
\textbf{Efficiency. } 
We now turn our attention to establishing the desired efficiency guarantees:

\begin{lemma}[Efficiency of the streamed circuit]
\label{lem:streamed-regev-efficiency}
\refalg{alg:streamed-regev-run} achieves the efficiency guarantees
claimed in \refthm{thm:streamed-phase-probe-algorithm}.
\end{lemma}
\begin{proof}
The mask is prepared once, contributing $G_{\mathrm{mask}}$ gates and depth
$D_{\mathrm{mask}}$. The algorithm then performs a backward Fibonacci
update at each level $k\in\{2,\ldots,K\}$, for a total of $K-1$
big-by-big updates.  Across the levels, there are
$M=\sum_{k=2}^K M_k$ phase probes. During each probe the selected-product tree is computed, the
resulting small value is multiplied into $Y$, and the tree is reversed.
Thus, each probe contributes
\begin{equation}
    2G_{\mathrm{tree}}+G_{\mathrm{sb}}
    \qquad\text{and}\qquad
    2D_{\mathrm{tree}}+D_{\mathrm{sb}}
\end{equation}
gates and depth, respectively. Preparing, rotating, measuring, and resetting
the $d$ controls requires $O(d)$ operations per probe. These operations act
in parallel across the controls and therefore contribute $O(dM)$ gates but
only $O(M)$ depth. Consequently,
\begin{align}
    G_{\mathrm{stream}}
    &=
    G_{\mathrm{mask}}
    +(K-1)G_{\mathrm{bb}}
    +M\bigl(2G_{\mathrm{tree}}+G_{\mathrm{sb}}\bigr)
    +O(dM),
    \label{eq:streamed-efficiency-gate-count}\\
    D_{\mathrm{stream}}
    &=
    D_{\mathrm{mask}}
    +(K-1)D_{\mathrm{bb}}
    +M\bigl(2D_{\mathrm{tree}}+D_{\mathrm{sb}}\bigr)
    +O(M).
    \label{eq:streamed-efficiency-depth}
\end{align}

It remains to bound the peak space. Mask preparation uses
$S_{\mathrm{mask}}$ qubits (assume prep before initializing $\ket{1}$ state). A backward Fibonacci update uses
$S_{\mathrm{bb}}$ qubits. Recall that
\reflem{lem:clean-same-order-accumulator-update} does resource accounting for the complete, clean same-order update of a probe. During the product-tree portion,
$S_{\mathrm{tree}}$ includes $C$, $P$, and the tree workspace; the two
$n$-qubit accumulators contribute an additional $2n$ qubits. During the
small-by-big multiplication, $S_{\mathrm{sb}}$ includes $P$, $Y$, and the
multiplication workspace, while $X$ and $C$ contribute another $n+d$
qubits. 
Since these stages occur sequentially, their workspace can be
reused. Therefore,
\begin{equation}
    S_{\mathrm{stream}}
    =
    \max\!\left\{
        S_{\mathrm{mask}},
        S_{\mathrm{bb}},
        2n+S_{\mathrm{tree}},
        n+d+S_{\mathrm{sb}}
    \right\},
    \label{eq:streamed-efficiency-qubit-count}
\end{equation}
and we proved the bounds claimed in
\refthm{thm:streamed-phase-probe-algorithm}.
\end{proof}

\subsection{Implications}

\subsubsection{Importing efficient primitives for in-place modular multiplication}
\label{sec:implementing-multiplication-unitaries}

The multiplication unitaries underlying Regev's algorithm determine the
gate, space, and depth costs of our streamed construction. For a multiplier
$x\in\ZZ_N^*$ and any $y\in\ZZ_N$, the operation we require is in-place multiplication modulo $N$:
\begin{equation}
    \ket{x}\ket{y}
    \longmapsto
    \ket{x}\ket{xy\bmod N}.
    \label{eq:dialog-in-place-multiplication-interface}
\end{equation}
We implement it using recent in-place modular
multiplication techniques developed in the elliptic-curve literature
\cite{BabbushEtAl26,Schrottenloher26,LuoEtAl26}. The techniques of Babbush \textit{et al.} ~\cite{babbush2018encoding} and Schrottenloher~\cite{ Schrottenloher26} temporarily replace
$x$ by its \textit{dialog}: the transcript of the branch decisions made while
running the Euclidean algorithm on $(N,x)$. This
transcript is then replayed to multiply $x$ into $y$, in-place. 
The technique of Luo \textit{et al.} \cite{LuoEtAl26} is similar to the custom-built in-place modular multiplication primitives built by Ragavan and Vaikuntanathan \cite{rv24, cryptoeprint:2024/636}; it is built from multiple out-of-place multiplication steps. The difference is that Ref~\cite{LuoEtAl26} utilizes measurement-based uncomputation -- uncomputing by measuring in the $X$ basis -- in building up this low-space primitive.

Here, we discuss what is required to plug both schemes into our algorithm. 
For the measurement-based uncomputation scheme of Ref~\cite{LuoEtAl26}, we review their construction at a high level and why it can be imported to our setting ``out-of-the-box.'' 
In contrast, we will need some additional machinery to import the dialog multiplication of Refs~\cite{BabbushEtAl26,Schrottenloher26}. We discuss this below, and then \hyperref[app:dialog-bounds-worst-case]
{Appendix~\ref*{app:dialog-bounds-worst-case}} provides the guarantees necessary to utilize the primitive in our streamed Regev setting.

\paragraph{Dialog schemes. \cite{BabbushEtAl26, Schrottenloher26}}

Dialog multiplication turns the information produced by the Euclidean
algorithm into a program for multiplying in-place.  Standard
constructions for in-place multiplication pose in-place
multiplication as several out-of-place arithmetic operations, which can utilize a large
workspace.  
In contrast, dialog multiplication realizes  in-place
 multiplication directly, by using the Euclidean trace to temporarily change how the multiplier $x$ is
represented. 
Although this approach underlies the multiplication used in the
elliptic-curve results of Babbush \textit{et
al}.~\cite{BabbushEtAl26} and Schrottenloher~\cite{Schrottenloher26}, the former team does not provide its underlying
arithmetic circuits. So, we only discuss the construction of
Schrottenloher~\cite{Schrottenloher26}.

The resource-optimized multiplier of
Ref.~\cite{Schrottenloher26} uses
$4.355n+O(\sqrt n)$
qubits. 
This is an empirically supported estimate for random inputs: the
circuit allocates enough iterations and operand space to accommodate \textit{most}
Euclidean trajectories observed experimentally. However, while the allocated space is plenty for almost all inputs, some inputs will overflow.
In our streamed Regev construction, consider the multiplication required for our backward Fibonacci update. The multipliers are successive components of a traveling eigenstate, obtained
recursively from earlier accumulator values. In other words, they are correlated
across levels and, therefore, are not fresh random inputs. Moreover, each multiplier may
be in superposition due to our use of masking, and so the circuit must act correctly on basically every superposition branch. As a result, we cannot immediately import this $4.355n+O(\sqrt n)$ space result into
our streamed Regev construction.

In \hyperref[app:dialog-bounds-worst-case]
{Appendix~\ref*{app:dialog-bounds-worst-case}} we give background on the dialog method of Ref.~\cite{Schrottenloher26} and provide a
worst-case space allocation that accommodates every promised multiplier
$x\in\ZZ_N^*$. It
implements the multiplication correctly, over the full support of our
accumulator state, and uses $S_{\mathrm{bb}}
=\frac{16}{3}n+O(1)=5.33n+O(1)$
qubits. Compared with Schrottenloher's empirical estimate, this worst-case space requires almost one additional (n)-qubit register. 
Finally, note that the backward Fibonacci update does not require a separate
multiplication primitive. We can simply run the circuit in reverse to implement the desired map
\begin{equation}
\ket{x}\ket{y}
\longmapsto
\ket{x}\ket{x^{-1}y\bmod N},
\end{equation}
with the same gate, space, and depth costs. Swapping the two output registers
gives exactly the Fibonacci accumulator update used by our streamed construction.

\paragraph{Measurement-based uncomputation scheme of Luo \textit{et al.} \cite{LuoEtAl26}}

Luo et al.~\cite{LuoEtAl26} give a measurement-assisted implementation of
in-place modular multiplication using $3n + O(\log n)$ space, at the expense of $O(n^2)$ toffolis. Their
technique also implements the inverse logical operation, which is exactly the backward Fibonacci update $\ket{x}_X\ket{y}_Y
    \longmapsto
    \ket{x^{-1}y\bmod N}_X\ket{x}_Y$,
with the same cost parameters. Their construction follows the steps below; we write out the inverse logical operation because this is the one we actually use.
\begin{align}
    \ket{x}_X\ket{y}_Y\ket{0}_A
    &\longrightarrow
    \ket{x^{-1}}_X\ket{y}_Y\ket{0}_A \\
    &\longrightarrow
    \ket{x^{-1}}_X\ket{y}_Y\ket{x^{-1}y\bmod p}_A\\
    &\longrightarrow
    e^{i\phi(y)}\ket{x^{-1}}_X\ket{0}_Y\ket{x^{-1}y\bmod p}_A \hspace{5mm} &\text{[unknown phase $\phi(y)$ acquired]}\\
    &\longrightarrow
    e^{i\phi(y)} \ket{x}_X\ket{0}_Y\ket{x^{-1}y\bmod p}_A \\
    &\longrightarrow
    \ket{x}_X\ket{y}_Y\ket{x^{-1}y\bmod p}_A \hspace{5mm} &\text{[phase cleaned up using $y$]}
    \\
    &
    \longrightarrow
    \ket{x}_X\ket{0}_Y\ket{x^{-1}y\bmod p}_A.
\end{align}
They compute temporary products in the clean register $A$, utilize measurement-based uncomputation to reset the $Y$
register to $\ket{0}$, and eventually reconstruct its value to correct the
measurement-dependent phase before returning ancillary space to $\ket{0}$.
Unlike the probabilistic arithmetic construction of
Schrottenloher~\cite{Schrottenloher26}, the circuit of
Luo et al.~\cite{LuoEtAl26} is exact on every promised input
$x\in\ZZ_N^*$.

We can immediately extend this construction because it works exactly on all inputs. The only detail that needs to be checked, when applying this construction to our Regev accumulator updates, is that this circuit holds for a non-prime modulus $N$. 
Luo \textit{et al.} state their construction for a prime modulus $p$.
Luckily, their
correctness argument uses primality only to ensure that the multiplier $x$ is
invertible: 
the EEA applies whenever $\gcd(x,N)=1$. Their remaining steps
use only reversible modular arithmetic
\cite[Sections~3 and~5, Equations~(4)--(7)]{LuoEtAl26}. 
Therefore,
the same construction
and resource bounds apply for an odd composite modulus \(N\) on
the promised subspace \(x\in\ZZ_N^*\). In both (1) our backwards Fibonacci accumulator update and (2) the selected-product update $p(c)=\prod_i a_i^{c_i}$ into one accumulator state, the multiplier lies in \(\ZZ_N^*\).

\subsubsection{Regev factoring in lowest-yet space}

We now combine the correctness of our streamed phase probes with the
phase-estimation guarantees  of
\refsec{sec:iterative-fibonacci-phase-estimation} and Regev's lattice-postprocessing guarantees \cite{Regev23}.
We first show that we can factor, using our streaming circuit, with constant probability. 
Moreover, in this statement we use the approximate mask $\ket{R_N}$ in place of $\ket{\mathsf M_N}$.
Then, we specify the concrete costs of our streaming circuit. We give a couple different costs by plugging in different primitives for in-place modular 
multiplication.

\begin{corollary}[Correctness of streamed Regev factoring]
\label{cor:streamed-regev-factoring-correctness}
Assume \refconj{conj:short-nontrivial-relation}. Let
$p_{\mathrm{fail}}$ be any polynomial, and choose $K$ and
$M_{\mathrm{lev}}$ according to
\refeqs{eq:regev-final-fibonacci-level}
and~\eqref{eq:regev-per-level-fibonacci-probes} respectively, with
$p=p_{\mathrm{fail}}$. For $k \in \{2,3,...,K\}$, let $M_k = M_{\mathrm{lev}}$. 
Execute \refalg{alg:streamed-regev-run} independently $d+4$ times,
using $\ket{R_N}$ in place of $\ket{\mathsf M_N}$ in each execution.
Regev's classical postprocessing then recovers a nontrivial factor of $N$
with probability at least
\begin{equation}
    \frac{1}{4}
    -\frac{1}{p_{\mathrm{fail}}(d)}
    -\delta_{\mathrm{mask}},
    \qquad \text{where } 
    \delta_{\mathrm{mask}}
    \coloneqq
    (d+4)
    \sqrt{1-\frac{\varphi(N)}{N}}.
    \label{eq:streamed-regev-factoring-success}
\end{equation}
For balanced semiprimes,
$\delta_{\mathrm{mask}}=2^{-\Omega(n)}$, and therefore the complete procedure
recovers a nontrivial factor of $N$ with constant probability.
\end{corollary}

\begin{proof}
First suppose that every execution begins with the ideal mask
$\ket{\mathsf M_N}$. \refprop{prop:streamed-regev-correctness}
shows that each streamed run has the same measurement-transcript
distribution as the corresponding sequence of ideal Fibonacci phase
probes. \refcor{cor:fibonacci-phase-estimation-for-regev} then
shows that all $d+4$ estimated vectors satisfy Regev's precision
requirement simultaneously,
except with probability at most
$1/p_{\mathrm{fail}}(d)$. By
\refthm{thm:regev-noisy-sample-recovery}, Regev's classical
postprocessing recovers a nontrivial factor with probability at least
$1/4$. Thus the exact-mask procedure succeeds with probability at least
\begin{equation}
    \frac{1}{4}
    -\frac{1}{p_{\mathrm{fail}}(d)}.
\end{equation}

Now consider using the approximate mask $\ket{R_N}$ in every streamed run.
Let $\widetilde P_{\mathsf M}$ denote the joint distribution of all $d+4$
transcripts when every run uses $\ket{\mathsf M_N}$, and let
$\widetilde P_R$ denote the corresponding joint distribution when every run
uses $\ket{R_N}$.
Here we use tildes to denote joint distributions across the streamed runs. 
We will show that total variation
distance between these distributions is bounded by $\delta_{\mathrm{mask}}$. 
To compare the joint distributions, replace the exact mask by
$\ket{R_N}$ one run at a time. 
\refprop{prop:approximate-mask-error} shows that each replacement
$\ket{\mathsf M_N}$ by $\ket{R_N}$ changes the transcript distribution
of a streamed run by total variation distance at most $\sqrt{1-\frac{\varphi(N)}{N}}$, while leaving the other runs unchanged. Recall that $\varphi(N)$ counts the number of residues invertible modulo $N$.
The triangle inequality over the
$d+4$ replacements gives
\begin{equation}
    d_{\mathrm{TV}}
\!\left(
    \widetilde P_R,
    \widetilde P_{\mathsf M}
\right)
    \leq
    (d+4)
    \sqrt{1-\frac{\varphi(N)}{N}}
    =
    \delta_{\mathrm{mask}}.
\end{equation}
Finally, the classical phase estimator and Regev's lattice postprocessing
are further operations on these transcripts, so the data-processing
inequality implies that their output distributions differ by at most the
same amount. In particular, the probability of recovering a nontrivial factor
can decrease by at most $\delta_{\mathrm{mask}}$, proving
\refeq{eq:streamed-regev-factoring-success}.

For a balanced semiprime $N=pq$, both factors have
$n/2+O(1)$ bits, so $p,q=2^{n/2+O(1)}$. Using the fact that ${\varphi(N)}
    =
    {(p-1)(q-1)}$, we can plug this into our distance bound. Simple accounting of the residues modulo $N$ gives
\begin{align}
    \sqrt{1-\frac{\varphi(N)}{N}}
    &=
    \sqrt{1-\frac{(p-1)(q-1)}{pq}} \\
    &=
    \sqrt{
    \frac{1}{p}
    +\frac{1}{q}
    -\frac{1}{pq}} \\
    &=
    O(2^{-n/4}).
\end{align}
Since $d=\floor{\sqrt n}$, multiplication by $d+4$ changes this quantity
only by a polynomial factor, and therefore
$\delta_{\mathrm{mask}}=2^{-\Omega(n)}$.
\end{proof}

\begin{corollary}[Space and multiplication count per Regev run]
\label{cor:streamed-regev-resource-bounds}
For an execution of \refalg{alg:streamed-regev-run} using the analytic phase-estimation parameters of
\refcor{cor:fibonacci-phase-estimation-for-regev}, the algorithm invokes an $n$-bit in-place modular multiplication primitive $O(\sqrt n\log n)$ times. 

The peak space is determined by the multiplication primitive used. If the space for that primitive is $M(n)$, then the algorithm requires $n + M(n) + O(\sqrt n)$ qubits.
Using the different primitives below, we obtain constant factor estimates:
\begin{enumerate}
    \item Exact Dialog multiplier of
    \refapp{app:dialog-bounds-worst-case}: $\frac{19}{3}n+O(\sqrt n)
        =
        6.33... n+O(\sqrt n) \text{ qubits.}$
    \item Measurement-based uncomputation multiplier of
    Luo et al.~\cite{LuoEtAl26}: $4n+O(\sqrt n) \text{ qubits.}$
\end{enumerate}
\end{corollary}

\begin{proof} 
Each execution of \refalg{alg:streamed-regev-run} invokes an $n$-bit in-place modular-multiplication primitive $\sum_{k=2}^{K}M_k+(K-1)$
times: once for each of the phase probes and once for each of the
$K-1$ backward Fibonacci updates. With the analytic phase-estimation
parameters of \refcor{cor:fibonacci-phase-estimation-for-regev}, each $M_k = O(\log n)$ and $K = \Theta(\sqrt n)$. Putting this together gives $O(\sqrt n\log n)$ in-place modular-multiplications. 
Note that the
product-tree arithmetic was excluded from this count because all of its
operands have $o(n)$ bits.

For space, \refthm{thm:streamed-phase-probe-algorithm} gives total qubit count in terms of the space required to prepare the mask $S_{\mathrm{mask}}$, the tree $S_{\mathrm{tree}}$, the Fibonacci accumulator update $S_{\mathrm{bb}}$, and product of primes into the traveling eigenstate $S_{\mathrm{sb}}$:  
\begin{equation}
    S_{\mathrm{stream}}
    =
    \max\!\left\{
        S_{\mathrm{mask}},
        S_{\mathrm{bb}},
        2n+S_{\mathrm{tree}},
        n+d+S_{\mathrm{sb}}
    \right\}.
\end{equation}
Here $S_{\mathrm{mask}}=2n+O(1)$ by
\reflem{lem:uniform-residue-mask-preparation},
$S_{\mathrm{tree}}=\widetilde O(d)=o(n)$ by
\reflem{lem:selected-product-computation}, and
$d=\floor{\sqrt n}$. 
The
small-by-big update (``sb'') can be implemented using the same $n$-bit primitive
as the big-by-big update (``bb''). 
We may therefore take
$S_{\mathrm{sb}}=S_{\mathrm{bb}}$. Plugging in all these separate space costs, we have 
\begin{equation}
    S_{\mathrm{stream}}
    =
    \max\!\left\{
        2n + O(1),
        S_{\mathrm{bb}},
        2n+o(n),
        n+\floor{\sqrt n}+S_{\mathrm{bb}}
    \right\}.
\end{equation}
$S_\mathrm{bb}$ is the space required for an $n$-bit in-place modular multiplication.
Therefore, 
the final item in the list sets the dominant space cost of the algorithm. This is  $n + S_\mathrm{bb} + O(\sqrt n)$ qubits.

We can now consider both in-place modular multiplication primitives. 
For the exact Dialog primitive,
$S_{\mathrm{bb}}=16n/3+O(1)$, and therefore we need $\frac{19}{3}n+O(\sqrt n)$ qubits. 
For the construction of Luo et al.,
$S_{\mathrm{bb}}=3n+O(\log n)$, and we need $  4n+O(\sqrt n)$ qubits.
\end{proof}

\paragraph{Empirically improved multiplication count with Bayesian Fibonacci IQPE protocol.}
Our numerical experiments
find that the adaptive Bayesian estimator uses approximately six
probes per Fibonacci level, for the cryptographically relevant problem sizes we simulate.
We look at $n=4096, 2048, \text{ and }1024.$ 
This is substantially better than the analytic phase estimator of
\refsec{sec:fibonacci-estimator-analysis}. 
As a heuristic, if approximately six simultaneous probes
suffice at every Fibonacci level, then including the one backward Fibonacci
update per level gives $30.25\sqrt n$
in-place modular multiplications per Regev run. 
This comes from needing 
$\sum_{k=2}^{K}M_k+(K-1)$ in-place modular multiplications for our algorithm. Here, we set $M_k = 6$ and so we need $7(K-1)$ multiplications. For the $R=2^{3\sqrt n}$ choice of previous work \cite{rv24, cryptoeprint:2024/636}, the largest required Fibonacci
level satisfies
\begin{equation}
    K
    =
    \log_\phi R+O(1)
    \approx
    4.321\sqrt n+O(1).
\end{equation}
Therefore, all together we need $7(K-1)
    \approx
    30.25\sqrt n$
$n$-bit in-place modular multiplications per Regev run.

\vspace{3mm}
\paragraph{Acknowledgments.}
The authors would like to thank Oded Regev, Ryan Babbush, Martin Ekerå, Seyoon Ragavan, Greg Meyer, Thiago Bergamaschi, Tommy Schuster, and Noah Shutty for helpful discussions and feedback on this work. 
KVK gratefully acknowledges the hospitality and support of Google Quantum AI, where this work was completed.

\paragraph{AI Acknowledgment.}
The core ideas in this paper are human generated, and all content in this paper is human verified.
AI assistance was used when writing some of the numerical code and when writing some of the manuscript.

\ifllncs
\bibliographystyle{splncs04}
\bibliography{main}
\else
\bibliographystyle{alpha-doi}
\bibliography{main}
\fi

\appendix

\section{Proof of \refthm{thm:slow-refinement-fibonacci-phase-estimation}: \(O(\sqrt n\log n)\) phase probes per Regev run}\label{app:asymptotic_factoring_proof}

Our proof proceeds by construction. 
We first describe the estimator and then state it formally. 
The construction has a base and inductive step. We first locate $\theta$ on the full phase
circle using two measurement quadratures. We then move through the Fibonacci
levels, using the estimate from level $k-1$ to isolate out and then reduce the remaining
phase error at level $k$. At level $k$ we take enough measurements to guarantee (with high probability) that we have reduced the phase error enough that the wrap around produced by transitioning $F_k \rightarrow F_{k+1}$ does not produce ambiguity in the phase.

\paragraph{Simple phase estimator construction for proof.}
Let $M_{\mathrm{lev}}$ be number of phase probes at each level, for all levels except level $2$. At this initial level, we will take $2M_{\mathrm{lev}}$.
At the initial level $F_2=1$, we take $M_{\mathrm{lev}}$ probes with rotation angle $0$ and
$M_{\mathrm{lev}}$ probes with rotation angle $-\pi/2$. We label their respective outcomes as $b^{(C)}_{2,s}\in\{0,1\}$ and $b^{(S)}_{2,s}\in\{0,1\}$, 
for $s\in[M_{\mathrm{lev}}]$. Following \refeq{eq:fibonacci-probe-likelihood}, we can use these outcomes to estimate
$\cos(2\pi\theta)$ and $\sin(2\pi\theta)$, respectively. 
\begin{equation}
    \widehat C_2
    \coloneqq
    \frac{1}{M_{\mathrm{lev}}}\sum_{s=1}^{M_{\mathrm{lev}}}(-1)^{b^{(C)}_{2,s}},
    \qquad
    \widehat S_2
    \coloneqq
    \frac{1}{M_{\mathrm{lev}}}\sum_{s=1}^{M_{\mathrm{lev}}}(-1)^{b^{(S)}_{2,s}}.
    \label{eq:fibonacci-bootstrap-statistics}
\end{equation}
Having both
quadratures removes the reflection ambiguity that would remain from either
one alone. We therefore can complete the base case by initializing
\begin{equation}
    \widehat\theta_2
    \coloneqq
    \frac{1}{2\pi}
    \operatorname{atan2}(\widehat S_2,\widehat C_2)
    \pmod 1.
    \label{eq:fibonacci-bootstrap-estimate}
\end{equation}

Now we can consider a generic Fibonacci level. Assume that we are at level $k\geq 3$, and assume that we have obtained
$\widehat\theta_{k-1}$. At Fibonacci level $k\geq 3$, we choose the same rotation $\phi_k = \phi_{k,s}$ for every probe and
\begin{equation}
    \phi_k
    \coloneqq
    2\pi\left(
        \frac14-F_k\widehat\theta_{k-1}
    \right)
    \pmod{2\pi}.
    \label{eq:slow-refinement-rotation}
\end{equation}
This rotation places the phase predicted by $\widehat\theta_{k-1}$ at the
quadrature point $\pi/2$, where the cosine in
\refeq{eq:fibonacci-probe-likelihood} has maximum slope. 
As a result, the
displacement from $\pi/2$ is exactly
$2\pi F_k(\theta-\widehat\theta_{k-1})$.
We take $M_{\mathrm{lev}}$ probes at level $k$ with this fixed $\phi_k$ and obtain outcomes $\{b_{k,s}\}_{s=1}^{M_{\mathrm{lev}}}$.
For the resulting outcomes, define the signed sample mean
\begin{equation}
    \overline Z_k
    \coloneqq
    \frac{1}{M_{\mathrm{lev}}}\sum_{s=1}^{M_{\mathrm{lev}}}(-1)^{b_{k,s}+1}.
    \label{eq:fibonacci-signed-sample-mean}
\end{equation}
The estimate $\widehat\theta_{k-1}$, and hence $\phi_k$, depends on the
outcomes from the earlier levels. Therefore, once $\widehat\theta_{k-1}$ and hence $\phi_k$ are fixed, 
these $M_{\mathrm{lev}}$ new probe outcomes are independent. They have a common expectation, which follows from \refeq{eq:fibonacci-probe-likelihood}:
\begin{equation}
    \mathbb E[\overline Z_k]
    =
    \sin\!\left(
        2\pi F_k(\theta-\widehat\theta_{k-1})
    \right).
    \label{eq:fibonacci-residual-expectation}
\end{equation}
Thus $\arcsin(\overline Z_k)$ estimates the residual phase left after
subtracting $\widehat\theta_{k-1}$. Dividing this residual by $F_k$ gives the
level-$k$ correction:
\begin{equation}
    \widehat\theta_k
    \coloneqq
    \widehat\theta_{k-1}
    +
    \frac{1}{2\pi F_k}\arcsin(\overline Z_k)
    \pmod 1.
    \label{eq:slow-refinement-phase-update}
\end{equation}
\refalg{alg:rigorous-slow-refinement-fibonacci-estimator} collects
these steps.

\vspace{5mm}
\begin{algorithm}[H]
    \KwData{A final Fibonacci level $K\geq2$, a common probe budget
    $M_{\mathrm{lev}}\geq1$, and access to the scalar probes in
    \refeq{eq:fibonacci-probe-likelihood}.}
    \KwResult{An estimate $w\in\RR/\ZZ$ of $\theta$.}
    \begin{enumerate}
        \item \textit{Base case ($k=2$)}: \hspace{3mm} 
        Take $M_{\mathrm{lev}}$ level-$2$ probes with rotation angle $0$, obtaining
        $b^{(C)}_{2,1},\ldots,b^{(C)}_{2,M_{\mathrm{lev}}}$, and
        take $M_{\mathrm{lev}}$ level-$2$ probes with rotation angle $-\pi/2$, obtaining
        $b^{(S)}_{2,1},\ldots,b^{(S)}_{2,M_{\mathrm{lev}}}$. 
        Using this data, compute $\widehat C_2$ and $\widehat S_2$, defined in \refeq{eq:fibonacci-bootstrap-statistics}. 
        Set
        \[
            \widehat\theta_2
            \gets
            \frac{1}{2\pi}
            \operatorname{atan2}(\widehat S_2,\widehat C_2)
            \pmod 1.
        \]

        \item \textit{Generic Fibonacci update}: \hspace{3mm}
        
        For
        $k=3,\ldots,K$:
        \begin{enumerate}
            \item Compute
            $\phi_k
            =
            \frac{\pi}{2}-2\pi F_k\widehat\theta_{k-1}
            \pmod{2\pi}.$
            \item Take $M_{\mathrm{lev}}$ level-$k$ probes with rotation angle $\phi_k$, obtaining
            $b_{k,1},\ldots,b_{k,M_{\mathrm{lev}}}$, and compute $\overline Z_k$ defined in \refeq{eq:fibonacci-signed-sample-mean}. Set 
            \[
                \widehat\theta_k
                \gets
                \widehat\theta_{k-1}
                +
                \frac{1}{2\pi F_k}\arcsin(\overline Z_k)
                \pmod 1.
            \]
        \end{enumerate}
        \item Return $w\coloneqq \widehat\theta_K$
    \end{enumerate}
    
    \caption{Rigorous slow-refinement Fibonacci phase estimation.}
    \label{alg:rigorous-slow-refinement-fibonacci-estimator}
\end{algorithm}

\vspace{3mm}

\paragraph{Formal Analysis.}
We now prove the accuracy and probe complexity of
\refalg{alg:rigorous-slow-refinement-fibonacci-estimator}.
Proving this performance guarantees boils down to bounding the intermediate estimates over the course of our algorithm.
This setup of \refalg{alg:rigorous-slow-refinement-fibonacci-estimator} already gives us a hint of how we want to bound the intermediate estimators $\widehat\theta_{k}$ in our performance guarantees below. Assume that the range of $\arcsin$ is $[-\pi/2,\pi/2]$, which we specify because sine is not one-to-one on the full
phase circle. As a result, we will want the residual angle $2\pi F_k(\theta-\widehat\theta_{k-1})$, which we estimate with $\arcsin$, to be in
$[-\pi/2,\pi/2]$. Crucially, on this interval the principal $\arcsin$ recovers the
residual without ambiguity. Taking the absolute value of the residual angle gives us an expression in terms of the torus distance, 
\begin{equation}
    2\pi F_k\operatorname{dist}_{\RR/\ZZ}
    (\widehat\theta_{k-1},\theta) \leq \pi/2,
\end{equation}
and so the precision guarantee we want to bound the probability of for all our estimators $\widehat\theta_k$ will be 
\begin{equation}
    \operatorname{dist}_{\RR/\ZZ}
    (\widehat\theta_k,\theta)
    \leq
    \frac{1}{8F_k}.
    \label{eq:slow-refinement-phase-invariant}
\end{equation}
Indeed, if this invariant holds at level $k$, then
\[
    2\pi F_{k+1}
    \operatorname{dist}_{\RR/\ZZ}
    (\widehat\theta_k,\theta)
    \leq
    2\pi\frac{F_{k+1}}{8F_k}
    \leq
    \frac{\pi}{2}.
\]
We obtain the last line because across all $k$, we have $F_k / F_{k-1} \leq 2$. 
Thus the residual at the next level remains within the principal branch of
$\arcsin$.
We can recover it without introducing
a new alias.

\vspace{3mm}
We will proceed by induction. 
The proof follows the levels of the algorithm. We first show that the
two-quadrature bootstrap establishes
\refeq{eq:slow-refinement-phase-invariant} at $k=2$. We then show
that, conditioned on the invariant holding at level $k-1$, one generic
Fibonacci update establishes it at level $k$ except with exponentially small
probability. A union bound over the levels gives the accuracy theorem for one
scalar phase.

\begin{lemma}[Base Case: Two-quadrature bootstrap]
\label{lem:two-quadrature-fibonacci-bootstrap}
The estimate $\widehat\theta_2$ produced by
\refalg{alg:rigorous-slow-refinement-fibonacci-estimator} satisfies
\begin{equation}
    \Pr\!\left[
        \operatorname{dist}_{\RR/\ZZ}
        (\widehat\theta_2,\theta)
        >
        \frac{1}{8}
    \right]
    \leq
    4e^{-M_{\mathrm{lev}}/8}.
    \label{eq:bootstrap-phase-failure-probability}
\end{equation}
\end{lemma}

\begin{proof}
Consider the unit vector $v
    \coloneqq
    \bigl(\cos(2\pi\theta),\sin(2\pi\theta)\bigr),    $ and its estimator $\widehat v$, where $\widehat v
    \coloneqq
    (\widehat C_2,\widehat S_2)$.
\refeq{eq:fibonacci-probe-likelihood} gives that $\mathbb E[\widehat v]
    =
    v$. 
    By \refeq{eq:fibonacci-probe-likelihood}, the summands defining
the empirical means in
\refeq{eq:fibonacci-bootstrap-statistics} are independent,
$\{-1,1\}$-valued random variables with the stated expectations.
    Therefore, applying
Hoeffding's inequality gives
\begin{align}
    \Pr\!\left[
        \left|
            \widehat C_2-\cos(2\pi\theta)
        \right|>\frac12
    \right]
    \leq 2e^{-M_{\mathrm{lev}}/8}, \hspace{5mm}\text{and}\hspace{5mm}
    \notag
    \Pr\!\left[
        \left|
            \widehat S_2-\sin(2\pi\theta)
        \right|>\frac12
    \right]
    \leq 2e^{-M_{\mathrm{lev}}/8}.
    \label{eq:bootstrap-coordinate-hoeffding-bounds}
\end{align}
A union bound over both gives
\begin{equation}
    \Pr\!\left[
        \left|
            \widehat C_2-\cos(2\pi\theta)
        \right|>\frac12
        \ \hspace{5mm}\text{or}\hspace{5mm} \
        \left|
            \widehat S_2-\sin(2\pi\theta)
        \right|>\frac12
    \right]
    \leq
    4e^{-M_{\mathrm{lev}}/8}.
    \label{eq:bootstrap-coordinate-concentration}
\end{equation}
We now convert these coordinatewise concentration bounds into a bound on
the phase estimate. Suppose that neither deviation in
\refeq{eq:bootstrap-coordinate-concentration} occurs. Then
\begin{equation}
    \norm{\widehat v-v}_2
    \leq
    \frac{1}{\sqrt2}
    =
    \sin\!\left(\frac{\pi}{4}\right).
    \label{eq:bootstrap-angle_between_points}
\end{equation}
We want to turn this into a bound on torus distance between $\widehat\theta_2$ and $\theta$. We can do so by thinking about the angular separation between $\widehat v$ and $v$. If this angular separation is more than $\pi/4$, then the length of the vector connecting their endpoints $\norm{\widehat v-v}_2$ must be greater than $\sin (\pi/4)$. Therefore, following \refeqlowerplain{eq:bootstrap-angle_between_points}, the angle between $v$ and $\widehat v$ must be at
most $\pi/4$. The angles of $v$ and $\widehat v$ are $2\pi\theta$ and
$2\pi\widehat\theta_2$, respectively, and so 
\[
    2\pi
    \operatorname{dist}_{\RR/\ZZ}
    (\widehat\theta_2,\theta)
    \leq
    \frac{\pi}{4}.
\]
We can combine this with our union bound over the
two coordinate deviations in
\refeq{eq:bootstrap-coordinate-concentration}: failure of the desired phase bound implies that at least one of the
two coordinate deviations in
\refeq{eq:bootstrap-coordinate-concentration} occurred. Dividing by $2\pi$ then completes the claim.
\end{proof}

Notice that, while not stated in \reflem{lem:two-quadrature-fibonacci-bootstrap},
the proof also shows that $(\widehat C_2,\widehat S_2)\neq(0,0)$ whenever the bootstrap succeeds. This is important because $\operatorname{atan2}$ is undefined at $(0,0)$. It tells us that the claimed bound can fail only through the event in
\refeq{eq:bootstrap-coordinate-concentration}.

Next, we prove a small statement about sine separation, which we abstract out of our inductive step Lemma for reading ease.

\begin{lemma}[Uniform sine separation]
\label{lem:uniform-sine-separation}
For every $x,y\in[-\pi/2,\pi/2]$,
\begin{equation}
    |x-y|>\frac{\pi}{4}
    \quad\Longrightarrow\quad
    |\sin x-\sin y|>1-\frac{1}{\sqrt2}.
    \label{eq:uniform-sine-separation}
\end{equation}
\end{lemma}
\begin{proof}
Let $d
    \coloneqq
    |x-y|,$ and $m
    \coloneqq
    \frac{x+y}{2}.$
Because $x,y\in[-\pi/2,\pi/2]$, their midpoint $m$ satisfies $|m|
    \leq
    \frac{\pi}{2}-\frac{d}{2}.$
The trigonometric sum-to-product identity therefore gives
\begin{align}
    |\sin x-\sin y|
    &=
    2\sin\!\left(\frac d2\right)\cos(m)
    \notag\\
    &\geq
    2\sin^2\!\left(\frac d2\right)
    \notag\\
    &=
    1-\cos d.
    \label{eq:principal-sine-separation}
\end{align}
The first line uses the sum-to-product identity. The inequality follows
because
$\cos(m)=\cos|m|
\geq
\cos\!\left(\frac{\pi}{2}-\frac d2\right)
=
\sin\!\left(\frac d2\right),$
while the final line uses the half-angle identity
\(2\sin^2(d/2)=1-\cos d\).
If $d>\pi/4$, then the following proves the claim:
\[
    1-\cos d
    >
    1-\cos\!\left(\frac{\pi}{4}\right)
    =
    1-\frac{1}{\sqrt2}.
\]
\end{proof}

\begin{lemma}[Inductive step: Fibonacci refinement]
\label{lem:one-fibonacci-refinement-step}
Fix $k\geq3$ and any measurement transcript through level $k-1$ for which
\begin{equation}
    \operatorname{dist}_{\RR/\ZZ}
    (\widehat\theta_{k-1},\theta)
    \leq
    \frac{1}{8F_{k-1}}.
    \label{eq:previous-level-phase-invariant}
\end{equation}
Conditioned on this transcript, the next estimate satisfies
\begin{equation}
    \Pr\!\left[
        \operatorname{dist}_{\RR/\ZZ}
        (\widehat\theta_k,\theta)
        >
        \frac{1}{8F_k}
    \right]
    \leq
    2e^{-M_{\mathrm{lev}}/24}.
    \label{eq:one-level-phase-failure-probability}
\end{equation}
\end{lemma}

\begin{proof}
Fix a measurement transcript through level $k-1$ satisfying \refeqlowerplain{eq:previous-level-phase-invariant}. This fixes both $\widehat\theta_{k-1}$ and, therefore, also the rotation
$\phi_k$. All probabilities below are conditioned on this transcript.
Let $e_{k-1}\in[-1/2,1/2)$ be the signed phase error defined by $
    e_{k-1}=
    \theta-\widehat\theta_{k-1}
    \pmod 1.$
    
First, we will set up the quantity the level-$k$ probe actually measures and show it lies in a range where the measurement is invertible.  Define $x_k
    \coloneqq
    2\pi F_ke_{k-1}$
to be the residual angle seen by the level-$k$ probe. Since
$F_k=F_{k-1}+F_{k-2}\leq2F_{k-1}$, \refeq{eq:previous-level-phase-invariant} gives
\begin{align}
    |x_k|
    &=
    2\pi F_k|e_{k-1}| \notag\\
    &= 2\pi F_k
    \operatorname{dist}_{\RR/\ZZ}
    (\widehat\theta_{k-1},\theta)
    \notag\\
    &\leq
    2\pi\frac{F_k}{8F_{k-1}}
    \leq
    \frac{\pi}{2}.
    \label{eq:fibonacci-residual-principal-branch-bound}
\end{align}
Thus $x_k$ lies in the interval on which the principal branch of
$\arcsin$ inverts the sine function.

Next, we will
convert an accuracy guarantee on what the probe estimates in expectation into an accuracy guarantee on
the phase.
By \refeq{eq:fibonacci-residual-expectation}, recall that
$\mathbb E[\overline Z_k]
    =
    \sin(x_k).$
Moreover, the update in
\refeq{eq:slow-refinement-phase-update} and our definition of $e_{k-1}$ give
\begin{align}
    \operatorname{dist}_{\RR/\ZZ}
    (\widehat\theta_k,\theta)
    &\leq
    \left|
        e_{k-1}
        -
        \frac{1}{2\pi F_k}\arcsin(\overline Z_k)
    \right|
    \notag\\
    &=
    \frac{
        \left|
            x_k-\arcsin(\overline Z_k)
        \right|
    }{2\pi F_k}.
    \label{eq:phase-error-from-residual-angle}
\end{align}
Consequently, the desired invariant holds if
$\left|
        x_k-\arcsin(\overline Z_k)
    \right|
    \leq
    \frac{\pi}{4}.$
We will use the contrapositive of this statement as well as \reflem{lem:uniform-sine-separation} to generate an accuracy guarantee on
the phase. The below flow maps what we have shown so far. 
\[
        \operatorname{dist}_{\RR/\ZZ}
        (\widehat\theta_k,\theta)
        >
        \frac{1}{8F_k} 
\quad\Longrightarrow\quad
\left|
        x_k-\arcsin(\overline Z_k)
    \right|
    >
    \frac{\pi}{4} 
    \quad\Longrightarrow\quad
    \left|
        \overline Z_k-\sin(x_k)
    \right|
    >
    1-\frac{1}{\sqrt2}.
\]
The first arrow above came from the contrapositive, and the second arrow came
from taking $x = x_k$ and
$y=\arcsin(\overline Z_k)$ in \reflem{lem:uniform-sine-separation}. 
These event inclusions imply
\begin{equation}
    \Pr\!\left[
        \operatorname{dist}_{\RR/\ZZ}
        (\widehat\theta_k,\theta)
        >
        \frac{1}{8F_k}
    \right]
    \leq
    \Pr\!\left[
        \left|
            \overline Z_k-\sin(x_k)
        \right|
        >
        1-\frac{1}{\sqrt2}
    \right],
    \label{eqn:event-inclusions-inductive-lemma}
\end{equation}
which sets us up to obtain our desired concentration inequality.
Conditioned on the preceding transcript, the $M_{\mathrm{lev}}$ variables whose average
is $\overline Z_k$ are independent, $\{-1,1\}$-valued, and have common
expectation $\sin(x_k)$. We can therefore use Hoeffding's inequality and get
\begin{align}
    \Pr\!\left[
        \left|
            \overline Z_k-\sin(x_k)
        \right|
        >
        1-\frac{1}{\sqrt2}
    \right]
    &\leq
    2e^{-M_{\mathrm{lev}}/24}.
    \label{eq:hoeffding-fibonacci-refinement-bound}
\end{align}
Combining this with \refeq{eqn:event-inclusions-inductive-lemma} completes the claim. 
\end{proof}

\vspace{5mm}

The preceding lemmas establish the invariant at the initial level and show
that it propagates through every subsequent Fibonacci level. Combining
their failure bounds with the probe count gives our main guarantee for $O(\log(n))$ shots per probe.

\begin{theorem}[Restatement of
\refthm{thm:slow-refinement-fibonacci-phase-estimation}]
\label{thm:appendix-slow-refinement-fibonacci-phase-estimation}
For every \(\theta\in\RR/\ZZ\), \(K\geq2\), and \(M_{\mathrm{lev}}\geq1\),
\refalg{alg:rigorous-slow-refinement-fibonacci-estimator} is
classically efficient, uses exactly \(KM_{\mathrm{lev}}\) scalar Fibonacci phase probes,
and outputs \(w=\widehat\theta_K\) satisfying
\[
    \Pr\!\left[
        \operatorname{dist}_{\RR/\ZZ}(w,\theta)
        >
        \frac{1}{8F_K}
    \right]
    \leq
    2K e^{-M_{\mathrm{lev}}/24}.
\]
\end{theorem}

\begin{proof}
We proceed by induction. 
The algorithm's base case holds, because
\reflem{lem:two-quadrature-fibonacci-bootstrap} gives
\[
    \Pr[\operatorname{dist}_{\RR/\ZZ}
    (\widehat\theta_{2},\theta)
    >
    \frac{1}{8}]
    \leq
    4e^{-M_{\mathrm{lev}}/8}
    \leq
    4e^{-M_{\mathrm{lev}}/24}.
\]
For the inductive step, 
\reflem{lem:one-fibonacci-refinement-step} shows that, conditioned on
the invariant holding at every preceding level, its probability of failing
at level $k$ is at most $2e^{-M_{\mathrm{lev}}/24}$.
If the final estimate fails to satisfy the desired bound, then either the
bootstrap failed or the invariant first failed at one of the levels
$k\in\{3,\ldots,K\}$. A union bound gives
\begin{align}
    \Pr\!\left[
        \operatorname{dist}_{\RR/\ZZ}
        (\widehat\theta_K,\theta)
        >
        \frac{1}{8F_K}
    \right]
    &\leq
    4e^{-M_{\mathrm{lev}}/8}
    +
    \sum_{k=3}^K2e^{-M_{\mathrm{lev}}/24}
    \notag\\
    &\leq
    2K e^{-M_{\mathrm{lev}}/24}.
    \label{eq:slow-refinement-union-bound}
\end{align}
Since $w=\widehat\theta_K$, this is the claimed accuracy bound.

Finally, the bootstrap uses $2M_{\mathrm{lev}}$ probes, and each of the remaining $K-2$
levels uses $M_{\mathrm{lev}}$ probes. Thus the total number of scalar probes is
\[
    2M_{\mathrm{lev}}+(K-2)M_{\mathrm{lev}}
    =
    KM_{\mathrm{lev}}.
\]

Finally, the estimator processes its measurement outcomes online. Each
outcome is incorporated into a running signed sum, giving \(O(KM_{\mathrm{lev}})\)
counter updates over the complete execution. At the end of each level, the
estimator performs a constant number of arithmetic operations and one
evaluation of either \(\operatorname{atan2}\) or \(\arcsin\), for \(O(K)\)
elementary-function evaluations in total. It stores only the current phase
estimate, the running sums for the current level, and two consecutive
Fibonacci numbers; the complete measurement transcript need not be
retained. Since \(\log F_K=O(K)\), these quantities require
\(O(K+\log M_{\mathrm{lev}})\) bits. The required elementary functions can be evaluated
to \(O(K)\) bits of precision in time polynomial in \(K\). This establishes
the claimed classical efficiency and completes the proof.
\end{proof}

We can finally specialize to Regev's parameters.
For one Regev run, processing all \(d\) coordinates requires
\(O(dKM_{\mathrm{lev}})\) running-count updates and \(O(dK)\) evaluations of
\(\operatorname{atan2}\) or \(\arcsin\). Its classical working space is
\[
    O\!\left(d(\log F_K+\log M_{\mathrm{lev}})\right)
\]
bits. For \(d=\Theta(\sqrt n)\), \(K=\Theta(\sqrt n)\),
\(M_{\mathrm{lev}}=O(\log n)\), and \(\log F_K=\Theta(\log R)=\Theta(\sqrt n)\), this
becomes \(O(n\log n)\) arithmetic operations and \(O(n)\) bits of working
space per Regev run.

\section{Dialog schemes for in-place modular multiplication}\label{app:dialog-bounds-worst-case}

\subsection{Background: How it works}

Fix $x\in\ZZ_N^*$ and run the binary Euclidean algorithm on $(N,x)$. In our setting $N$ is always odd. So, if we initialize the Euclidean state to $(u,v)=(N,x)$ and run the algorithm on $(u,v)$, $u$ will remain odd throughout the computation. And each iteration will take one of three steps. It must decide whether to ($\mathtt d$) Divide an even entry by two, ($\mathtt a$) Subtract one entry from the other before dividing by two, or ($\mathtt i$) first interchange
the entries and then subtract and divide.  
We record these
three possibilities by the symbols $\mathtt d$, $\mathtt a$, and $\mathtt i$.  
On input $(u,v)$ these take the form:
\begin{align}
    \mathtt d:
    (u,v)&\longmapsto (u,v/2)
    &&\text{if $v$ is even},
    \label{eq:dialog-divide-step}\\
    \mathtt a:
    (u,v)&\longmapsto \bigl(u,(v-u)/2\bigr)
    &&\text{if $v$ is odd and $u\leq v$},
    \label{eq:dialog-add-step}\\
    \mathtt i:
    (u,v)&\longmapsto \bigl(v,(u-v)/2\bigr)
    &&\text{if $v$ is odd and $u>v$}.
    \label{eq:dialog-interchange-step}
\end{align}
A complete run produces a word such as $g(x)=\mathtt{iadiddiaa},$
which records the branch taken at every iteration.  This word is the
\emph{dialog} of $x$ relative to $N$, and if there are $T$ iterations, then
the dialog is
$g(x)\in\{\mathtt d,\mathtt a,\mathtt i\}^T$. 

The dialog contains exactly the information needed to reverse the Euclidean
computation.  In the backward direction, the three symbols instruct us to
double, to double and add, or to double, add, and interchange.  Acting on a
column vector, these three linear updates can be represented by
\begin{equation}
    M_{\mathtt d}
    =\begin{pmatrix}1&0\\0&2\end{pmatrix},
    \qquad
    M_{\mathtt a}
    =\begin{pmatrix}1&0\\1&2\end{pmatrix},
    \qquad
    M_{\mathtt i}
    =\begin{pmatrix}1&2\\1&0\end{pmatrix}.
    \label{eq:dialog-symbol-matrices}
\end{equation}
Thus, once the dialog is known, running the Euclidean computation backward
requires no further comparisons: we simply apply the update specified by each
symbol, in reverse order.
Let $\mathsf{Replay}_{g(x)}$ denote this backward computation.  If the Euclidean algorithm on $(N,x)$ takes the steps described by $g(x)= \mathtt{iadid}$ and ends in  $(1,0)$, we can use the dialog to reconstruct $(N,x)$. We replay the
dialog on $(1,0)$: 
\begin{equation}
    \mathsf{Replay}_{g(x)}
    \begin{pmatrix}1\\0\end{pmatrix}
    =
    M_{\mathtt d}M_{\mathtt i}M_{\mathtt d}M_{\mathtt a}M_{\mathtt i}
    \begin{pmatrix}1\\0\end{pmatrix}
    =
    \begin{pmatrix}N\\x\end{pmatrix},
    \label{eq:dialog-reconstructs-multiplier}
\end{equation}
which is equivalent to $(0,x)^T \pmod N$.
The key observation is that every update in
\refeq{eq:dialog-symbol-matrices} is linear.  Because
$\left(y,0\right)=y\left(1,0\right)$, the same replay can instead begin at
$(y,0)$:
\begin{align}
    \mathsf{Replay}_{g(x)}
    \begin{pmatrix}y\\0\end{pmatrix}
    &=
    y\,
    \mathsf{Replay}_{g(x)}
    \begin{pmatrix}1\\0\end{pmatrix}
    \notag
    \equiv
    \begin{pmatrix}0\\xy\end{pmatrix}
    \pmod N.
    \label{eq:dialog-replay-multiplies}
\end{align}
This final trick is what allows the Euclidean transcript to perform multiplication in-place. In other words, the dialog
that reconstructs $x$ from $1$ reconstructs $xy$ from $y$.  

The complete in-place multiplication has three stages.  First, reversibly
convert $x$ into its dialog $g(x)$.  Second, replay this dialog on $(y,0)$ to
obtain $(0,xy\bmod N)$.  Finally, reverse the first stage to recover $x$ and
erase the dialog. The logical data flow is
\begin{equation}
\begin{aligned}
    \ket{x}\ket{y}\ket{0}
    &\longmapsto
    \ket{g(x)}\ket{y}\ket{0} \\
    &\longmapsto
    \ket{g(x)}\ket{0}\ket{xy\bmod N} \\
    &\longmapsto
    \ket{x}\ket{0}\ket{xy\bmod N},
\end{aligned}
\label{eq:dialog-multiplication-three-stages}
\end{equation}
where here we do not represent any (sublinear) workspace registers. 
Dropping the clean zero register from the notation and taking the final
register as the output gives the interface in
\refeq{eq:dialog-in-place-multiplication-interface}.
Schrottenloher's main space saving is to share storage between the Euclidean
state and its dialog: as the active Euclidean operands shrink, the released
qubits store later pieces of the dialog.  Therefore, in the fixed-length encoding used
there, the replay stage 
\begin{equation}
    \underbrace{\ket{g(x)}}_{2.355n}
     \underbrace{\ket{y}}_{n\text{ }}
     \underbrace{\ket{0}}_{n\text{ }}
    \longmapsto
     \underbrace{\ket{g(x)}}_{2.355n\text{ }}
     \underbrace{\ket{0}}_{n\text{ }}
     \underbrace{\ket{xy\bmod N}}_{n\text{ }}.
\label{eq:dialog-replay-register-layout}
\end{equation}
exhibits the peak space usage of the published in-place multiplier. The multiplier requires $4.355n+O(\sqrt n)$ qubits,
and its Toffoli count is $O(n^2)$. 
Note that this space usage corresponds to an approximate implementation: Schrottenloher uses a fixed
Euclidean-iteration and operand-width allocations, which empirically works for most inputs but can overflow on some inputs.

\subsection{Worst-case resource bound for general inputs.}

\paragraph{Why we cannot immediately import the resource analysis of
Ref.~\cite{Schrottenloher26}.}
We cannot immediately import Schrottenloher's result because the size of the dialog register
follows from an empirical study on the
number of Euclidean iterations for random inputs. 
More specifically,
Schrottenloher fixes the number of Euclidean iterations to
$1.413n+o(n)$. As described above, every three iterations are encoded using five Dialog bits.
This
gives $(5/3) \times ( 1.413n+o(n) ) = 2.355n+o(n)$-qubit Dialog register. This cutoff on the number of allowed Euclidean iterations allows the circuit
to fail on inputs whose computation requires more iterations.

This resource estimate was developed for elliptic-curve arithmetic, where
Schrottenloher numerically studies the circuit on random field elements. However, this is crucially not the same as our setting. Our multiplication inputs
have a different structure: they are successive components of a single
traveling eigenstate, obtained by repeatedly applying backward Fibonacci
updates. They are, therefore, correlated across levels and need not follow the
random-input distribution studied in Ref.~\cite{Schrottenloher26}. 
Recall that  we perform in-place multiplication on our accumulators to generate the next level, $k \longrightarrow k+1$, of our traveling eigenstate. At level $k$, the two accumulators have the form
\begin{equation}
    \ket{\widetilde{u_\theta}^{(k)}}_{XY}
    =
    \sum_{x,y}
    \alpha_{x,y}^{(k)}
    \ket{x}_X\ket{y}_Y,
    \label{eq:traveling-eigenstate-basis-expansion}
\end{equation}
and the next Fibonacci update should act correctly on every pair $(x,y)$ in
this superposition. If even one supported value of $x$ requires more
Dialog space than we have allocated, the circuit does not implement the
desired Fibonacci update on that component of the state. The resulting
accumulator state may then fail to produce the correct phase kickback in
later probes.
Moreover, the multipliers at different levels arise recursively from the same
initial accumulator values. For example, on a lattice vector branch \(z\in\mathbb Z^d\), beginning with mask
value $m$ and $f(z)$, successive multipliers include
\[
    m,\qquad f(z)m^{-1},\qquad m^2f(z)^{-1},\qquad\ldots.
\]
They are therefore correlated and do not immediately align within Schrottenloher's random-input
distribution. This leads us to prove a worst-case bound over the full support
of the traveling eigenstate.

\paragraph{Worst-case resource bound.}
We now give a worst-case resource bound for Dialog multiplication when both
inputs may contain $n$ bits. 
In this bound, we use exact comparisons, where \cite{Schrottenloher26} employed approximate gadgets. This choice affects
the concrete gate cost but not the leading space bound.
For a nonnegative integer $t$, let $\textsf{len}(t)$ denote its binary
length, with $\textsf{len}(0)=0$. We first bound the number of iterations
needed to produce the Dialog.

\begin{lemma}[Worst-case Euclidean computation length]
\label{lem:worst-case-euclidean-computation-length}
For every odd $N<2^n$ and $x\in\ZZ_N^*$, the Euclidean computation in
\refeqs{eq:dialog-divide-step}--\eqref{eq:dialog-interchange-step},
initialized at $(N,x)$, reaches $(1,0)$ after at most
$L_{\max}=2n-1$ iterations. 
\end{lemma}

\begin{proof}
For a positive integer $t$, let $\operatorname{len}(t)$ denote its binary
length, and set $\operatorname{len}(0)=0$. Consider the potential
  $  \Phi(u,v)
    \coloneqq
    \operatorname{len}(u)+\operatorname{len}(v).$
Initially, $\Phi(N,x)\leq 2n$. Every iteration before reaching $(1,0)$
decreases this potential by at least one:
\begin{itemize}
    \item In a $\mathtt d$ step, the positive even value $v$ is divided by
    two, so its binary length decreases by one.

    \item In an $\mathtt a$ step, if $u<v$, then
    $(v-u)/2<v/2$, so the binary length of the second entry decreases by at
    least one. If $u=v$, the second entry becomes zero.

    \item In an $\mathtt i$ step, the new pair is
    $(v,(u-v)/2)$. The first entry has the binary length of the old $v$,
    while $(u-v)/2<u/2$, so the binary length of the second entry is at
    least one less than that of the old $u$.
\end{itemize}
Because $x\in\ZZ_N^*$, the terminal state is $(1,0)$ and has potential $\Phi(1,0)=1$.
The computation therefore reaches $(1,0)$ after at most $2n-1$ iterations.
\end{proof}

Note that the above bound is tight. For any odd $n$-bit $N$, setting $x=2^{n-1}$ gives
$n-1$ division steps before reaching $(N,1)$, followed by one interchange
step that produces $(1,(N-1)/2)$. Each subsequent iteration replaces the
second entry by its quotient upon division by two, requiring another
$n-1$ iterations.

\begin{lemma}[Worst-case Dialog length]
\label{lem:worst-case-dialog-length}
For every odd $N<2^n$ and $x\in\ZZ_N^*$, the Dialog produced by the
Euclidean computation in
\refeqs{eq:dialog-divide-step}--\eqref{eq:dialog-interchange-step},
initialized at $(N,x)$, can be stored reversibly using
$\frac{10}{3}n+O(1)$ qubits.
\end{lemma}

\begin{proof}
Following Schrottenloher~\cite[Algorithm~2]{Schrottenloher26}, each
Euclidean iteration records the pair $(b_0,b_0b_1)$, where
$b_0=v\bmod 2$ and $b_1$ indicates whether $u>v$. This pair takes only the
three values $(0,0)$, $(1,0)$, and $(1,1)$, corresponding respectively to
the steps $\mathtt d$, $\mathtt a$, and $\mathtt i$.
Schrottenloher~\cite[Figure~1]{Schrottenloher26} gives an explicit
reversible circuit that compresses three such pairs from six qubits to five,
returning the sixth qubit to $\ket{0}$. We apply this circuit after every
three iterations and reuse the cleared qubit and the constant-size encoding
workspace for the next block. Running the circuit backward recovers the
three original Dialog symbols, so the resulting encoding is reversible.

By \reflem{lem:worst-case-euclidean-computation-length}, the Euclidean
computation reaches $(1,0)$ after at most $2n-1$ iterations. So we
must run the computation for exactly $T$ iterations, where $T = 3\ceil{\frac{2n-1}{3}}$.
Note that some inputs may reach $(1,0)$ well before all $T$ iterations have happened. When this is the case,
the remaining iterations are padded with $\mathtt d$ steps, which leave
$(1,0)$ unchanged.
We can complete the claim by noticing that, since $T$ is divisible by three, the complete Dialog requires the following space:
\begin{equation}
    \frac{5}{3}T
    =
    5\ceil{\frac{2n-1}{3}}
    =
    \frac{10}{3}n+O(1).
    \label{eq:worst-case-compressed-dialog-size}
\end{equation}
\end{proof}

Schrottenloher observes that, as the Euclidean operands $(u,v)$ shrink, the space they no longer occupy can be reallocated to store the growing Dialog~\cite[Section~3.1]{Schrottenloher26}. We use the same register-sharing idea and analyze it in the worst case. At every iteration, we upper bound the total space required by both active operands and the Dialog generated so far. This gives a bound that holds even for computations attaining the worst-case Dialog length established in \reflem{lem:worst-case-dialog-length}. While Ref~\cite
{Schrottenloher26} similarly reallocates space released by the operands, Schrottenloher estimates the amount available using an empirically observed average-case shrinkage rate.

\begin{lemma}[Worst-case-safe Dialog storage]
\label{lem:worst-case-safe-dialog-storage}
Let $(u_j,v_j)$ denote the operands after $j$ iterations.
After $j\in\{0,\ldots,T\}$ Euclidean iterations, both active operands
fit in
\begin{equation}
    w_j
    \coloneqq
    \max\!\left\{1,\min\{n,2n-j\}\right\}
    \label{eq:worst-case-euclidean-register-width}
\end{equation}
qubits. The two active operands and the first $j$ Dialog symbols can 
be stored using at most $\frac{11}{3}n+O(1)$
qubits.
\end{lemma}
\begin{proof}
The proof of
\reflem{lem:worst-case-euclidean-computation-length} shows that the
total binary length of our operands $(u_j,v_j)$ decreases by at least one
per iteration until reaching $(1,0)$ and remains one thereafter. Hence $\operatorname{len}(u_j)+\operatorname{len}(v_j)
    \leq
    \max\{1,2n-j\}.$
Moreover, neither operand ever exceeds $n$ bits. So after
iteration $j$, we can store each operand using 
\begin{equation}
    \max\{1,\min\{n,2n-j\}\}
    \label{eq:register_operand_size}
\end{equation}
qubits, which matches the definition of $w_j$.
Next, the first $j$ Dialog symbols occupy $\ceil{5j/3}$ qubits, up to a
constant-size scratch buffer used to compress the current block.

Following \refeqlowerplain{eq:register_operand_size}, for $0\leq j\leq n$ both operand registers retain their original
$n$-qubit width. The total space is at most
\[
    2n+\ceil{5j/3}+O(1)
    \leq
    \frac{11}{3}n+O(1).
\]
For $n\leq j\leq2n-1$, each operand requires at most $2n-j$ qubits. The
total space is at most
\[
    2(2n-j)+\ceil{5j/3}+O(1)
    =
    4n-\frac{j}{3}+O(1)
    \leq
    \frac{11}{3}n+O(1).
\]
Finally, during the at most two padding iterations with $j>2n-1$, the
Euclidean state remains $(1,0)$, while the complete Dialog occupies
$\frac{10}{3}n+O(1)$ qubits. The same bound therefore continues to hold.

Finally, notice that operand-register widths depend only on $j$. So on every basis branch, qubits outside the active windows are in $\ket{0}$ and may be
reassigned to the Dialog register. During the reverse traversal, the
corresponding Dialog bits are erased before these zero qubits are returned
to the active windows. The same predetermined sharing schedule therefore
works coherently when $x$ is in superposition.
\end{proof}

Using the careful resource accounting we have built up, we can now state our in-place modular multiplication resource bounds. We will follow it up with a resource bound for the backward
Fibonacci updates in our streamed algorithm.

\begin{lemma}[Exact big-by-big Dialog multiplication]
\label{lem:exact-big-by-big-dialog-multiplication}
For odd $N<2^n$, there exists an exact quantum circuit implementing 
\begin{equation}
    \ket{x}\ket{y}
    \longmapsto
    \ket{x}\ket{xy\bmod N}
    \label{eq:exact-big-by-big-dialog-interface}
\end{equation}
for every $x\in\ZZ_N^*$ and $y\in\ZZ_N$. Its gate count, total qubit count,
and depth satisfy
\begin{equation}
    G_{\mathrm{bb}}=O(n^2),
    \qquad
    S_{\mathrm{bb}}=\frac{16}{3}n+O(1),
    \qquad
    D_{\mathrm{bb}}=O(n^2).
    \label{eq:exact-big-by-big-depth}
\end{equation}
\end{lemma}

\begin{proof}
We first analyze space. Following Schrottenloher~\cite[Algorithm~2]{Schrottenloher26}, initialize the Euclidean pair to $(N,x)$ and run $T$
iterations to convert $x$ into its compressed Dialog. Register $y$ remains
untouched during this stage. By
\reflem{lem:worst-case-safe-dialog-storage}, its peak space is at most
\begin{equation}
    n+\frac{11}{3}n+O(1)
    =\frac{14}{3}n+O(1).
    \label{eq:exact-dialog-generation-total-space}
\end{equation}
This conversion consumes the register containing $x$: its information is
carried by the Dialog.
Next, replay the Dialog on the pair $(y,0)$, which maps $(y,0)$ to
$(0,xy\bmod N)$. At this stage both $n$-qubit replay
registers and the complete Dialog are present, so \reflem{lem:worst-case-dialog-length} gives
\begin{equation}
   2n+\frac{10}{3}n+O(1)
    =\frac{16}{3}n+O(1)
    \label{eq:exact-dialog-replay-total-space}
\end{equation}
qubits. Finally, reverse the Dialog-generation circuit to recover $x$ and
erase the Dialog. The product remains untouched, so this stage has the same
$\frac{14}{3}n+O(1)$ space bound as the forward generation stage. The replay
stage therefore determines the peak space.

The gate-count analysis follows Schrottenloher's Dialog-generation and
replay circuits~\cite[Algorithms~2 and~3]{Schrottenloher26}. Each iteration
uses a constant number of arithmetic operations, on at most $n$ qubits.
Using exact ripple-carry comparisons, additions, and subtractions, each
iteration requires $O(n)$ gates, $O(n)$ depth, and $O(1)$ ancillary qubits;
exact modular doubling and controlled modular addition satisfy the same
bounds~\cite{roetteler17}. 
And note that the Dialog-compression circuit of
Ref.~\cite[Figure~1]{Schrottenloher26} uses only $O(1)$ gates per block.
Since there are $T=O(n)$ iterations in each
traversal, the complete multiplier requires $O(n^2)$ gates and
depth.
\end{proof}

\begin{corollary}[Exact backward Fibonacci update]
\label{cor:exact-big-by-big-fibonacci-update}
For every odd $N<2^n$, $x\in\ZZ_N^*$, and $y\in\ZZ_N$, there exists an
exact quantum circuit implementing
\begin{equation}
    \mathsf{Fib}^{-1}\ket{x}\ket{y}
    =
    \ket{yx^{-1}\bmod N}\ket{x}.
    \label{eq:exact-big-by-big-fibonacci-update}
\end{equation}
Its gate count, total qubit count, and depth satisfy the bounds in the resource bounds of \reflem{lem:exact-big-by-big-dialog-multiplication}.
\end{corollary}

\begin{proof}
This follows trivially by running the circuit from
\reflem{lem:exact-big-by-big-dialog-multiplication} backwards.
\end{proof}

\section{Adaptive Bayesian Fibonacci IQPE: Algorithm and Simulation Details}\label{app:numerics-iqpe-fibonacci}

This appendix gives the complete specification of the adaptive Bayesian
estimator used in the numerical experiments of
\refsec{sec:fibonacci-estimator-numerics}. We now define these
operations in detail.

\subsection{Abstract construction}

It is useful to label each scalar measurement by its Fibonacci level $k$ and
shot number $s$. Let $M_k$ denote the number of shots actually taken at level
$k$; this number is determined by the measurement outcomes rather than fixed
in advance. The corresponding probe uses exponent $F_k$ and rotation
$\phi_{k,s}$ and produces an outcome $b_{k,s}\in\{0,1\}$. We record the
triples
\[
    (F_k,\phi_{k,s},b_{k,s})
\]
in chronological order and denote the resulting ordered transcript by
$\tau$. When only the outcomes are relevant, we refer to the outcome
component of this transcript. If the unknown phase were
$\theta\in\RR/\ZZ$, the likelihood of an observed outcome is given by
\refeq{eq:fibonacci-probe-likelihood}.

\paragraph{Representing the posterior by cells.}
We would like to retain the relative likelihood of every possible value of
$\theta$, but directly storing a continuous probability density becomes
increasingly expensive as the transcript grows. We instead partition the
phase circle into $B$ intervals,
which we call cells.
Each live cell $C$ is represented by its midpoint $m(C)$, width $|C|$,
current posterior weight $p_\tau(C)$, and creation level $\kappa(C)$. The
creation level is the Fibonacci level upon whose entry $C$ was created. We will refine our cells over the course of the algorithm -- specifically, we will allow refinement each time we begin probing a new Fibonacci level.
We represent the posterior within each cell by a
point mass at its midpoint. Thus the cell weights form a discrete midpoint
approximation to the posterior and satisfy
\begin{equation}
    p_\tau(C)\geq0,
    \qquad
    \sum_{C\in\mathcal C}p_\tau(C)=1.
    \label{eq:fibonacci-cell-normalization}
\end{equation}
Before any measurements, we use an equal-width partition and assign each
cell weight $1/B_0$.

Suppose that the $s$th probe at level $k$ produces outcome $b_{k,s}=b$. We
first append $(F_k,\phi_{k,s},b_{k,s})$ to $\tau$ and then update the cell
weights by Bayes' rule:
\begin{equation}
    p_\tau(C)
    \longleftarrow
    \frac{
        p_\tau(C)\Pr[b\mid m(C),F_k,\phi_{k,s}]
    }{
        \displaystyle
        \sum_{C'\in\mathcal C}
        p_\tau(C')\Pr[b\mid m(C'),F_k,\phi_{k,s}]
    }.
    \label{eq:fibonacci-cell-weight-update}
\end{equation}
Here the weights on the right-hand side are their values before the new
outcome is processed; afterward, $p_\tau$ denotes the weights associated
with the enlarged transcript.

\paragraph{Cell refinement: splitting large, high-probability cells.}

Ultimately, we want to estimate $\theta$ to high precision while keeping the
classical postprocessing efficient. This is particularly important for
streaming Regev factoring. For an $n$-bit integer, Regev's postprocessing
requires scalar precision
\(
    \varepsilon_R=2^{-(A+o(1))\sqrt n},
\)
as stated in
\refeq{eq:regev-required-scalar-precision}. A uniform grid at this
resolution would require $2^{\Theta(\sqrt n)}$ cells and therefore
superpolynomial classical space. We instead impose
$B_{\max}=\poly(n)$ and add resolution selectively in regions carrying
substantial posterior weight. Refinement occurs once upon entry to each new
Fibonacci level. If the refinement condition still requests a split when
$|\mathcal C|=B_{\max}$, the protocol returns \textsc{failure}.

At level $k$, a probe acquires relative phase
$\exp(2\pi iF_k\theta)$ and therefore gives information about
$F_k\theta\bmod1$. Multiplication by $F_k$ locally scales an interval of
possible phases by $F_k$, so $F_k|C|$ measures the width of $C$ at the
resolution of a level-$k$ probe. We combine this scaled width with the
posterior weight of the cell.

\begin{definition}[Cell refinement score]
\label{def:fibonacci-cell-refinement-score}
At Fibonacci level $k$, the refinement score of a cell $C\in\mathcal C$ is
\begin{equation}
    \mathsf{score}_k(C)
    \coloneqq
    p_\tau(C)\bigl(F_k|C|\bigr)^{2}.
    \label{eq:fibonacci-cell-refinement-score}
\end{equation}
\end{definition}

A cell receives a large refinement score when it remains wide relative to
the resolution $1/F_k$, carries substantial posterior weight, or both. Given
a threshold $\eta_{\mathrm{grid}}>0$, we repeatedly bisect a cell of greatest
score until
\[
    \mathsf{score}_k(C)
    \leq
    \eta_{\mathrm{grid}}
\]
for every $C\in\mathcal C$.
A smaller threshold creates more, narrower cells and represents the
posterior at finer resolution; a larger threshold uses fewer cells.

To bisect the parent cell $C$, we replace it with its left and right children
$C_{\mathrm L}$ and $C_{\mathrm R}$. Both children have width $|C|/2$ and
creation level $\kappa(C_{\mathrm L})  = \kappa(C_{\mathrm R})  = k.$
Their midpoints are
\[
    m(C_{\mathrm L})
    =
    m(C)-\frac{|C|}{4}\pmod 1,
    \qquad
    m(C_{\mathrm R})
    =
    m(C)+\frac{|C|}{4}\pmod 1.
\]
We divide the parent weight $p_\tau(C)$ between its two children in
proportion to how likely their respective midpoints are, given the
measurements collected since the parent cell $C$ was created at level
$\kappa(C)$.

We perform cell refinement only upon entering a new Fibonacci level, before
taking any measurements at that level. Therefore, the observations collected
since the parent cell $C$ was created are precisely those entries of $\tau$
whose Fibonacci level is at least $\kappa(C)$. For a candidate phase $x$,
define the parental likelihood
\begin{equation}
    \mathcal L_{\tau,C}(x)
    \coloneqq
    \prod_{\substack{
        (F_k,\phi_{k,s},b_{k,s})\in\tau\\
        k\geq\kappa(C)
    }}
    \Pr[b_{k,s}\mid x,F_k,\phi_{k,s}].
    \label{eq:fibonacci-since-birth-likelihood}
\end{equation}
We then assign the child weights according to
\begin{equation}
    p_\tau(C_\alpha)
    \coloneqq
    p_\tau(C)
    \frac{
        \mathcal L_{\tau,C}\!\left(m(C_\alpha)\right)
    }{
        \mathcal L_{\tau,C}\!\left(m(C_{\mathrm L})\right)
        +
        \mathcal L_{\tau,C}\!\left(m(C_{\mathrm R})\right)
    },
    \qquad
    \alpha\in\{\mathrm L,\mathrm R\}.
    \label{eq:fibonacci-child-cell-weights}
\end{equation}
By construction,
$p_\tau(C_{\mathrm L})+p_\tau(C_{\mathrm R})=p_\tau(C),$
so splitting a cell preserves its total posterior mass. If no observations
have been collected since $C$ was created, the two parental likelihoods are
equal and the parent mass is divided equally between its children. This can
occur, for example, when a newly created cell is subdivided again during the
same refinement stage. Subsequent observations update the child weights
using \refeq{eq:fibonacci-cell-weight-update}.

\paragraph{Adaptivity: choosing the IQPE rotation angle to maximize information gain}

Before every adaptive shot $s$ at level $k$, the protocol chooses a rotation
angle $\phi_{k,s}$ intended to make the next measurement as informative as
possible. For a candidate angle $\phi\in\Phi$, the posterior-predictive
probability of outcome $b$ is
\begin{equation}
    \Pr[b\mid\tau,F_k,\phi]
    =
    \sum_{C\in\mathcal C}
        p_\tau(C)
        \Pr[b\mid m(C),F_k,\phi].
    \label{eq:fibonacci-predictive-outcome-probability}
\end{equation}
Let $h_2$ denote the binary entropy function. The expected information gain
from the candidate angle $\phi$ is
\begin{align}
    \textsf{EIG}(\phi)
    \coloneqq{}&
    h_2\!\left(
        \sum_{C\in\mathcal C}
        p_\tau(C)
        \Pr[0\mid m(C),F_k,\phi]
    \right)
    \notag\\
    &-
    \sum_{C\in\mathcal C}
        p_\tau(C)
        h_2\!\left(
            \Pr[0\mid m(C),F_k,\phi]
        \right).
    \label{eq:fibonacci-expected-information-gain}
\end{align}
The first term is the entropy of the predicted measurement outcome. The
second is the expected outcome entropy conditioned on the cell containing
the phase. Their difference is the mutual information between the discrete
cell index and the next outcome under the midpoint-posterior approximation
\cite[Chapter~2]{CoverThomas06}. For every adaptive shot, we choose
\begin{equation}
    \phi_{k,s}
    \in
    \argmax_{\phi\in\Phi}
        \textsf{EIG}(\phi),
    \label{eq:fibonacci-per-shot-eig-angle}
\end{equation}
breaking ties in favor of the smallest candidate angle. After observing the
outcome, we update the posterior and maximize the EIG again before the next
shot.

\paragraph{Initial level.}

At the initial level $F_2=1$, the multiplied phase is simply
$F_2\theta\bmod1=\theta$, so these measurements introduce no modular
aliases. Their purpose is to obtain an initial localization of $\theta$ on
the unit circle. A single quadrature cannot distinguish phases related by a
reflection. At angle $0$,
\refeq{eq:fibonacci-probe-likelihood} becomes
\begin{equation}
    \Pr[b_{2,s}=b\mid\theta,F_2,0]
    =
    \frac{1+(-1)^{b}\cos(2\pi\theta)}{2},
    \label{eq:fibonacci-bootstrap-cosine-likelihood}
\end{equation}
whereas at angle $-\pi/2$ it becomes
\begin{equation}
    \Pr[b_{2,s}=b\mid\theta,F_2,-\pi/2]
    =
    \frac{1+(-1)^{b}\sin(2\pi\theta)}{2}.
    \label{eq:fibonacci-bootstrap-sine-likelihood}
\end{equation}
The sine and cosine quadratures together resolve the reflection ambiguity.
We therefore take exactly one bootstrap shot at angle $0$ and one at angle
$-\pi/2$, updating the cell weights after each outcome. We then check the
ordinary level-$2$ transition guard. If it does not pass, we take additional
level-$2$ shots using the per-shot EIG rule until the guard passes or the
ordinary-level adaptive-shot cap is exhausted.

\paragraph{Generic Fibonacci level: the half-circle transition guard.}

For every level $k\in\{2,\ldots,K-1\}$, the transition guard determines when
the protocol advances to level $k+1$. Larger Fibonacci exponents provide
finer phase resolution but also admit more inverse branches for the unknown
phase. The guard is a soft geometric criterion intended to limit this
ambiguity.

For each live cell $C\in\mathcal C$, define its midpoint location in
level-$k$ phase space by
\begin{equation}
    u_k(C)
    \coloneqq
    F_k\,m(C)\bmod 1.
    \label{eq:level-k-cell-location}
\end{equation}
For a circular arc $A\subseteq\RR/\ZZ$, define its level-$k$ posterior mass
by
\begin{equation}
    \textsf{arcmass}_{\tau,k}(A)
    \coloneqq
    \sum_{C\in\mathcal C}
    p_\tau(C)\,
    \mathbf 1\!\left\{u_k(C)\in A\right\}.
    \label{eq:fibonacci-midpoint-arc-mass}
\end{equation}
Circular arcs may wrap around zero. Consistent with the posterior
representation above, each cell contributes its full weight at its midpoint;
boundary cells are not divided fractionally. Let $A^\star_{\tau,k}$ be a shortest
level-$k$ arc containing at least $1-\delta_{\mathrm{move}}$ of the posterior
mass:
\begin{equation}
    A^\star_{\tau,k}
    \in
    \operatorname*{arg\,min}_{\substack{
        A\subseteq\RR/\ZZ\ \\
        \textsf{arcmass}_{\tau,k}(A)
        \geq1-\delta_{\mathrm{move}}
    }}
    |A|.
    \label{eq:fibonacci-shortest-transition-arc}
\end{equation}
The protocol advances to level $k+1$ precisely when
\begin{equation}
    \frac{F_{k+1}}{F_k}|A^\star_{\tau,k}|
    \leq
    \frac12.
    \label{eq:fibonacci-half-circle-transition-guard}
\end{equation}
This criterion requires the high-mass arc, after rescaling by the ratio of
successive Fibonacci exponents, to have width at most one half-turn. A half
turn is the natural scale on which a suitably shifted sinusoidal likelihood
can be restricted to one inverse branch. We therefore use
\refeq{eq:fibonacci-half-circle-transition-guard} as an
anti-aliasing criterion, rather than as a guarantee that every modular alias
has been eliminated.

The protocol checks the guard before every possible adaptive shot, including
immediately after the bootstrap at level $2$ and immediately upon entering
each later nonfinal level. If it fails, the protocol chooses another
level-$k$ angle by \refeq{eq:fibonacci-per-shot-eig-angle}, takes the
probe, updates the posterior, and checks again. Thus a newly entered level
may require no shots. If the guard remains false after the configured
ordinary-level adaptive-shot cap, the protocol returns \textsc{failure}.

The guard is soft because posterior mass outside the selected
$(1-\delta_{\mathrm{move}})$-mass arc is retained and participates in all
later Bayesian updates. No cells are pruned. If refinement requests another
split after the live-cell count reaches $B_{\max}$, the protocol instead
returns a grid-saturation failure.

\subsection{Formal estimator and certification algorithm}
\label{sec:formal-fibonacci-estimator}

The measurement counts $M_2,\ldots,M_K$ are determined adaptively. At each
nonfinal level $k<K$, the protocol takes measurements until the transition
guard passes and then refines the cell representation upon entering level
$k+1$. Upon entry to the final level $K$, this refinement is performed once
and the resulting grid is frozen. Certification is checked before every
possible level-$K$ shot. If the condition is false, the protocol takes one
EIG-selected level-$K$ measurement, updates the weights on the frozen grid,
and checks again. This continues until the certificate succeeds or the
final-level measurement cap is reached. The certificate is checked after the
last permitted shot; the protocol returns \textsc{failure} only if it still
fails at that point.

To define the final estimator, consider every circularly consecutive run
$\mathcal R\subseteq\mathcal C$ of cell midpoints whose circular span is at
most $2\varepsilon$. Define its posterior mass by
\begin{equation}
    p_\tau(\mathcal R)
    \coloneqq
    \sum_{C\in\mathcal R}p_\tau(C),
    \label{eq:fibonacci-certification-run-mass}
\end{equation}
and let $\mathcal R^\star$ be a run of maximum posterior mass, breaking ties
by the smallest starting midpoint. Equivalently,
\begin{equation}
    \mathcal R^\star
    \in
    \operatorname*{arg\,max}_{\substack{
        \mathcal R\subseteq\mathcal C
        \text{ circularly consecutive}\\
        \text{circular span of }\mathcal R\leq2\varepsilon
    }}
    p_\tau(\mathcal R).
    \label{eq:fibonacci-heaviest-certification-run}
\end{equation}
Let $a^\star$ and $b^\star$ be the first
and last midpoint representatives in $\mathcal R^\star$, unwrapped in
circular order. We set
\begin{equation}
    w(\mathcal R^\star)
    \coloneqq
    \frac{a^\star+b^\star}{2}
    \pmod1.
    \label{eq:fibonacci-certified-estimate}
\end{equation}
Because $b^\star-a^\star\leq2\varepsilon$, every midpoint in
$\mathcal R^\star$ lies within circular distance $\varepsilon$ of $w$. The
protocol returns $w$ when
\begin{equation}
    p_\tau(\mathcal R^\star)
    \geq
    1-\delta_{\mathrm{final}}.
    \label{eq:fibonacci-final-certification-condition}
\end{equation}

Let \(M_{\mathrm{move}}^{\max}\) denote the maximum number of
additional EIG-selected measurements permitted while waiting for any
nonfinal transition guard to pass. The two bootstrap measurements are not
counted against this cap. Let \(M_{\mathrm{final}}^{\max}\) denote the
maximum number of measurements permitted at level \(K\).

\begin{algorithm}[H]
\caption{Slow-refinement adaptive Bayesian Fibonacci IQPE.}
\label{alg:adaptive-bayesian-fibonacci-estimator}

\KwData{Target accuracy $\varepsilon$; parameters
$\delta_{\mathrm{move}}$ and $\delta_{\mathrm{final}}$; final level $K$;
initial cell count $B_0$; cell cap $B_{\max}$; refinement threshold
$\eta_{\mathrm{grid}}$; candidate-angle set $\Phi$; maximum number
$M_{\mathrm{move}}^{\max}$ of additional adaptive measurements at each
nonfinal level; and final-level measurement cap
$M_{\mathrm{final}}^{\max}$.}
\KwResult{A certified estimate $w$ of $\theta$, or \textsc{failure}.}

\begin{enumerate}
\item \textit{Initialize.}
Set $\tau\gets\emptyset$ and $M_k\gets0$ for every
$k\in\{2,\ldots,K\}$. Partition $\RR/\ZZ$ into $B_0$ equal-width
cells $\mathcal C$. For every $C\in\mathcal C$, set
$p_\tau(C)\gets1/B_0$ and $\kappa(C)\gets2$.

\item \textit{Bootstrap.}
Set $\phi_{2,1}\gets0$ and $\phi_{2,2}\gets-\pi/2$.
For $s=1,2$:
\begin{enumerate}
\item Take the level-$2$ probe with angle $\phi_{2,s}$, obtaining
$b_{2,s}$. Append $(F_2,\phi_{2,s},b_{2,s})$ to $\tau$.
\item Simultaneously update $p_\tau(C)$ for every
$C\in\mathcal C$ using
\refeq{eq:fibonacci-cell-weight-update}.
\end{enumerate}
Set $M_2\gets2$.

\item \textit{Climb the Fibonacci levels.}
For $k=2,\ldots,K-1$:
\begin{enumerate}
\item Set $s\gets M_k$.
\item While $s<M_k+M_{\mathrm{move}}^{\max}$ and
$(F_{k+1}/F_k)|A^\star_{\tau,k}|>1/2$:
\begin{enumerate}
\item Set $s\gets s+1$ and choose
\[
    \phi_{k,s}\in
    \operatorname*{arg\,max}_{\phi\in\Phi}\textsf{EIG}(\phi).
\]
\item Take the level-$k$ probe with angle $\phi_{k,s}$,
obtain outcome $b_{k,s}$, and append
$(F_k,\phi_{k,s},b_{k,s})$ to $\tau$.
\item Simultaneously update $p_\tau(C)$ for every
$C\in\mathcal C$ using
\refeq{eq:fibonacci-cell-weight-update}.
\end{enumerate}
Set $M_k\gets s$. If $(F_{k+1}/F_k)|A^\star_{\tau,k}|>1/2$, return
\textsc{failure}.

\item While there exists
$C\in\mathcal C$ such that
$\mathsf{score}_{k+1}(C)>\eta_{\mathrm{grid}}$:
\begin{enumerate}
    \item If $|\mathcal C|=B_{\max}$, return
    \textsc{failure}.
    \item Choose
        $C\in
        \operatorname*{arg\,max}_{C'\in\mathcal C}
        \mathsf{score}_{k+1}(C').$
    \item Bisect $C$ into $C_{\mathrm L}$ and $C_{\mathrm R}$, assign their
    midpoints, widths, and weights using
    \refeq{eq:fibonacci-child-cell-weights}, and set
    $\kappa(C_{\mathrm L})\gets k+1$ and
    $\kappa(C_{\mathrm R})\gets k+1$.
\end{enumerate}
\end{enumerate}

\item \textit{Certify at level $K$.}
\begin{enumerate}
\item Set $s\gets M_K$ and choose $\mathcal R^\star$ using
\refeq{eq:fibonacci-heaviest-certification-run}.
\item While $s<M_{\mathrm{final}}^{\max}$ and
$p_\tau(\mathcal R^\star)<1-\delta_{\mathrm{final}}$:
\begin{enumerate}
\item Set $s\gets s+1$ and choose
\[
    \phi_{K,s}\in
    \operatorname*{arg\,max}_{\phi\in\Phi}
    \textsf{EIG}(\phi).
\]
\item Take a level-$K$ probe with angle $\phi_{K,s}$,
obtain outcome $b_{K,s}$, and append
$(F_K,\phi_{K,s},b_{K,s})$ to $\tau$.
\item Simultaneously update $p_\tau(C)$ for every
$C\in\mathcal C$ using
\refeq{eq:fibonacci-cell-weight-update}.
\item Recompute $\mathcal R^\star$ using
\refeq{eq:fibonacci-heaviest-certification-run}.
\end{enumerate}
Set $M_K\gets s$. If
$p_\tau(\mathcal R^\star)<1-\delta_{\mathrm{final}}$, return
\textsc{failure}.
\item Set $w$ using
\refeq{eq:fibonacci-certified-estimate} and return $w$.
\end{enumerate}
\end{enumerate}
\end{algorithm}

\vspace{3mm}

The protocol stores at most $B_{\max}$ weighted cells, together with the
ordered transcript and the creation point of each live cell. A posterior
update and EIG evaluation require work polynomial in $B_{\max}$ and
$|\Phi|$; a since-birth split replays at most the current transcript; and the
transition and certification tests can be performed by sorting the current
midpoint representatives. Consequently, if $B_{\max}$, $|\Phi|$, and both
measurement caps are polynomial in $n$, then the protocol uses polynomial
classical time and memory.

\subsection{Detailed Results of Fibonacci IQPE Simulations}

This appendix records the implementation choices and complete results for the
numerical experiments in
\refsec{sec:fibonacci-estimator-numerics}. The complete simulation code is available in our public GitHub repository,
\href{https://github.com/katherinevankirk/streaming-regev/}
{\textsf{streaming-regev}}. \reftab{tab:fibonacci-iqpe-parameters} gives all parameters. We test many values of $n$ and, in particular, the cryptographically relevant sizes $n = 1024,2048,4096$.
We use $1{,}000$ scalar trials at each $n$. The trials choose a $\theta$ to learn by choosing some length-$P$ bit string. As such, the true phase is 
sampled uniformly from
$\{0,1/2^P,\ldots,(2^P-1)/2^P\}$.
We choose $P$ such that $2^P
   \geq
   2^{24} B_{\max} F_K$ -- in other words, 
   the fixed-point representation retains at least \(24\) guard bits beyond the combined scales set by \(F_K\) and \(B_{\max}\).
We choose \(P\) so that, even after multiplication by the largest Fibonacci exponent \(F_K\), the fixed-point grid remains much finer than the nominal \(1/B_{\max}\) resolution of the posterior representation. Specifically, we set 
\begin{equation}
    P
    =
    \max\!\left\{
        16,
        \operatorname{bitlength}(F_K)
        +\left\lceil\log_2 B_{\max}\right\rceil
        +24
    \right\}.
    \label{eq:numerical-fixed-point-precision}
\end{equation}
The lower bound \(P\geq16\) prevents an unnecessarily coarse grid for small test instances.
Finally, also note that all quantities on the unit circle -- true $\theta$, cell midpoints, their Fibonacci multiples, etc -- are represented as integers.

For every trial we record the total number of scalar probes.  Since the
experiment visits the $K-1$ levels $2,\ldots,K$, we report the total number of probes divided by 
$(K-1)$ as the observed number of ``scalar probes per Fibonacci level.''
\reftab{tab:fibonacci-iqpe-complete-sweep} gives the complete results. The column $\overline S$ gives the mean total probe count.
The columns $N_{\mathrm{cert}}$, $N_{\mathrm{ac}}$, and
$N_{\mathrm{cap}}$ give, respectively, the number of certified trials, the
number that are both accurate and certified, and the number that fail to
certify before reaching the $2{,}000$-probe cap at the final Fibonacci
level, each out of $1{,}000$ trials.
The column $\overline S$ gives the mean total probe count.
We are interested in probe complexity scaling with the problem size. Traditional powers-of-2 IQPE is linear in terms of the levels visited (here, the max level visited is $K$), so we fit the mean total probe
count to the following: $\overline S
    =
    a(K-1)+b.$ 
The resulting fit has a strong $R^2$:
\begin{equation}
    \overline S
    =
    5.9372(K-1)+3.0467,
    \qquad
    R^2=0.999982.
    \label{eq:numerical-probe-fit}
\end{equation}
Note, however, that this quick regression uses the $21$ per-setting means, rather than treating the
$21{,}000$ individual trials as independent regression observations.

\begin{table*}[t]
    \centering
    \begin{tabular}{l l}
        \toprule
        Parameter & Value \\
        \midrule
        Trials per value of $n$ & $1000$ \\
        Regev scale & $R(n)=2^{3\sqrt n}$ \\
        Target accuracy & $\varepsilon_R=1/(\sqrt2R(n))$ \\
        Required final level
            & $K_{\mathrm{required}}=\min\{k\geq2:F_k>R/(4\sqrt2)\}$ \\
        Largest simulated level & $K=K_{\mathrm{required}}+3$ \\
        Initial number of cells & $B_0=32$ \\
        Candidate rotations
            & $\Phi=\{2\pi j/64:j=0,\ldots,63\}$ \\
        Transition parameter & $\delta_{\mathrm{move}}=10^{-3}$ \\
        Final certification parameter & $\delta_{\mathrm{final}}=10^{-2}$ \\
        Refinement threshold & $\eta_{\mathrm{grid}}=10^{-4}$ \\
        Maximum number of cells & $B_{\max}=30{,}000$ \\
        Guard-waiting shot cap per nonfinal level & $500$ \\
        Final-level shot cap & $2000$ \\
        Base pseudorandom seed & $20260807$ \\
        \bottomrule
    \end{tabular}
    \caption{Complete parameters for the numerical experiments of Adaptive Bayesian Fibonacci IQPE.}
    \label{tab:fibonacci-iqpe-parameters}
\end{table*}

\begin{table*}[p]
    \centering
    \setlength{\tabcolsep}{3.2pt}
    \renewcommand{\arraystretch}{1.04}
    \scriptsize
    \resizebox{\textwidth}{!}{%
    \begin{tabular}{r r r r c r r r c r}
        \toprule
        $n$ & $K$ & $\overline S$
        & Mean/level & Median/level [IQR]
        & $N_{\mathrm{cert}}$ & $N_{\mathrm{ac}}$ & $N_{\mathrm{cap}}$
        & Median cells [IQR] & Maximum cells \\
        \midrule
        $4$    & $10$  & $54.668$   & $6.074$ & $6.111\ [6.000,\,6.222]$ & $1000$ & $1000$ & $0$ & $77\ [75,\,79]$ & $92$ \\
        $6$    & $12$  & $68.703$   & $6.246$ & $6.000\ [5.909,\,6.182]$ & $999$  & $999$  & $1$ & $90\ [88,\,93]$ & $120$ \\
        $8$    & $14$  & $78.717$   & $6.055$ & $6.077\ [5.923,\,6.154]$ & $1000$ & $999$  & $0$ & $100\ [97,\,103]$ & $124$ \\
        $11$   & $16$  & $92.178$   & $6.145$ & $6.000\ [5.917,\,6.133]$ & $1000$ & $999$  & $0$ & $112\ [109,\,115]$ & $150$ \\
        $16$   & $19$  & $110.016$  & $6.112$ & $6.000\ [5.889,\,6.056]$ & $999$  & $998$  & $1$ & $128\ [125,\,132]$ & $147$ \\
        $23$   & $22$  & $125.674$  & $5.984$ & $6.000\ [5.905,\,6.048]$ & $1000$ & $1000$ & $0$ & $144\ [141,\,148]$ & $170$ \\
        $32$   & $26$  & $150.428$  & $6.017$ & $6.000\ [5.920,\,6.080]$ & $1000$ & $999$  & $0$ & $165\ [161,\,169]$ & $222$ \\
        $45$   & $31$  & $179.411$  & $5.980$ & $5.967\ [5.900,\,6.033]$ & $1000$ & $999$  & $0$ & $191\ [187,\,196]$ & $234$ \\
        $64$   & $36$  & $212.305$  & $6.066$ & $5.971\ [5.914,\,6.029]$ & $1000$ & $996$  & $0$ & $218\ [213,\,223]$ & $310$ \\
        $91$   & $43$  & $251.863$  & $5.997$ & $5.952\ [5.905,\,6.024]$ & $1000$ & $999$  & $0$ & $255\ [250,\,261]$ & $303$ \\
        $128$  & $50$  & $294.917$  & $6.019$ & $5.959\ [5.898,\,6.020]$ & $999$  & $998$  & $1$ & $292\ [287,\,299]$ & $470$ \\
        $181$  & $60$  & $354.873$  & $6.015$ & $5.949\ [5.898,\,6.000]$ & $999$  & $998$  & $1$ & $347\ [341,\,353]$ & $409$ \\
        $256$  & $71$  & $419.441$  & $5.992$ & $5.943\ [5.886,\,5.986]$ & $999$  & $998$  & $1$ & $407\ [401,\,415]$ & $468$ \\
        $362$  & $84$  & $498.588$  & $6.007$ & $5.928\ [5.892,\,5.976]$ & $998$  & $997$  & $2$ & $478\ [470,\,485]$ & $764$ \\
        $512$  & $99$  & $590.363$  & $6.024$ & $5.939\ [5.898,\,5.980]$ & $996$  & $993$  & $4$ & $556\ [549,\,565]$ & $822$ \\
        $724$  & $118$ & $695.470$  & $5.944$ & $5.923\ [5.889,\,5.957]$ & $999$  & $997$  & $1$ & $658\ [649,\,668]$ & $779$ \\
        $1024$ & $140$ & $827.210$  & $5.951$ & $5.921\ [5.892,\,5.957]$ & $998$  & $993$  & $2$ & $779\ [768,\,789]$ & $1145$ \\
        $1448$ & $166$ & $984.729$  & $5.968$ & $5.927\ [5.897,\,5.958]$ & $997$  & $993$  & $3$ & $920\ [909,\,932]$ & $1335$ \\
        $2048$ & $197$ & $1164.012$ & $5.939$ & $5.918\ [5.893,\,5.944]$ & $998$  & $996$  & $2$ & $1084.5\ [1072,\,1097]$ & $1429$ \\
        $2896$ & $234$ & $1386.979$ & $5.953$ & $5.918\ [5.893,\,5.944]$ & $998$  & $991$  & $2$ & $1284\ [1270,\,1298]$ & $1758$ \\
        $4096$ & $278$ & $1646.391$ & $5.944$ & $5.921\ [5.895,\,5.942]$ & $998$  & $990$  & $2$ & $1521\ [1505,\,1536]$ & $2018$ \\
        \bottomrule
    \end{tabular}%
    }
    \caption{
    Complete numerical results for the parameters specified in \ref{tab:fibonacci-iqpe-parameters}. The cell-count
    columns report the final number of live posterior cells; because cells
    are only bisected and never merged, this is also the peak cell count.
    }
    \label{tab:fibonacci-iqpe-complete-sweep}
\end{table*}

Next notice that no trial reaches the grid-saturation condition. 
At $n=4096$, final number of live cells has median is $1521$ and
its interquartile range is $[1505,1536]$.  The largest number of live cells
in any $n=4096$ trial was $2018$, which is only $6.73\%$ of our allowed grid refinement budget of $30{,}000$.  These cell counts measure the size
of the maintained posterior representation and provide a proxy for its
classical storage and expected-information-gain workload.  For the problem sizes we consider, we are well within what is classically plausible, and our simulations were all run locally on a laptop.

\paragraph{Error diagnostic.}
For every trial that certifies, let $w$ denote the estimate returned by
\refalg{alg:adaptive-bayesian-fibonacci-estimator}. We define its
normalized circular error by
\begin{equation}
    \operatorname{error}(w)
    \coloneqq
    \frac{
        \operatorname{dist}_{\RR/\ZZ}(w,\theta)
    }{
        \varepsilon_R
    }.
    \label{eq:normalized-fibonacci-estimation-error}
\end{equation}
A certified estimate is accurate precisely when
$\operatorname{error}(w)\leq1$. Trials that do not certify return
\textsc{failure} rather than an estimate $w$, so we report them separately
and do not assign them a value of $\operatorname{error}(w)$.

The following table partitions all $21{,}000$ trials according to their
certification status and, among the certified trials, according to their
final error. Among the $45$ false certifications, seven satisfy
$1<\operatorname{error}(w)\leq2$, one has
$\operatorname{error}(w)=8.36$, and the remaining $37$ satisfy
$\operatorname{error}(w)>10$. No observed false certification has error
strictly between $8.36$ and $77.1$. An additional $23$ trials fail to
certify before reaching the final-level probe cap.

\begin{center}
\begin{tabular}{lrr}
    \toprule
    Outcome & Count & Fraction of all trials \\
    \midrule
    Accurate and certified
        & $20{,}932$ & $99.676\%$ \\
    Not certified
        & $23$ & $0.110\%$ \\
    Falsely certified:
    $1<\operatorname{error}(w)\leq2$
        & $7$ & $0.033\%$ \\
    Falsely certified:
    $2<\operatorname{error}(w)\leq10$
        & $1$ & $0.005\%$ \\
    Falsely certified:
    $\operatorname{error}(w)>10$
        & $37$ & $0.176\%$ \\
    \bottomrule
\end{tabular}
\end{center}

Investigating the $45$ false certifications reveals two dominant, though
rare, failure modes.  Seven are \emph{near misses}: the returned estimate is
between one and two target radii from the true phase.  In these trials the
posterior concentrates near the correct phase, but the true phase lies just
outside the radius-$\varepsilon_R$ neighborhood of the returned midpoint.
Apart from one intermediate error of $8.36\varepsilon_R$, the remaining
$37$ false certifications have error greater than $10\varepsilon_R$.  We
refer to these large-error events as \emph{aliased misses}.

To see why aliasing can occur, recall that a level-$k$ likelihood depends on
a candidate phase $x$ only through $F_kx\bmod1$.  The map
$x\mapsto F_kx\bmod1$ can wrap the phase circle around itself many times.
Consequently, a wrapped phase $y$ has the branches
\begin{equation}
    x
    =
    \frac{y+j}{F_k}
    \pmod 1,
    \qquad
    j\in\{0,\ldots,F_k-1\},
    \label{eq:numerical-fibonacci-inverse-branches}
\end{equation}
which are separated by $1/F_k$ around the original phase circle.  Data from
the earlier levels and the transition guard are intended to identify the
branch $j$ containing the true phase before the estimator advances.  If the
posterior instead concentrates on another branch, subsequent probes can
continue to narrow that incorrect mode and eventually certify around it.

\begin{figure*}[t]
    \centering
    \includegraphics[width=0.72\textwidth]{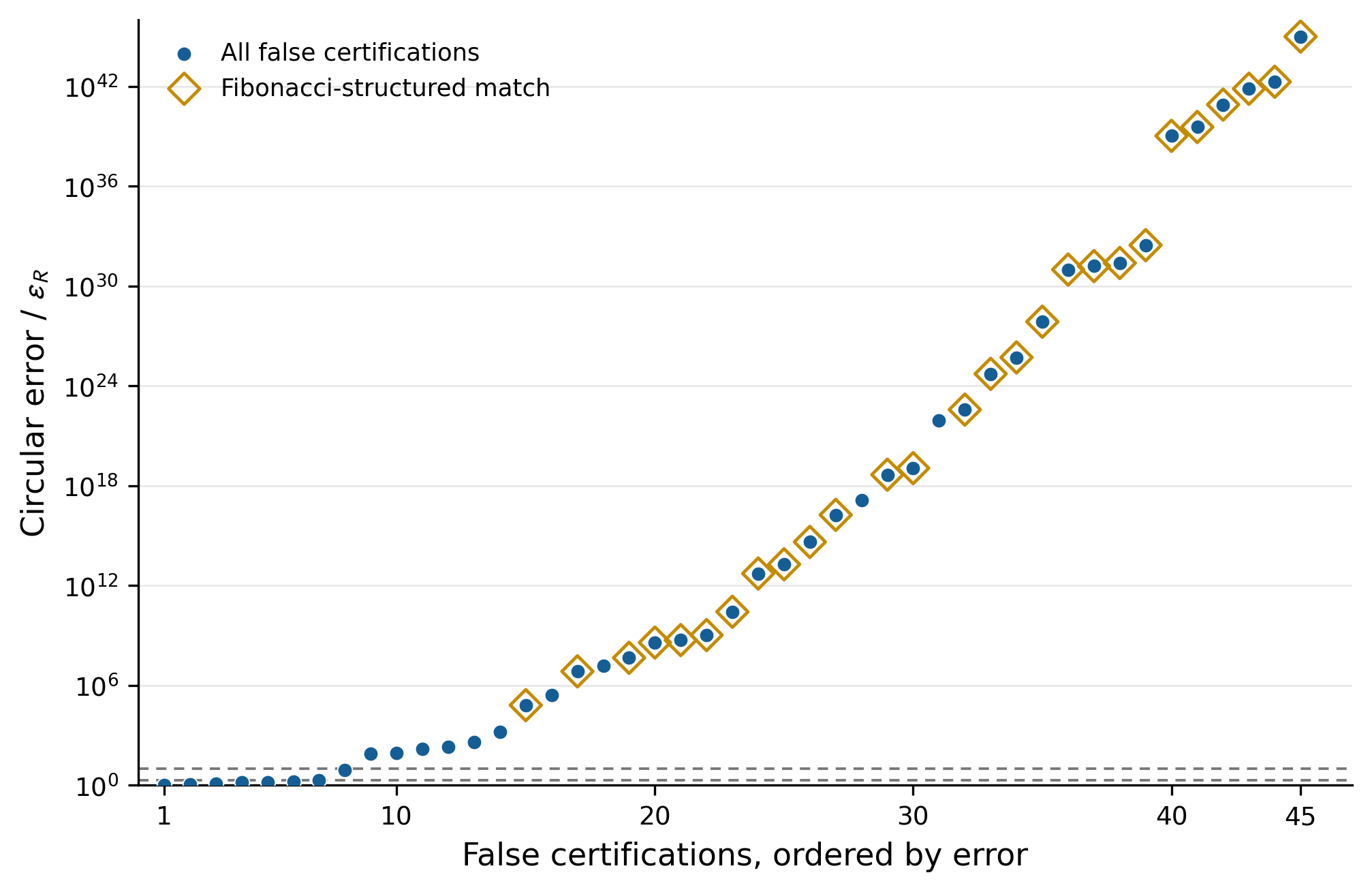}
    \caption{
    Circular error among the $45$ false certifications, ordered by magnitude
    and displayed in units of $\varepsilon_R$ on a logarithmic scale.  Seven
    estimates have error between $\varepsilon_R$ and $2\varepsilon_R$, one
    has error $8.36\varepsilon_R$, and the remaining $37$ have error greater
    than $10\varepsilon_R$.  The dashed lines at $2\varepsilon_R$ and
    $10\varepsilon_R$ are descriptive only and are not part of the protocol.
    The large-error events are consistent with selection of an incorrect
    inverse branch. 
    Hollow diamonds identify the events for which \(F_K\operatorname{dist}_{\mathbb R/\mathbb Z}(w,\theta)\) lies within relative error \(10^{-6}\) of \(F_j\) or \(2F_j\) for some \(j\). 
    }
    \label{fig:fibonacci-iqpe-error-diagnostic}
\end{figure*}

The Fibonacci exponents give these aliases additional structure.  
At the final level, every offset \(\ell/F_K\), for integer \(\ell\), is an exact inverse-branch alias. Among these aliases the choices \(\ell=F_j\) and \(\ell=2F_j\) form a distinguished subset: Fibonacci identities cause these offsets to remain close to inverse branches at several preceding levels.
Consider the offset $\Delta_j \coloneqq \frac{F_j}{F_K}$.
The Fibonacci identities give
\begin{equation}
    F_{K-r}\Delta_j
    =
    F_{j-r}
    +
    (-1)^{j-r}\frac{F_rF_{K-j}}{F_K},
    \qquad 0\leq r\leq j.
    \label{eq:numerical-fibonacci-alias-persistence}
\end{equation}
The first term is an integer, while the second remains small for the first
several levels preceding and following $K$.  The same offset can therefore remain close to
an inverse branch as the estimator advances.  We can investigate our aliases misses to see if any have this structure. In particular, we can consider the following diagnostic criterion
\begin{equation}
    \min_{\substack{2\leq j\leq K-2\\a\in\{F_j,2F_j\}}}
    \frac{|F_K\operatorname{dist}_{\RR/\ZZ}(w,\theta)-a|}{a}
    \leq
    10^{-6}.
    \label{eq:numerical-fibonacci-alias-diagnostic}
\end{equation}
It turns out that $27$ of our $37$ aliases misses satisfy this criterion: they exhibit a Fibonacci-structured branch displacement. In particular,
$22$ match $F_j$ and five match $2F_j$.  That said, while this alignment supports the inverse-branch interpretation, it remains a
retrospective diagnostic. Our terminal outputs do not record the
posterior trajectory and cannot identify the level at which the true branch
was lost or prove the mechanism that caused any particular trial.
\reffig{fig:fibonacci-iqpe-error-diagnostic} displays the $45$ false
certifications, ordered by error and labeled with whether they have a Fibonacci-structured branch displacement.

\end{document}